\documentclass[11pt,letter]{article}
\usepackage{amsmath}
\usepackage{xspace}
\usepackage{todonotes}
\usepackage{amsthm}
\usepackage{fullpage}

\usepackage[utf8]{inputenc}
\usepackage{indentfirst,verbatim} 
\usepackage{listings} 
\usepackage[normalem]{ulem}
\usepackage{soul}
\usepackage{amsfonts}
\usepackage{amsthm}
\usepackage{enumitem}
\usepackage{hyperref}

\usepackage{amsmath,amsfonts,amsthm}
\usepackage{mathbbol}
\usepackage{amsthm}
\usepackage{graphicx,color}
\usepackage{boxedminipage}
\usepackage[linesnumbered, boxed]{algorithm2e}
\usepackage{framed}
\usepackage{thmtools}
\usepackage{thm-restate}
\usepackage{xspace}
\usepackage{todonotes}
\usepackage{pgfplots}
\usetikzlibrary{calc}
\usetikzlibrary{shapes}
\usetikzlibrary{arrows.meta}
\usepackage{colortbl}
\usepackage{tcolorbox} 

\usepackage{xifthen}
\usepackage{tabularx}
\usepackage{capt-of}
\usepackage{thmtools}
\usepackage{thm-restate}
\usepackage{subcaption}
\usepackage{booktabs}
\usepackage{calc}
\usepackage{cleveref}

\usepackage[margin=1in]{geometry} 

\theoremstyle{plain}
\newtheorem{theorem}{Theorem}

\newtheorem{lemma}{Lemma} 
\newtheorem{claim}{Claim} 
\newtheorem{proposition}{Proposition} 
\theoremstyle{definition}
\newtheorem{definition}{Definition}

\newcommand{\cqed}{\renewcommand{\qedsymbol}{$\lrcorner$}\qed}
\newcommand{\pname}{\textsc}
\newcommand{\polyn}{n^{\Oh(1)}}
\newenvironment{claimproof}{\medskip\noindent \emph{Proof of Claim~\theclaim.}  }{\hfill\cqed\medskip}

\newlength{\RoundedBoxWidth}
\newsavebox{\GrayRoundedBox}
\newenvironment{GrayBox}[1]%
{\setlength{\RoundedBoxWidth}{.93\textwidth}
	\def\boxheading{#1}
	\begin{lrbox}{\GrayRoundedBox}
		\begin{minipage}{\RoundedBoxWidth}}%
		{   \end{minipage}
	\end{lrbox}
	\begin{center}
		\begin{tikzpicture}%
			\node(Text)[draw=black!20,fill=white,rounded corners,%
			inner sep=2ex,text width=\RoundedBoxWidth]%
			{\usebox{\GrayRoundedBox}};
			\coordinate(x) at (current bounding box.north west);
			\node [draw=white,rectangle,inner sep=3pt,anchor=north west,fill=white] 
			at ($(x)+(6pt,.75em)$) {\boxheading};
		\end{tikzpicture}
\end{center}}     

\newenvironment{defproblemx}[2][]{\noindent\ignorespaces%
	\FrameSep=6pt%
	\parindent=0pt%
	\vspace*{-1.5em}
	\ifthenelse{\isempty{#1}}{%
		\begin{GrayBox}{#2}%
		}{%
			\begin{GrayBox}{#2 parameterized by~{#1}}%
			}
			\begin{tabular*}{\textwidth}{@{\hspace{.1em}} >{\itshape} p{1.8cm} p{0.8\textwidth} @{}}%
			}{
			\end{tabular*}%
		\end{GrayBox}%
		\ignorespacesafterend
	}

	\newcommand{\defproblema}[3]{
		\begin{defproblemx}{#1}
			Input:  & #2 \\
			Task: & #3
		\end{defproblemx}
	}%
	
	\DeclareMathOperator{\ndef}{\nu_{\it k}}
	
	\newcommand{\Oh}{\mathcal{O}}

	\newcommand{\probSTIAG}{\pname{Tree Containment  Above Minimum Degree}\xspace}
	\newcommand{\probSTI}{\pname{Tree Containment}\xspace}
	\newcommand{\probHitIso}{\pname{Annotated Hitting Subtree Containment}\xspace}

	\newcommand{\separable}[1]{size-#1-separable}
	
	\DeclareMathOperator{\Ima}{Im}
	\DeclareMathOperator{\ld}{ld}
	
	\DeclareMathOperator{\diam}{diam}

	\DeclareMathOperator{\operatorClassNP}{{\sf NP}}
	\newcommand{\classNP}{\ensuremath{\operatorClassNP}}

	\DeclareMathOperator{\operatorClassFPT}{{\sf FPT}\xspace}
	\newcommand{\classFPT}{\ensuremath{\operatorClassFPT}\xspace}

	\DeclareMathOperator{\operatorClassXP}{{\sf X}P\xspace}
	\newcommand{\classXP}{\ensuremath{\operatorClassXP}\xspace}
	\lstnewenvironment{code}{\lstset{language=Haskell,basicstyle=\small}}{}
	
\title{\textsc{Tree Containment Above Minimum Degree is FPT}\thanks{The research leading to these results has been
supported by the Research Council of Norway via the project BWCA (grant no. 314528) and DFG Research Group ADYN via grant DFG 411362735.}}
	\author{
		Fedor V. Fomin\thanks{
			Department of Informatics, University of Bergen, Norway.}\\fedor.fomin@uib.no
		\and
		Petr A. Golovach\addtocounter{footnote}{-1}\footnotemark{}\\petr.golovach@ii.uib.no
		\and
		Danil Sagunov\thanks{
			St.\ Petersburg Department of V.A.\ Steklov Institute of Mathematics, Russia.
		}\\danilka.pro@gmail.com
		\and 
		Kirill Simonov\thanks{Hasso Plattner Institute, University of Potsdam, Germany.}\\kirillsimonov@gmail.com
	}
	\date{}
	
	\usetikzlibrary{arrows}
	
\begin{document}
		
		\maketitle	
 \begin{abstract}
According to the classic  Chv{\'{a}}tal's Lemma from 1977, a graph of minimum degree $\delta(G)$ contains every tree on $\delta(G)+1$ vertices. 
 Our main result is the following algorithmic ``extension'' of Chv{\'{a}}tal's Lemma: For any $n$-vertex graph $G$, integer $k$, and a tree $T$ on at most $\delta(G)+k$ vertices, deciding whether $G$ contains a subgraph isomorphic to $T$,   can be done in time $f(k)\cdot n^{\Oh(1)}$ for some function $f$ of $k$ only.

The proof of our main result is based on an interplay between extremal graph theory and parameterized algorithms. 
 \end{abstract}		

\newpage 
\tableofcontents

\newpage

 \section{Introduction}
In the  \probSTI problem we are given an $n$-vertex graph $G$ and a tree $T$. The task is to identify whether $G$ has a subgraph isomorphic to $T$.\footnote{Let us remark that in computational biology the name tree containment is used for a different problem of deciding   whether a phylogenetic network displays a phylogenetic tree over the same set of labeled leaves.} For the very special case of $T$ being an $n$-vertex path, solving \probSTI is equivalent to deciding whether $G$ contains a Hamiltonian path and thus is \classNP-complete. Our work on \probSTI is strongly motivated by the recent advances in algorithmic ``extensions'' of the classic theorems of extremal combinatorics 
\cite{FominGSS22,FominGSS22esa,HanK20}. 

For example, the classic theorem of Dirac states that every 2-connected graph contains a cycle (and thus a path) of length at least $\min\{2\delta(G),n\}$, where $\delta(G)$ is the minimum degree of $G$.  In~\cite{FominGSS22}, we gave an \classFPT algorithm for parameterization ``above Dirac's bound''---an algorithm that for any $k\geq 1$,  decides whether a connected $G$ contains a path of length at least $2\delta(G)+k$ in time $f(k)\cdot n^{\Oh(1)}$ for some function   $f$ of $k$ only.  

The question of how to impose conditions on vertex degrees of the host graph $G$ to guarantee that it contains a certain tree $T$ as a subgraph is a fundamental question in extremal graph theory. However, compared with path and cycle containments, tree containment is much more challenging. For example, the theorem of Erd{\H{o}}s and Gallai~\cite{ErdosG59} from 1959 asserts that every graph of average degree $>d$ contains a cycle with at least $d+1$ vertices. Similarly,  Erd\H{o}s and Sós \cite{Erdos64} conjectured in 1963
that every graph with average degree $>d$ contains any tree on $d+1$ vertices. This conjecture remains open. 

The starting point of our algorithmic study of \probSTI is the following cute result first published by 
Chv{\'{a}}tal.
\begin{lemma}[Chv{\'{a}}tal's Lemma \cite{Chvatal77}]\label{chvatal-theorem}
	If $G$ is a graph of minimum degree $\delta(G)$, then $G$ contains every tree on $\delta(G)+1$ vertices.
\end{lemma}
From the combinatorial point of view the result of \Cref{chvatal-theorem} is tight: a $\delta$-regular graph does not contain a star of degree $\delta(G)+1$. The proof of \Cref{chvatal-theorem} is constructive and it yields a polynomial time algorithm computing a subtree in $G$ isomorphic to a tree $T$ on $\delta(G)+1$ vertices. Whether \Cref{chvatal-theorem} is tight from the algorithmic point of view, that is, whether it is possible to decide in polynomial time if a tree on $\delta(G)+k$ vertices, for some fixed constant $k>1$, is in $G$, was open prior to our work. Our main result is the following ``algorithmic extension'' of  Chv{\'{a}}tal's Lemma.

\begin{restatable}{theorem}{maintheorem}
	\label{thm:treecontainment-fpt}
	For any $n$-vertex graph $G$, integer $k$, and a tree $T$ on at most $\delta(G)+k$ vertices, there is a randomized algorithm deciding with probability at least $\frac{1}{2}$ whether $G$ contains a subgraph isomorphic to $T$ in time $2^{k^{\Oh(1)}}\cdot\polyn$.
	The algorithm is with one-sided error and reports no false-positives.
	
\end{restatable}

In other words,  \probSTIAG  admits a randomized  \classFPT algorithm.  We state \Cref{thm:treecontainment-fpt} for the decision variant of the problem. However, the proof of the theorem is constructive and if it exists, the corresponding subgraph isomorphism can also be constructed in the same running time.

It is useful to compare and contrast  \Cref{thm:treecontainment-fpt} and the algorithm ``above Dirac'' from \cite{FominGSS22} that decides whether a connected graph contains a path of length at least $2\delta(G)+k$ in time $f(k)\cdot n^{\Oh(1)}$. On the one hand,  the statement of \Cref{thm:treecontainment-fpt} holds for any tree, not only paths. On the other hand, the ``combinatorial threshold'' in the ``above Dirac'' algorithm is $2\delta(G)$ and in   \Cref{thm:treecontainment-fpt} it is $\delta(G)$. While in the statement of   Chv{\'{a}}tal's Lemma the value $\delta(G)$ cannot be replaced by $(1+\varepsilon)\delta(G)$ for $\varepsilon>0$, it is not clear a priory that the threshold  $\delta(G)$ in \Cref{thm:treecontainment-fpt} cannot be increased. Our next theorem rules out this option. 

\begin{restatable}{theorem}{lowerboundtheorem}
	\label{thm:lower-bound}
	For any $\varepsilon>0$, \probSTI is \classNP-complete when restricted to instances $(G,T)$  
	with $|V(T)|\leq (1+\varepsilon)\delta(G)$.
\end{restatable}

\paragraph{Related Work}   \probSTI plays an important role both in graph theory and in graph algorithms. 

\medskip\noindent\textsl{Extremal Graph Theory.} According to Maya Stein~\cite{Stein20}:  \emph{``One of the most intriguing open questions in the area is to determine degree conditions a graph $G$ has to satisfy in order to ensure it contains a fixed tree $T$, or more generally, all trees of a fixed size.''}  While the conjecture of Erd\H{o}s and Sós \cite{Erdos64}  about the average degree remains open, various other conditions have been suggested that might ensure the appearance of all trees or forests of some fixed size \cite{Babu05,Brandt94,MR4075918,MR1356577,Lucky13,ErdosFLS95}. We refer to the survey of Stein   \cite{Stein20} for a comprehensive overview of the area. 

Our work is also closely related to \emph{stability theorems} in extremal combinatorics. Informally, a stability theorem establishes that an ``almost nice'' extremal structure can always be obtained by slightly modifying a ``nice structure''. For example, coming back to Dirac's and Erd{\H{o}}s-Gallai theorems, there is a significant amount of literature devoted to sharper versions of these classic results \cite{MR3548292,MR3777040,MR4078813,MR4140611,MR0463030,zhu2022stability}.  The typical statement of such results is that when we weaken the condition on the minimum vertex degree or an average degree of a graph, the graph contains a long cycle (or path) unless it possesses a very specific structure.  To prove  \Cref{thm:treecontainment-fpt}, we have to establish several stability variants of  Chv{\'{a}}tal's Lemma.

\medskip\noindent\textsl{Algorithms.} \probSTI is the special case of \textsc{Subgraph Isomorphism}, where the guest graph $T$ is a tree. Matula in \cite{matula1978subtree} gave a polynomial time algorithm for \probSTI when the host graph $G$ is also a tree. 
According to 
Matou{\v{s}}ek  and Thomas  \cite{MatousekT92}, \probSTI is  \classNP-complete when all vertices of $T$  but one are of degree $\leq 3$ and $G$ is a treewidth $2$ graph and all vertices of $G$ but one are of degree $\leq 3$. The result of Matou{\v{s}}ek  and Thomas shows a sharp difference in the complexity of  \probSTI  and the \textsc{Longest Path}, which is \classFPT parameterized by the treewidth of $G$. The exhaustive study of  \textsc{Subgraph Isomorphism} by Marx and  Pilipczuk \cite{MarxP13} establishes several hardness results about \probSTI for different classes of graphs $G$ and trees $T$. There is a broad literature in graph algorithms on a related problem of finding a spanning tree in a graph with specified properties, see e.g.  \cite{papadimitriou1982complexity,furer1992approximating,goemans2006minimum}.  

The seminal work of Alon, Yuster,  and Zwick on color coding  \cite{AlonYZ95}  shows that \probSTI is \classFPT parameterized by the size of $T$. In other words, deciding whether $G$ contains a tree $T$ of size $t$ could be done in time $2^{\Oh(t)} n^{\Oh(1)}$.
Let us remark that in the setting of \Cref{thm:treecontainment-fpt}, the  color coding method provides an algorithm of running time 
$2^{\Oh(\delta(G)+k)} n^{\Oh(1)}$, which is not  \classFPT in $k$. 

Several results in the literature provide \classFPT algorithms for long paths, and cycles parameterized above some degree conditions. Our work is an extension of this line of research to more general subgraph isomorphism problems.  
Fomin, Golovach, Lokshtanov, Panolan, Saurabh, and Zehavi~\cite{FominGLPSZ20}  gave  \classFPT algorithms for computing long cycles and paths above the degeneracy of a graph. The tractability of these problems was extended by the authors in~\cite{FominGSS22esa} above the so-called Erd{\H{o}}s-Gallai bound, which is above the average vertex degree 
of a graph. In \cite{FominGSS21,FominGSS22}, we established that finding a cycle above Dirac's bound
is  \classFPT. In other words, we gave an algorithm of running time $2^{\Oh (k)} \cdot n^{\Oh (1)}$ deciding whether a 
$2$-connected graph $G$ contains a cycle of length at least $\min\{2\delta(G), n\}+k$. The ideas and methods used to prove all the above results about cycles and paths are quite different from those we use in the proof of   \Cref{thm:treecontainment-fpt}.

From a more general perspective, \Cref{thm:treecontainment-fpt} belongs to a rich subfield of  Parameterized Complexity concerning parameterization above/below specified guarantees ~\cite{AlonGKSY10,CrowstonJMPRS13,GargP16,GutinIMY12,LokshtanovNRRS14,MahajanRS09,DBLP:conf/wg/Jansen0N19}. We refer to the recent survey of Gutin and Mnich \cite{GutinMnich22} for an overview of this area.
In particular, the parameterized complexity of finding an $(s,t)$-path above the distance between vertices $s$ and $t$, the \textsc{Detour} problem, attracted significant attention recently \cite{BezakovaCDF17,FominGLPSZ21,FominGLSS022,HatzelMPS23,JacobWZ23}.


\section{Definitions and preliminaries}
For a positive integer $t$, we define $[t]=\{1,\ldots,t\}$.

We use standard graph-theoretic notation and refer to the textbook of Diestel~\cite{Diestel} for non-defined notions. We consider only finite simple undirected graphs. We use $V(G)$ and $E(G)$ to denote the sets of vertices and edges of a graph $G$, respectively; $n$ and $m$ are used to denote the number of vertices and edges  if this does not create confusion. A vertex $v$ is a \emph{non-neighbor} of $u$ if $v\neq u$ and $uv\notin E(G)$.
For a  graph $G$ and a subset $X\subseteq V(G)$ of vertices, we write $G[X]$ to denote the subgraph of $G$ induced by $X$.
We use $G-X$ to denote the  graph obtained by deleting the vertices of $X$, that is, $G-X=G[V(G)\setminus X]$; we write $G-v$ instead of $G-\{v\}$ for a single element set.
For a vertex $v$, $N_G(v)=\{u\in V(G)\mid vu\in E(G)\}$ is the \emph{open neighborhood} of $v$ and $N_G[v]=N_G(v)\cup\{v\}$ is the \emph{closed neighborhood}. For a set of vertices $X$, 
$N_G(X)=\big(\bigcup_{v\in X}N_G(v)\big)\setminus X$ and  $N_G[X]=\bigcup_{v\in X}N_G[v]$.  We use $\deg_G(v)=|N_G(v)|$ to denote the \emph{degree} of a vertex $v$; $\Delta(G)=\max_{v\in V(G)}\deg_G(v)$ is the \emph{maximum degree} of $G$ and $\delta(G)=\min_{v\in V(G)}\deg_G(v)$ is the \emph{minimum degree}. In the above notation, we may omit subscripts denoting graphs if this does not create confusion. 
We write $P=v_1-\dots -v_k$ to denote a (simple) path in a graph $G$ with $k$ vertices $v_1,\ldots,v_k$ of length $k-1$; $v_1$ and $v_k$ are the \emph{end-vertices} of $G$ and we say that $P$ is an $(v_1,v_k)$-path. 
The \emph{diameter} of $G$, denoted  $\diam(G)$, is the maximum length of a shortest $(u,v)$-path in $G$ over all $u,v\in V(G)$. Two vertices $u$ and $v$ compose a \emph{diametral} pair if the distance between them, i.e. the length of the shortest path, is $\diam(G)$.  

An \emph{isomorphism} of a graph $H$ into a graph $G$, a bijective mapping $\varphi\colon V(H)\rightarrow V(G)$ such that $uv\in E(H)$ for $u,v\in V(H)$ if and only if $\varphi(u)\varphi(v)\in E(G)$. A \emph{subgraph isomorphism} of $H$ into $G$ is an injective mapping  $\sigma\colon V(H)\rightarrow V(G)$ such that $uv\in E(H)$ for $u,v\in V(H)$ if and only if $\varphi(u)\varphi(v)\in E(G)$. In words, this means that $G$ contains $H$ as a subgraph. We use $\Ima\sigma$ to denote $\sigma(V(H))$.


Throughout our paper, we use the following specific notions. 

\begin{definition}[Maximum leaf-degree $\ld(T)$]\label{def:nd}
	The \emph{leaf-degree} of a vertex $v$ in $T$ is  the number of leaves of $T$ that are neighbors of $v$. The \emph{maximum leaf-degree} of $T$, $\ld(T)$,  is the maximum of the leaf-degrees over all vertices of $T$.
\end{definition}

\begin{definition}[Neighbor deficiency $\ndef(v)$]\label{def:ndef}
	For a graph $G$, an integer $k\ge 0$ and a vertex $v \in V(G)$ we define the \emph{neighbor deficiency} of $v$ as $$\ndef(v)=\max\{(\delta(G)+k-1)-\deg_G(v),0\}.$$ 
	If $\ndef(v)=0$, we say that $v$ is \emph{non-deficient}.
\end{definition}

\begin{definition}[$q$-escape vertex]\label{def:escape}
	For a graph $G$ and integer $q$, a vertex $v$ in $G$ is an \emph{$q$-escape vertex}, if $\deg_G(v)\ge \delta(G)+q$ or the maximum matching size between $N[v]$ and $V(G)\setminus N[v]$ in $G$ is at least $q$.
\end{definition}

\begin{definition}[\separable{$q$}]\label{def:separable}
	We say that a tree $T$ is \emph{\separable{$q$}} if there is an edge in $T$ whose removal separates $T$ into two subtrees consisting of at least $q$ vertices each.
\end{definition}

The following simple lemma about relation between the size of the tree, its diameter, and the number of leaves will be useful.

\begin{lemma}[folklore]
    Let $T$ be a non-empty tree with $|V(T)| \ge q \cdot \diam(T)$, for some integer $q$. Then $T$ has at least $q$ leaves.
    \label{lemma:leaves_diameter}
\end{lemma}
\begin{proof}
    We show the statement by induction on $q$.
    For $q = 1$, any non-empty tree has at least one leaf.
    Assume the statement already holds for $q$, and consider a tree $T$ with $|V(T)| \ge (q + 1) \cdot \diam(T)$, we argue that it has at least $(q + 1)$ leaves. Let $\ell$ be an arbitrary leaf of $T$, and let $P$ be a maximal path in $T$ that starts with $\ell$ and contains only degree-two vertices of $T$ as internal vertices; let $v$ be its other endpoint.
    Now consider the tree $T'$ obtained by removing the vertices of $V(P) \setminus \{v\}$ from $T$. Clearly, $|V(P) \setminus \{v\}| \le \diam(T)$, thus $|V(T')| \ge |V(T)| - \diam(T) \ge q \cdot \diam(T')$. By induction, $T'$ has at least $q$ leaves. Observe that $v$ is either the only vertex of $T'$, or not a leaf of $T'$, otherwise the path $P$ is not maximal. Therefore all leaves of $T'$ remain leaves in $T$, and $T$ has at least $q + 1$ leaves, including $\ell$.
\end{proof}

Finally, we give here the following extension of  Chv\'{a}tal's Lemma (\Cref{chvatal-theorem}).  While Chv\'{a}tal's Lemma
indeed provides the guarantee of $\delta(G)+1$ for \probSTIAG, throughout the proof of the main theorem, we often need a more general statement. 
Its proof repeats the original proof of Chv\'{a}tal, and we provide it here for completeness.

\begin{proposition}[\cite{Chvatal77}]\label{prop:chvatal-generalized}
	Let $G$ be a graph and let $T$ be a tree on at most $\delta(G)+1$ vertices.
	Let $\sigma': V(T') \to V(G)$ be a subgraph isomorphism mapping a connected subtree $T'$ of $T$ into $G$.
	Then there is a subgraph isomorphism $\sigma:V(T)\to V(G)$ mapping $T$ into $G$ such that $\sigma$ is an extension of $\sigma'$.
\end{proposition}
\begin{proof}[Proof of \Cref{prop:chvatal-generalized}]
	The proof is constructive.
Let $T_0, T_1, \ldots, T_q$ be a sequence of trees such that $q=|V(T)|-|V(T')|$, $T_0=T'$, $T_q=T$,  and for each $i \in [q]$, $T_{i-1}=T_{i}-\ell_i$, where  $\ell_i$ is a leaf of $T_i$ that is not present in $T'$.
One way to construct the sequence in reverse order is to start from $T_q:=T$.
Then, to obtain $T_{i-1}$ from $T_{i}$ for each consecutive $i$ in $\{q,q-1,\ldots, 1\}$, we just delete a leaf of $T_i$ that does not belong to $T'$. Since $T'$ is a connected subtree of $T_i$, 
such a leaf always exists in $T_i$.

	We then construct a series of subgraph isomorphisms $\sigma_0, \sigma_1, \ldots, \sigma_q$, such that for each $i \in \{0,\ldots,q\}$, $\sigma_i$ is a subgraph isomorphism of $T_i$ into $G$.
	We start from $\sigma_0=\sigma'$. 
	Then consecutively for each $i\in [q]$, $\sigma_i$ is obtained by extending $\sigma_{i-1}$ on $\ell_i$.
	The leaf $\ell_i$ has only one neighbor $s_i$ in $T_i$.
	Then the image of $\ell_i$ in $\sigma_i$ should be the neighbor of $\sigma_{i-1}(s_i)$.
	Since $|\Ima\sigma_{i-1}|\le\delta(G)$, and $\sigma_{i-1}(s_i) \in \Ima\sigma_{i-1}$,  
	$\Ima \sigma_{i-1}$ contains 
	at most $\delta(G)-1$ neighbors of $\sigma_{i-1}(s_i)$ in $G$.
	Then at least one neighbor of the image of $s_i$ is not occupied by $\sigma_{i-1}$. We obtain $\sigma_i$ from $\sigma_{i-1}$ by extending mapping  $\ell_i$ to such a neighbor.
	
	The procedure produces the subgraph isomorphism $\sigma_q$ of $T_q=T$ into $G$.
\end{proof}

 
 \section{Main ideas and structure of the proof of \Cref{thm:treecontainment-fpt}}
 
 In this overview, we will provide some intuition on how several various structural cases for $G$ and $T$ guarantee that $G$ contains $T$. In several situations, we push the structural analysis to the limit. In the remaining cases, when the structural analyses (or stability theorems) cannot be pushed further, the obtained structural properties allow to design algorithms.  Such a WIN/WIN approach 
 results is a randomized \classFPT algorithm running in time $2^{k^{\Oh(1)}}\cdot\polyn$.
The order of presentation of the case analysis in the overview slightly disagrees with the order of the section in the main part of the paper.
In the overview, we aimed to make the presentation more intuitive while in the main part of the paper, we put parts using the same techniques and ideas closer to each other.
The summary of the case analysis in the main part is given in \Cref{section:mainresult}.

 The natural way to proceed above the Chv{\'{a}}tal's  $\delta(G)+1$ guarantee is to ask whether $G$ contains an arbitrary $T$ of size at most $\delta(G)+2$.
The answer to this question appears to be quite simple, yet it settles the starting point of our work.
We state it as \Cref{prop:first_step}. 
While this result is not explicitly used in the proof of the main theorem,  its proof exposes the key ideas we use to prove the main theorem.

\begin{proposition}\label{prop:first_step}
	Every connected graph $G$ contains every tree on at most $\min\{|V(G)|,\delta(G)+2\}$ vertices, unless $G$ is $\delta(G)$-regular and $T$ is isomorphic to the star graph $K_{1,\delta(G)+1}$ with $\delta(G)+1$ leaves.
\end{proposition}

\begin{proof}[Proof of \Cref{prop:first_step}]
	Note that the only case that has to be proved above Chv\'{a}tal's Lemma is when $|V(T)|=\delta(G)+2$.
	
	We have that $G$ is connected and has at least $\delta(G)+2$ vertices.
	Let $T$ be an arbitrary tree on exactly $\delta(G)+2$ vertices.
	Clearly, if $T$ has a vertex of degree $\delta(G)+1$ but the maximum degree of $G$ equals $\delta(G)$  then $G$ does not contain $T$ as a subgraph.
	This is equivalent to $T$ being isomorphic to $K_{1,\delta(G)+1}$ and $G$ being $\delta(G)$-regular.
	It is left to show that in any other case, $G$ contains $T$.
	
	Assume first that  $G$ has a vertex of degree at least $\delta(G)+1$.
	Denote this vertex by $u$.
	Since $T$ has at least two vertices, there is a vertex in $T$ adjacent to a leaf.
	Denote this vertex by $t$ and its adjacent leaf by $\ell$.
	We start constructing a mapping of $T$ into $G$ by mapping $t$ to $u$.
	Since $T-\ell$ is a tree consisting of exactly $\delta(G)+1$ vertices, we apply \Cref{prop:chvatal-generalized} and extend mapping $t\to u$ to a subgraph isomorphism $\sigma$ of $T-\ell$ into $G$.
	The size of $N_G[u]$ is at least $\delta(G)+2$. Hence, at least one vertex in $N_G[u]$ is not used by $\sigma$.
	We extend $\sigma$ by mapping $\ell$ to this vertex.
	Then $\sigma$ becomes a subgraph isomorphism of $T$ into $G$.
	That is, if $G$ has a vertex of degree at least $\delta(G)+1$ then $G$ contains $T$.

	The remaining case is when all vertices of $G$ are of degree $\delta(G)$ but $T$ is not isomorphic to a star.
	Again, we take a leaf $\ell$ and its neighbor vertex $t$ in $T$.
	Since $T$ is not a star, there is a neighbor of $t$ in $T$ that is not a leaf.
	Denote this neighbor by $x$. Then $x$ has at least one neighbor distinct from $t$, denote it by $y$.
	
	We have a path on three vertices $t-x-y$ in the tree $T-\ell$ and we map this path into $G$ as follows.
	Take an arbitrary vertex $u$ in $G$.
	Since $G$ is connected and $|V(G)|>\delta(G)+1=|N_G[u]|$, there is at least one edge $vw\in E(G)$ such that $v\in N_G[u]$ and $w\in V(G)\setminus N_G[u]$.
	We have that $u-v-w$ is a path in $G$ where $u$ is not adjacent to $w$.
	We initiate the construction of a subgraph isomorphism of $T$ into $G$ by mapping $t,x,y$ into $u,v,w$ respectively.
	Then, by making use of  \Cref{prop:chvatal-generalized},  we extend this mapping into a subgraph isomorphism $\sigma$ of $T-\ell$ into $G$.
	
	In contrast to the previous case, the closed neighborhood of $u$, $N_G[u]$ is of size $\delta(G)+1$.
	But we ensured that $\sigma$ uses at least one vertex outside $N_G[u]$ by mapping $y$ to $w$ initially.
	Hence, $N_G[u]\setminus\Ima\sigma $ is not empty.
	Therefore, we can  extend $\sigma$ by mapping $\ell$ to an arbitrary vertex in $N_G[u]\setminus\Ima\sigma$.
	The obtained mapping is a subgraph isomorphism of $T$ into $G$.
	The proof of the proposition is complete.
\end{proof}

Let us highlight and discuss the key ideas of the proof above.

\medskip\noindent\textbf{Idea \textrm{I}.}
\textit{Saving space for mapping leaves starts from mapping their neighbors.}

\noindent In the proof of \Cref{prop:first_step},  we cut off a leaf $\ell$ of $T$ and decide where it will be mapped later.
Since the pruned tree has exactly $\delta(G)+1$ vertices, we can map it into $G$ by \Cref{chvatal-theorem}. 
However, we have to make sure that the removed leaf can be mapped to a free vertex. To achieve that, we initially set the image of its neighbor $t$ to a specific vertex in $G$.
If $G$ has a vertex of degree more than $\delta(G)$ then any such vertex will do as an image of $t$.
This is because the isomorphism can occupy at most $\delta(G)$ of its neighbors, so one of them is always vacant to host $\ell$.
The other case occurs when $G$ is regular and hence has no vertex of a large degree.
This brings us to the next key idea.

\medskip\noindent\textbf{Idea II.}
\textit{Saving space for leaves is achieved by mapping outside specific sets.}

\noindent
When we deal with the case of a regular graph $G$, we map $t$ to an arbitrary vertex $u$ of degree $\delta(G)$.
Now a subgraph isomorphism of $T-\ell$ into $G$ can occupy all $\delta(G)$ neighbors of $u$. If this happens, then
the isomorphism uses no vertex outside $N_G[u]$ since already 
$\delta(G)+1$ vertices are occupied. 
Thus, if we force the subgraph isomorphism to use at least one vertex outside $N_G[u]$, then at least one vertex in $N_G[u]$ would be saved.
This is exactly what we do in the proof of \Cref{prop:first_step}.
We initially map a path of three vertices in $T$ starting in $t$ into a path of three vertices in $G$ starting in $u$ making 
sure that the last vertex of this path in $G$ is not a neighbor of $u$.
When this isomorphism is extended onto $T-\ell$, it automatically leaves at least one neighbor of $u$ in $G$ unused.
This neighbor finally becomes the image of $\ell$.

These two ideas bring to a polynomial time algorithm for finding trees of size $\delta(G)+2$ in a graph $G$.
It appears that we can push the applicability of these ideas further.
However, it does not come without additional effort. Let us provide some intuition on how the two ideas could be extended to the case when  $T$ is a tree on $\delta(G)+k$ vertices for $k\ge 3$.

\subsection{Saving neighbors of a single vertex}

Let $T$ be a tree on $\delta(G)+k$ vertices for $k\ge 3$.
The first question arising when we try to adapt \Cref{prop:first_step} to $T$ is the following: If there is a vertex $t$ in $T$ with at least $k-1$ adjacent leaves, can we shave off $k-1$ leaves $t$ and repeat the same arguments to map the shaved tree and save $k-1$ neighbors of the image of $t$?
This would allow us to map the shaved leaves into the saved neighbors and obtain the subgraph isomorphism of $T$ into $G$.
According to the notion of the \emph{maximum-leaf degree} (see \Cref{def:nd}), the existence of a vertex $t$ in $T$ is equivalent to the condition $\ld(T)\ge k-1$.
\Cref{section:vertex-many-leaves} of our work is devoted to this particular case.
We continue the current subsection with a discussion of this case as well.

Note that different possible images of $t$ can require different numbers of saved neighbors.
For example, if $\deg_G(u)\ge\delta(G)+k-1$, where $u$ is the image of $t$,  then saving neighbors is not required at all, as $k-1$ neighbors of $u$ remain vacant after mapping the shaved tree of size $\delta(G)+1$.
The less the degree of $u$ is, the more vertices are required to be mapped outside $N_G[u]$. We use $\ndef(u)$, the \emph{neighbor deficiency} of $u$, to denote the number of vertices to map outside $N_G[u]$ (see \Cref{def:ndef}).

In the proof of \Cref{prop:first_step}, we initially mapped a path of three vertices in order to map one vertex outside $N_G[u]$.
\Cref{lemma:2k_path}, the major auxiliary result of \Cref{section:vertex-many-leaves}, pushes the applicability of this method further.
It shows that mapping a path of length $3k$ (starting from $t$ in $T$) initially can save $\ndef(u)$ neighbors of some vertex $u$ in $G$.
This is achieved by mapping the path in $T$ into a path in $G$ starting in $u$ that has roughly every third vertex outside $N_G[u]$.
The choice of $u$ depends drastically on the structure of $G$.
The sufficient requirement for $u$ is having enough vertices in $N_G[u]$ that have more than $k-1$ neighbors outside $N_G[u]$.

If such a choice of $u$ is not possible, the regular ``map every third outside'' procedure is not possible.
In this case, we show that the minimum vertex degree of graph $G-N_G[u]$ is at least $k$.
Such a graph has a path of length at least $k$ starting in an arbitrary vertex.
Then we construct a path in $G$ that starts in $u$, goes to some vertex $v$ in $V(G)\setminus N_G[u]$ through an intermediate vertex in $N_G(u)$, and ends by following a path of length $k-2$ inside $G-N_G[u]$.
This path is of length $k$ and has its last $k-1$ vertices outside $N_G[u]$.
It is easy to see that mapping the first $k+1$ vertices of the path in $T$ into the constructed path in $G$ saves $k-1\ge\ndef(u)$ neighbors of $u$.  
This concludes the overview of \Cref{lemma:2k_path}.

However, \Cref{lemma:2k_path} is not applicable when (i) $T$ has no path of length at least $3k$ starting in $t$ \emph{or} (ii) the minimum degree of $G$ is in $\Oh(k^2)$.
In each of these cases, containment of $T$ in $G$ is not guaranteed and requires an algorithmic approach.
We handle each case separately.
The second case is easier than the first: It guarantees that the size of $T$ is bounded by $\Oh(k^2)$.
Since \probSTI can be solved in $2^{\Oh(|V(T)|)}\cdot\polyn$ running time using the color-coding technique of Alon, Yuster and Zwick \cite{AlonYZ95}, we achieve a time $2^{\Oh(k^2)}\cdot\polyn$ algorithm for the case $\delta(G)=\Oh(k^2)$.
We highlight this idea for further use in the overview.

\medskip\noindent\textbf{Idea III.} \textit{Color-coding applies for graphs of minimum degree bounded by $k^{\Oh(1)}$.}
\medskip

It remains to discuss case (i) when $T$ has no path of length at least $3k$ starting in $t$.
The algorithm for this case also involves color-coding.
The initial step towards the algorithm here is to note that $T$ has a bounded diameter and so does every subtree of $T$.
For the previous cases, we constructed a path-to-path isomorphism that hits enough vertices in $V(G)\setminus N_G[u]$, where $u$ is the fixed\footnote{The choice of the vertex $u$ in $G$ is fixed in the outer loop of the algorithm. Thus, the algorithm considers every possible vertex in $G$ as a candidate for the image of $t$.} image of $t$.
In this case, $T$ has no path of enough length starting in $u$.
However, we know that if $T$ is isomorphic to a subgraph of $G$, then the subgraph isomorphism of $T$ into $G$ occupies at least $\ndef(u)$ vertices outside $N_G[u]$.
Thus, solving the problem is equivalent to finding an isomorphism of a connected subtree of $T$ containing $t$ into $G$ that maps $t$ to $u$ and also maps at least $\ndef(u)$ vertices of $T$ to vertices in $G$ that are outside $N_G[u]$.

We aim to find such a subtree of minimum size, the existence of such a tree is equivalent to the containment of $T$ in $G$.
By minimality, the leaves of such a subtree are necessarily mapped to either $t$ or into $V(G)\setminus N_G[u]$.
Then by the minimality, this tree should have at most $\ndef(u)+1$ leaves, so the size of any minimal subtree is at most $\diam(T)\cdot (\ndef(u)+1)\le k\cdot \diam(T)$.
The obtained bound allows us to use color-coding to find the required subtree in $G$.
We color $G$ with $k\cdot\diam(T)$ colors uniformly at random and use the coloring to find a subgraph of $G$ that is 
(a) isomorphic to the required subtree of $T$, (b) the isomorphism maps $t$ to $u$, and  (c) the isomorphism occupies $\ndef(u)$ vertices in $V(G)\setminus N_G[u]$. 
This is done via dynamic programming in time $2^{\Oh(k\cdot\diam(T))}\cdot\polyn$.
As there is no vertex at a distance at least $3k$ from $t$ in $T$, we have that  $\diam(T)=\Oh(k)$. Hence the overall running time of the algorithm is $2^{\Oh(k^2)}\cdot\polyn$, similar to the case (ii).

This algorithm is important not only for the particular case $\ld(T)\ge k-1$---we shall recall it once again in the overview in a slightly more general setting.
\Cref{subsec:annotated} is dedicated to the corresponding problem which we call \probHitIso.
In this problem, it is allowed to have an arbitrary starting mapping from $T$ into $G$ and have multiple sets in $G$ that are required to be hit by the isomorphism.
Our algorithm for this problem works in time $2^{k^{\Oh(1)} \cdot \diam(T)}\cdot\polyn$, if the total number of vertices to hit by the isomorphism is bounded by $k^{\Oh(1)}$.

The above arguments bring us to an algorithm for \probSTI running in time $2^{\Oh(k^2)}\cdot\polyn$, where $k=|V(T)|-\delta(G)$, for the special case $\ld(T)\ge k-1$.
This is the main result of \Cref{section:vertex-many-leaves} formulated in \Cref{thm:large_leaf_degree}.

\subsection{Filtering out yes-instances}

We still have to consider the case of the {maximum-leaf degree} $\ld(T)<k-1$.
As in the case $\ld(T)\geq k-1$, we first filter out some structural properties of $G$ and $T$ that guarantee  $G$ containing $T$.
 Taking the success of the complement case into account, a good candidate for being filtered out first is the case of $T$ having less than $k-1$ leaves (that actually stops us from applying Idea II by shaving leaves).

In this case, $T$ can be covered by at most $k-3$ (the number of leaves minus one) paths starting in the same arbitrary leaf of $T$ and ending in pairwise distinct leaves of $T$.
Each of these paths consists of at most $\diam(T)+1$ vertices.
Since these paths share at least one common vertex, we obtain $|V(T)|\le 1+(k-3)\cdot \diam(T)<k \cdot \diam(T).$
It follows that $\diam(T)$ is at least $\delta(G)/k$.
As we can handle small values of $\delta(G)$ by making use of Idea III, we need to look only at trees of large diameter.

Up to this point, we still did not introduce any mechanism of constructing isomorphisms of $T$ into $G$ other than mapping outside a neighborhood of a \emph{specific} vertex.
To avoid keeping track of this, we discover that it is possible to save neighbors of \emph{all} vertices of $G$ simultaneously.
Moreover, path-to-path isomorphisms are very suitable for this purpose.
This is the next  idea in the proof of the main result

\medskip\noindent\textbf{Idea IV.} \emph{Shortest paths are very good for saving neighbors.}

\noindent
Assume that $G$ has a shortest $(s,t)$-path that consists of at least $k+2$ vertices.
On the one hand, each vertex of $G$ that does not belong to this path cannot have more than three neighbors belonging to this path (otherwise we can make the $(s,t)$-path shorter).
On the other hand, each vertex of the $(s,t)$-path cannot have more than two neighbors inside the path for the same reason.
Then every vertex in $G$ has at least $k-1$ non-neighbors in the path.
Hence, if we map a path of $T$ into the $(s,t)$-path initially, we can extend this mapping to a subgraph isomorphism of $T$ into $G$ without any additional work following Idea II.
Due to that for each $u\in V(G)$, we have at least $k-1\ge\ndef(u)$ non-neighbors in the isomorphism initially.

The idea above is the basis for an alternative mechanism that is the core part of \Cref{sec:largediameter}.
Idea IV allows to deal with $G$ for the case $\diam(G)\ge k+1$ automatically and the additional work is required only to deal with the case $\diam(G)\le k$.
We turn this upper bound on the diameter of $G$ into our advantage.
In fact, we are able to take any set $S$ of size $o(\delta(G))$ and construct a path of size $\Oh(k\cdot |S|)$ that traverses all vertices of $S$.
We do this iteratively by connecting consecutive vertices in $S$ by a shortest path in $G$.
Since we do not want such segments to intersect, we have to remove the prefix of the path from $G$ before finding the shortest path to the next vertex of $S$.
If during this process the diameter of $G$ becomes even greater than $2k$, we make good use of it (following Idea IV with additional arguments) and construct a path that contains enough non-neighbors of each vertex of $G$.

The rest of \Cref{sec:largediameter} is devoted to finding the appropriate set $S$ in $G$.
This set is required to contain at least $\ndef(u)$ non-neighbors for every vertex $u \in V(G)$.
If the size of the graph is slightly above $\delta(G)$, that is, $|V(G)|\le (1+\epsilon)\cdot \delta(G)$, the existence of a small $S$ is hardly possible.
For this reason, we dedicate a separate  \Cref{sec:diracs} that especially deals with this case for arbitrary $\epsilon\le \frac{1}{4k}$.
The main result of this section, \Cref{lem:dense}, states that if $\delta(G)=\Omega(k^2)$, $|V(G)|\le (1+\epsilon)\cdot \delta(G)$ and $\ld(T)<k$, then $G$ contains $T$ as a subgraph.
The proof is based on extending a partial isomorphism of $T$ into $G$ using unoccupied vertices.
This is always possible to do because the vertices outside the partial isomorphism in $G$ should have many neighbors inside it.
It grants many options for an outer vertex to be inserted between the vertices of the isomorphism, and at least one option will always suffice for extension.

By achieving the lower-bound $|V(G)|\ge (1+\epsilon)\cdot \delta(G)$, we are able to construct the set $S$ of bounded size.
We use the probabilistic method here. For simplicity, we slightly increase the lower bound for the number of vertices up to $|V(G)|\ge (1+\epsilon)\cdot \delta(G)+k^{\Omega(1)}\cdot\log\delta(G)$. (The increase is achieved by lowering the value of $\epsilon$.)
Each vertex $u$ in $G$ with $\ndef(u)>0$ has at least $|V(G)|-\delta(G)>\epsilon \cdot \delta(G)$ non-neighbors in $G$.
For a fixed $u \in V(G)$, a random choice of a vertex in $V(G)$ gives a non-neighbor of $u$ with probability at least $\frac{\epsilon}{1+\epsilon}$.
Then the expected value of deficient vertices $u\in V(G)$ that are not neighbors of a random vertex of $G$ is at least $\frac{\epsilon}{1+\epsilon}\cdot |V(G)|$, so there should be a single vertex in $G$ that is non-neighbor to at least $\frac{\epsilon}{1+\epsilon}$ ratio of all deficient vertices.
We find such vertex in $G$ and put it inside $S$.
With a careful analysis, we show that repeating this step for $\Oh(\log \delta(G) / \epsilon)$ times results in a set $S$ having at least one non-neighbor for each deficient vertex of $G$.
We have to repeat the whole process  $k-2$ times and obtain the required set $S$ of size $\Oh(k\cdot \log \delta(G)/\epsilon)$.
Note that the additive part of the lower bound is required for maintaining the probability $\frac{\epsilon}{1+\epsilon}$ during the consecutive choices of distinct vertices into $S$.
The value of $\epsilon$ we use throughout \Cref{sec:largediameter} is of the form $\frac{1}{k^{\Theta(1)}}$, so the final bound on the size of $S$ is $k^{\Oh(1)}\cdot\log\delta(G)$.

The vertices of $S$ are finally tied into a path in $G$ of length at most $k^{\Oh(1)}\cdot\log\delta(G)$ as discussed above, and a path of sufficient length in $T$ is mapped to the path in $G$.
The discussion sums up into the main result of \Cref{sec:largediameter}, \Cref{thm:large_diameter_many_vertices}.
It filters out trees of diameter $k^{\Omega(1)}\cdot\log\delta(G)$, and the gap between this bound and the desirable $k^{\Oh(1)}$ still remains.

\subsection{Shaving off leaves from distinct neighbors}

The results discussed above leave us with a tree $T$ satisfying $\ld(T)<k-1$ and $\diam(T)\le k^{\Oh(1)}\cdot \log\delta(G)$.
The number of leaves in $T$, in this case, is at least  $k^{\Omega(1)}$ (following the $\delta(G)/\diam(T)$ lower bound discussed before).
Then we can choose at most $k-1$ vertices in $T$, such that in total they have at least $k-1$ adjacent leaves.
We denote this set by $W$. 
The strategy we want to implement in this case is to use Idea II and to construct a mapping of a connected subtree of $T$ containing the vertices of $W$ into $G$ such that at least $\ndef(w)$ non-neighbors of each $w\in W$ are occupied by this mapping.
This is done in \Cref{sec:medium}. The formal proof encapsulates several methods of constructing a mapping that hits non-neighbors of $W$ by exploiting the structure of $T$ and $G$.
This is the final step of filtering out yes-instances before applying the last (and the most involved) of the algorithms we use to prove the main theorem.

The first obstacle encountered here is that the mapping we require to construct should map all vertices of $W$ (as we should know their respective images in $G$ in order to collect non-neighbors).
To our advantage, the choice of $W$ can vary (while $N_T(W)$ has at least $k-1$ leaves) and we will make this choice depending on the structure of $T$.

The initial step here  (in \Cref{subsection:contracting-trivial-paths}) filters out the case of $T$ with $\diam(T)\ge \Omega(k^4)$.
First, $W$ is chosen in a way that it contains two vertices on distance exactly $\diam(T)-2$.
Then we consider the minimum spanning tree $T_W$ of $W$ in $T$.
Its leaves form a subset of $T$, so the number of leaves in $T_W$ is at most $k-1$.
Since $T_W$ has large diameter, there exist long paths in $T_W$ consisting of degree-$2$ vertices, which we refer to as \emph{trivial path}.
We further shrink $T_W$ by contracting edges of each long trivial paths down until its length reaches $2k$.
The shrank tree $T'_W$ has at most $\Oh(k^2)$ vertices and we initialize the mapping with an arbitrary subgraph isomorphism of $T'_W$ into $G$.

The rest of the work is to transform this isomorphism to the isomorphism of $T_W$ into $G$ by embedding vertices of $G$ into trivial paths of the isomorphism image.
We exploit the existence of a trivial path of length $\Omega(k^3)$ and embed a set $S$ (containing $k-1$ non-neighbors for the image of each $w\in W$) into the image of this path.
The technical work here is done using the arguments of \Cref{sec:largediameter} that we discussed already.
If it cannot be done at some moment, then Idea IV helps us to construct a mapping in an alternative way.

The case of $\diam(T)\ge \Omega(k^4)$ is now dealt with yet.
To proceed further, we cannot rely on path-to-path isomorphisms anymore since $T$ has no very long paths.
Then we focus more on the structure of $G$.
This is when we put the notion of $q$-escape vertices into play (see \Cref{def:escape}).
At the beginning of \Cref{subsec:escape}, we start with an isomorphism of $T_W$ into $G$ and show how a $k^{\Omega(1)}$-escape vertex helps in embedding the isomorphism with non-neighbors of $W$.
This, however, requires a vertex of degree $k^{\Omega(1)}$ in $T$.

The rest of \Cref{subsec:escape} partially resolves this issue by making additional assumptions of the structure of $T$.
\Cref{lemma:escape_or_separator} shows an alternative (to exploiting an escape vertex) mechanism to extend a mapping of $T_W$ with enough non-neighbors of $W$.
If this mechanism fails, it produces a separator of $G$ of size $k^{\Omega(1)}$.
Finally, \Cref{lemma:separator_and_separable_tree} allows to use this separator to map a \separable{$k^{\Oh(1)}$} tree $T$ (see \Cref{def:separable}) inside $G$.
This exhausts the discussion of constructions of \Cref{sec:medium} used to prove its main result \Cref{thm:medium_diameter_escape_vertices}.

\subsection{Solving remaining case algorithmically}

We move on to the remaining case of \probSTIAG: $G$ contains no $k^{\Omega(1)}$-escape vertices and $T$ is not \separable{$k^{\Omega(1)}$}, while $\delta(G)\ge k^{\Omega(1)}$.
The main result of \Cref{section:smalldiameter}, \Cref{thm:densest_part}, states that there is a randomized algorithm solving such cases in $2^{k^{\Oh(1)}}\cdot\polyn$ running time.
In the rest of the current part of the overview, we discuss the main parts of the proof of this result.

Following  Idea II, the first part addresses hitting non-neighbors of images of vertices in $W$ (the set of leaf-adjacent vertices in $T$) similarly to the previous part of the overview.
Before we used it for proving that $G$ contains $T$ as a subgraph.
That was achieved by occupying at least $\ndef(w)$ non-neighbors of the image of $w$ in $G$ for each $w\in W$.
But this is only a sufficient condition.
In fact, not all vertices in $W$ require so many saved neighbors.
For example, if $W$ has exactly $k-1$ vertices then we require exactly one vacant neighbor when it comes to mapping the leaf to the neighbor of the image of $w\in W$.
Then if we have one vacant neighbor for $w_1$, two vacant neighbors for $w_2$, three vacant neighbors for $w_3$, and so on, the extension of the isomorphism is guaranteed to be possible.

The paragraph above suggests that a necessary condition should be discovered.
This is exactly what we do in \Cref{subsec:hard_mapping}.
We prove that if the mapping of $W$ into $G$ is known initially then the existence of a subgraph isomorphism of $T$ into $G$ respecting this mapping is equivalent to the existence of a mapping that hits specific sets in $G$.
The details of this are quite complex, so we refer the reader to the statement of \Cref{lemma:multi_leaf_algo}.
It automatically provides a one-to-many reduction to \probHitIso.  
Since the diameter of $T$ is at most $k^{\Oh(1)}$, we obtain an algorithm that runs in $2^{k^{\Oh(1)}}\cdot\polyn$ time and correctly decides the containment of $T$ in $G$, provided the mapping of $W$ into $G$ is given.

The first part alone does not provide any clue on how to choose the mapping of $W$ into $G$.
Trivial enumeration of all possible mappings gives $|V(G)|^{|W|}< |V(G)|^{k}$ possible options and yield only an \classXP-algortihm for \probSTI parameterized by $k$.
The second part of \Cref{section:smalldiameter} resolves this issue. It allows to reduce the situation when a set $W$ guessed by making use of a random sampling.
Its central result, \Cref{lemma:guess_a_mapping} can be turned into a polynomial-time randomized procedure that takes $G$, $T$, $W$ and a single vertex $u\in V(G)$ as input and produces a mapping of $W$ into $G$.
If $T$ is contained in $G$ as a subgraph then with probability at least $2^{-k^{\Oh(1)}}$, this mapping is a restriction of some subgraph isomorphism of $T$ into $G$. This is the only source of randomness in our algorithm which we do not know how to derandomize. 
Combining the results of two parts we obtain a (one-sided error) randomized algorithm for the specific case of \probSTI  with running time $2^{k^{\Oh(1)}}\cdot\polyn$.



\section{Vertex of high leaf-degree}\label{section:vertex-many-leaves}

In this section, we make the first step toward the proof of the main result of the paper. It concerns the case when $T$ has a vertex adjacent to many leaves. 
One of the trivial cases of the problem is when  
a tree $T$ on $\delta(G)+k$ vertices is a star---we guess where the center of the start could be mapped and then check whether there is enough space to map the leaves. 
The main result of this section extends such arguments to the case when $T$ has a vertex of high-leaf degree.

\begin{theorem}\label{thm:large_leaf_degree}
Let $G$ be an $n$-vertex graph, $k$ an integer, and $T$ a tree on $\delta(G)+k$ vertices and 
	with $\ld (T) \ge k - 1$. There is an algorithm deciding whether $G$ contains $T$  as a subgraph in time $2^{\Oh(k^2)} \cdot \polyn$.
\end{theorem}

\subsection{Hitting sets with isomorphism}\label{subsec:annotated}
Before we move on to the main result of this section, we first show auxiliary algorithms for finding small-sized tree containments.
We use the following result of Alon, Yuster, and Zwick \cite{AlonYZ95}.


\begin{proposition}[Theorem~6.1, \cite{AlonYZ95}]\label{col:color_coding_subtree} 
 Let $G$ be an $n$-vertex graph and $T$ be a tree. It can be determined in time $2^{\Oh(|V(T)|)} \cdot \polyn$ whether $G$ contains $T$ as a subgraph. 
\end{proposition}

In the rest of this section, and especially in the later section, a variant of \Cref{col:color_coding_subtree}  with additional constraints will be used.
First, it is helpful to fix the images of some vertices, e.g., the root of the tree.
Second, in the view of upcoming in the later sections structural properties of the solution, it is required to find partial isomorphisms of $T$ ``hitting'' certain subsets of $V(G)$.
A generalized problem encapsulating the properties above is defined next.

\defproblema{\probHitIso}{Graph $G$, tree $T$, an injection $\kappa: S\to V(G)$ for some $S\subset V(T)$, a family of $r$ sets $F_1,F_2,\ldots,F_r\subset V(G)$, and $r$ non-negative integers $k_1,k_2,\ldots,k_r$.}{Determine whether there exists a connected subgraph $T'$ of $T$, and subgraph isomorphism $\sigma:V(T')\to V(G)$ from $T'$ into $G$ respecting $\kappa$, such that $V(T')\supset S$ and $|\Ima\sigma\cap F_i|\ge k_i$ for each $i\in[r]$.}

First we observe that just as the original \probSTI, \probHitIso is in \classFPT parameterized by the size of the tree $T$, by an application of color coding. We show this for the case when the target subtree $T'$ is equal to the whole tree $T$, in order to use later as a blackbox when the subtree achieving the isomorphism is already fixed.

\begin{lemma}\label{lemma:full_tree_color_coding}
	There is an algorithm that determines whether the desired isomorphism of the tree $T' = T$ into $G$ in an instance of \probHitIso exists in time $2^{\Oh(|V(T)| + k)}\cdot\polyn$, where $k = \sum_{j = 1}^r k_j$.
\end{lemma}
\begin{proof}
    When the size of the tree is a parameter, accounting for the additional constraints of the \probHitIso problem is a straightforward extension of the color-coding technique of  \Cref{col:color_coding_subtree}. For the color-coding step, we use $|V(T)|$ colors, where $|S|$ colors are assigned to the images of the vertices of $S$ under $\kappa$,  and the remaining ones are picked randomly and independently throughout the rest of the graph. The dynamic programming then computes, whether for every rooted subtree $T'$ of $T$ there exists a colorful solution with the given characteristics: the partial isomorphism uses precisely the colors among the set $C$, and it hits exactly $k_j'$ elements in each $F_j$, for a given composition $(k_1', \ldots, k_r')$, where $0 \le k_j' \le k_j$ for each $j \in [r]$.
    
    For a fixed solution, random coloring makes it colorful with probability at least $2^{-\Omega(|V(T)|)}$, and the dynamic programming runs in time $2^{\Oh(|V(T)|)} \cdot \polyn$, since there are $2^{|V(T)|}$ choices for the set of colors $C$, and $2^{\Oh(k)}$ choices for the composition $(k_1', \ldots, k_r')$. The total running time is thus $2^{\Oh(|V(T)| + k)}\cdot\polyn$.
    The algorithm can also be derandomized in the standard fashion.
\end{proof}

Finally, we show the main result of this subsection: that \probHitIso is efficiently solvable as long as the total size of ``special'' targets for the isomorphism and the diameter of $T$ are bounded.

\begin{lemma}\label{lemma:general_color_coding_algo}
	\probHitIso admits an algorithm running in time $2^{\Oh(p\cdot \diam(T))}\cdot\polyn$, where $p=|S|+\sum_{i=1}^r k_i$.
\end{lemma}
\begin{proof}
%

    It suffices to show that a sufficient set of choices for $T'$ may be enumerated in the desired time, and then use \Cref{lemma:full_tree_color_coding} to determine whether the desired isomorphism exists for a fixed choice  of $T'$.

    Let $T''$ be the unique minimal subtree of $T$ containing all vertices of $S$, clearly $T''$ is also a subtree of any $T'$ such that the desired isomorphism of $T'$ into $G$ exists. We may assume that $S$ and $T''$ are non-empty, as otherwise we can branch on an arbitrary vertex in $T'$ and its image in polynomial time.

    The algorithm tries all possible choices to extend $T''$ to $T'$. Namely, let $k = k_1 + k_2 + \ldots + k_r$, $t = |V(T'')|$,
    and let $V(T'') = \{w_1, w_2, \ldots, w_t\}$. First the algorithm branches on all compositions of the value of at most $k$ into $t$ terms $a_1$, \ldots, $a_t$, i.e., $k \ge a_1 + \cdots + a_t$, where for every $i \in [t]$, $a_i$ is a nonnegative integer. Intuitively, for every $i \in [t]$, $a_i$ is the number of preimages in $T'$ among the required $k$ vertices in $\bigcup_{j = 1}^r F_j$, that are not in $T''$ but are ``rooted'' in $w_i$ with respect to $T''$. Then for every $i \in [t]$, the algorithm branches on the choice of a rooted tree $T_i'$ that has  $a_i$ leaves each at depth at most $\diam(T)$; if $a_i = 0$ then $T_i'$ is the one-vertex tree. Intuitively, $T_i'$ describes exactly how the selected $a_i$ vertices are connected to $w_i$ in $T$. Finally, the resulting tree $T'$ is obtained from $T''$ by rooting the tree $T'_i$ at $w_i$, for every $i \in [t]$.

    The algorithm then verifies whether the constructed tree $T'$ is a subtree of $T$ such that each vertex in $S$ keeps its place under the isomorhism from $T'$ into $T$. For every $i \in [t]$, let $T_i$ be the subtree of $T$ rooted at $w_i$ that avoids $T''$. More formally, $T_i$ is the connected component of $w_i$ in $T$ after removing the edges of $T''$, that is additionally rooted in $w_i$.
    
    It is easy to observe that $T'$ admits the desired isomorphism into $T$ if and only if for each $i \in [t]$, $T_i'$ is a rooted subtree of $T_i$. The algorithms verifies the latter by invoking the algorithm of \Cref{col:color_coding_subtree} on $T_i'$ and $T_i$. If there exists $i \in [t]$ such that $T_i'$ is not a rooted subtree of $T_i$, the algorithm rejects the choice of $T'$.

    Finally, for every choice of $T'$ that passes the procedure above, the algorithm invokes \Cref{lemma:full_tree_color_coding} to determine whether the desired isomorphism exists from $T'$ into $G$. The algorithm reports that the given instance admits a solution if for at least one choice of $T'$  such an isomorphism exists.

    We next verify the correctness of the algorithm. Clearly, if the algorithm finds a subtree $T'$ of $T$ for which the desired isomoprphism from $T'$ into $G$ exists, then the given instance is a yes-instance. Thus, it only remains to verify that the considered set of candidate subtrees $T'$ is sufficient. Let $T^*$ be a minimal subtree of $T$ such that the desired isomorphism $\sigma^*: V(T^*) \to V(G)$ exists. We show that in this case there also exists a subtree $T'$ of $T$ considered by the algorithm, for which the desired isomorphism also exists.

    Similarly to $T_i$ and $T_i'$, for each $i \in [t]$ let $T_i^*$ be the subtree of $T^*$ rooted at $w_i$ that avoids other vertices of $T''$. Consider a leaf $\ell \ne w_i$ of $T_i^*$, by minimality of $T^*$ there exists $j \in [r]$ such that $\sigma^*(\ell) \in F_j$, and moreover, $|\sigma^*(V(T)) \cap F_j| = k_j$, otherwise removing $\ell$ from $T^*$ would still give a valid solution. This implies that the total number of leaves in $\{T_i^*\}_{i \in [t]}$ that are not in $T''$ is at most $k = k_1 + \cdots + k_r$. Let $a_i^*$ be the number of leaves in $T_i^*$ not counting $w_i$, by the above $a_1^* + \cdots + a_t^* \le k$. Therefore at a certain step the algorithm considered the same composition $(a_1, \ldots, a_t)$, i.e., for each $i \in [t]$, $a_i = a_i^*$. Since $T_i^*$ is a tree with $a_i$ leaves, and the distance from each to the root does not exceed $\diam(T)$, as $T_i^*$ is a subtree of $T$, the algorithm also considered the choice $T_i' = T_i^*$, for each $i \in [t]$. At this step of the algorithm, the constructed tree $T'$ is exactly $T^*$, and since $T_i^*$ is a subtree of $T$ rooted at $w_i$ and avoiding $T''$, the verification correctly determined that $T'$ is a valid choice for a subtree. Thus the algorithm of \Cref{lemma:full_tree_color_coding} is queried with the subtree $T'$ at the respective step of the algorithm, and since $\sigma^*$ is an isomorphism of $T^* = T'$ into $G$ with the desired properties, the algorithm correctly identified the given instance as a yes-instance.

    It remains to analyze the running time of the algorithm. Since $T''$ is the minimal subtree of $T$ containing the selected $|S|$ vertices, its size $t$ is at most $|S| \cdot \diam(T)$. The number of choices for the composition $a_1 + \cdots + a_t \le k$ is at most $2^{\Oh(k + t)}$. The rooted tree $T_i'$ has $a_i$ leaves and at most $a_i \cdot \diam(T) + 1$ vertices, thus the number of choices is at most $2^{\Oh(a_i \cdot \diam(T))}$, resulting in $2^{\Oh(p \cdot \diam(T))} \cdot \polyn$ choices for the tree $T'$, as $p = |S| + k$. Verifying that the tree $T_i'$ is a rooted subtree of $T_i$ takes time $2^{\Oh(a_i \cdot \diam(T))}$ by \Cref{col:color_coding_subtree}, which is still under $2^{\Oh(p \cdot \diam(T))} \cdot \polyn$ in total. Finally, the size of $T'$ is at most $|V(T'')| + \sum_{i = 1}^t a_i \cdot \diam(T) \le p \cdot \diam(T)$, thus the algorithm of \Cref{lemma:full_tree_color_coding} takes time $2^{\Oh(p \cdot \diam(T))} \cdot \polyn$. Therefore, the total running time is bounded by $2^{\Oh(p \cdot \diam(T))} \cdot \polyn$.
\end{proof}

\subsection{Extending isomorphism of $k - 1$ leaves}
As the next preparatory step, we show another general tool for constructing isomorphisms. Namely, if we manage to find a partial isomorphism of a subtree of $T$ that contains $k - 1$ leaves, while leaving sufficiently many neighbors of the leaves' parents unoccupied, such an isomorphism can always be extended to the isomomorphism of the whole tree $T$.

\begin{lemma}\label{lemma:multi_saved_k_neighbours}
Let $G$ be a graph and $T$ be a tree on $\delta(G)+k$ vertices. 
	Let $L$ be a set of $k - 1$ leaves of $T$ and let $W:=N_T(L)$ be the set of vertices adjacent to $L$ in $T$.
	Let $T'$ be a  subtree of $T$ such that $W\subseteq V(T')\subseteq V(T-L)$.
	If there exists a subgraph isomorphism $\sigma: V(T')\to V(G)$ such that for each $w\in W$ $$|\Ima\sigma\setminus N_G[\sigma(w)]|\ge \ndef(\sigma(w)),$$
	then $G$ contains $T$ as a subgraph. 
\end{lemma}
\begin{proof}
	We first extend $\sigma$ to an isomorphism of $T-L$ into $G$.
	This is always possible by \Cref{prop:chvatal-generalized} since $|V(T-L)|=\delta(G)+1$.
	Now we have an isomorphism of $T-L$ into $G$, and it remains to map the $k - 1$ vertices of $L$.
	
	Let vertices in $L$ be $\ell_1, \ell_2, \ldots, \ell_{k - 1}$ and let the only neighbour of $\ell_i$ in $T$ be $w_i \in W$.
	Since the isomorphism occupies at least $\ndef(\sigma(w_i))$ non-neighbours of $\sigma(w_i)$ in $G$, at least $k - 1$ neighbours of $\sigma(w_i)$ are not occupied, for each $i \in [k - 1]$.
	First, take any of the free neighbours of $\sigma(w_1)$ as the image of $\ell_1$.
	Each $\sigma(w_i)$ still has at least $k-2$ neighbours that are not occupied; take any of the free neighbours of $\sigma(w_2)$ as the image of $\ell_2$.
	Repeat this process for each $j \in \{3,\ldots, k - 1\}$: at the $j$-th step, each $\sigma(w_i)$ has at least $k-j$ non-occupied neighbours, so there is always a free neighbour for the image of $\ell_j$.
\end{proof}

\subsection{\probSTI when there is a vertex with many leaves}

Finally, we work directly towards the proof of Theorem~\ref{thm:large_leaf_degree}.
We state the main combinatorial observation behind the theorem next.

\begin{lemma}
    \label{lemma:2k_path}
    Let $G$ be a graph and $T$ be a tree on $\delta(G)+k$ vertices.  Suppose that $T$ contains a vertex 
 $s \in V(T)$, such that $s$ has $k - 1$ leaf neighbors, and such that there is a path of length $3k$ in $T$ starting with $s$. If 
 $\delta(G) \ge 11k^2$, and $|V(G)| \ge |V(T)|$, then $G$ contains $T$ as a subgraph. 
%
%
\end{lemma}
Intuitively, if the conditions of the lemma are not satisfied, then either $\delta(G)$ and so $|V(T)|$ are bounded by $\Oh(k^2)$, or the diameter of $T$ is bounded by $\Oh(k)$. In this case \Cref{col:color_coding_subtree}  or \Cref{lemma:general_color_coding_algo} can be used to find a suitable  (partial) isomorphism of $T$, thus Lemma~\ref{lemma:2k_path} is the main obstacle towards showing Theorem~\ref{thm:large_leaf_degree}.

\begin{proof}[Proof of Lemma~\ref{lemma:2k_path}]
    By~\Cref{lemma:multi_saved_k_neighbours}, it suffices to find a suitable subtree $T'$ of $T$. 
    If there is a vertex $v \in V(G)$ of degree at least $\delta(G) + k - 1$, then $T'$ is the subtree containing the single vertex $s$, and the isomorphism maps $s$ to $v$. Since $\ndef(v) = 0$ in this case, \Cref{lemma:multi_saved_k_neighbours} is immediately applicable.
    Thus in the following we assume $\Delta(G) < \delta(G) + k - 1$.

    We now consider two cases based on whether there is a vertex in $G$ that has a good ``expansion'' out from its neighborhood. For a vertex $v \in G$, we call a vertex in $N_G(v)$ \emph{expanding} if it has at least $k - 1$ neighbors outside of $N_G(v)$. Intuitively, if we have sufficiently many expanding vertices in $N_G(v)$, we can always embed the long path of $T$ into $G$ by mapping $s$ to $v$ and then going outside of $N_G(v)$ sufficiently often. In the complement case, we have that no vertex has many expanding vertices and thus the neighborhood of every vertex is extremely dense, which is helpful for finding a suitable embedding in a different way. We now move on to the details of both cases.

    \subparagraph*{Expanding case.} Let $v$ be a vertex in $G$ that has at least $3k$ expanding vertices in $N_G(v)$. Let $P = p_0-p_1-\dots -p_t$ be a path in $T$ where $t = 3k$ and $p_0 = s$. We construct a (partial) isomorphism $\sigma: P \to G$ that we will later use to invoke \Cref{lemma:multi_saved_k_neighbours} on. We start by setting $\sigma(p_0) = v$, and then continue defining $\sigma$ inductively. Let $i \in \{0, 1, \ldots, t - 1\}$, we consider three cases based on $\sigma(p_i)$:
    \begin{description}
        \item[Outer] $\sigma(p_i) \notin N_G(v)$, then we set $\sigma(p_{i + 1})$ to be an arbitrary neighbor of $\sigma(p_i)$ that is not yet used by $\sigma$;
        \item[Non-expanding] $\sigma(p_i) \in N_G(v)$ and $\sigma(p_i)$ is not expanding in $N_G(v)$, then we set $\sigma(p_{i + 1})$ to be an arbitrary neighbor of $\sigma(p_i)$ in $N_G(v)$ that is expanding and not previously used by $\sigma$;
        \item[Expanding] $\sigma(p_i) \in N_G(v)$ and $\sigma(p_i)$ is expanding in $N_G(v)$, then we set $\sigma(p_{i + 1})$ to be an arbitrary neighbor of $\sigma(p_i)$ outside of $N_G(v)$ that is not previously used by $\sigma$.
    \end{description}
    In case the respective rule out of the above is not applicable (i.e., a suitable next vertex does not exist), we stop the procedure and let $t'$ be the last index such that $p_{t'}$ is assigned by $\sigma$; we also let $P' = p_0p_1\ldots p_{t'}$ be the respective subpath of $P$.
By construction, $\sigma$ is an isomorphism of $P'$ into $G$. We now argue that $\sigma$ maps at least $k - 1$ vertices of $P'$ outside $N[v]$.

First, by construction of $\sigma$, every \textbf{Expanding} vertex is necessarily followed by an \textbf{Outer} vertex, and every \textbf{Non-expanding} vertex is followed by an \textbf{Expanding} vertex. Thus, at least every third vertex in the sequence is \textbf{Outer}, and starting from $i \ge 1$ none of them coincide with $v$. Therefore, if $t' = t$, then the claim holds as $t \ge 3k$. Otherwise, $t' < t$ and the respective rule is not applicable to $\sigma(p_{t'})$. Observe that $\sigma(p_{t'})$ is not \textbf{Outer} as any neighbor of $\sigma(p_{t'})$ can be chosen in this case, and $\deg \sigma(p_{t'}) \ge \delta(G) \ge 3k + 1 = |V(P)|$. Assume $\sigma(p_{t'})$ is \textbf{Expanding}, then there are no more neighbors of $\sigma(p_{t'})$ outside of $N_G(v)$ not taken by $\sigma$. By definition of expanding vertices, there are at least $k - 1$ such neighbors and thus $\sigma$ maps at least $k - 1$ vertices outside of $N[v]$, fulfilling the claim. Finally, assume $\sigma(p_{t'})$ is \textbf{Non-expanding}. Then all of its expanding neighbors in $N_G(v)$ are taken by $\sigma$. There are at least $\delta(G)$ neighbors of $\sigma(p_{t'})$ in $G$, and at least $\delta(G) - k + 1$ of them are in $N_G(v)$, since $\sigma(p_{t'})$ is not expanding. In total, there are at most $\delta(G) + k - 2$ vertices in $N_G(v)$ and at least $3k$ of them are expanding, thus the number of expanding neighbors of $\sigma(p_{t'})$ is at least $k - 1$. Since all of them are taken by $\sigma$ and each is followed by a vertex outside of $N[v]$, $\sigma$ takes at least $k - 1$ vertices outside of $N[v]$ in this case as well.

It remains to observe that invoking \Cref{lemma:multi_saved_k_neighbours} on $\sigma$ gives the desired isomorphism, as at least $k - 1 \ge \ndef(v)$ vertices outside of $N[v]$ are taken by the constructed partial isomorphism.

    \subparagraph*{Dense case.}
    Since the conditions for the expanding case are not satisfied, we can assume that for every $v \in V(G)$, there are less than $3k$ expanding vertices in $N_G(v)$. By the pigeon-hole principle, this implies that every non-expanding vertex in $N_G(v)$ has at least $\delta(G) - 4k$ non-expanding neighbors in $N_G(v)$, as it has at least $\delta(G)$ neighbors in total, at most $k - 1$ outside of $N_G(v)$ since it is not expanding, and less than $3k$ expanding vertices exist in $N_G(v)$.
    We further subdivide the dense case into two subcases depending on the diameter of $G$.

    \textit{Diameter of $G$ is at least $3$.} 
    Consider vertices $u, v \in V(G)$ at distance $3$ from each other. Take vertices $a \in N_G(u)$, $b \in N_G(v)$ with $ab \in E(G)$, i.e., $u-a-b-v$ 
    is the shortest $(u,v)$-path. Denote by $B$ the set of non-expanding neighbors of $v$, observe that $N_G(v) \cap N_G(u) = \emptyset$, as otherwise there is a shorter path between $u$ and $v$, thus $B \cap N_G(u) = \emptyset$. Also, by the starting assumption of the dense case, $\delta(G[B]) \ge \delta(G) - 4k \ge k$. Let $P'$ be the prefix of $P$ on $k + 1$ vertices, i.e., $P'$ is also a path in $T$ starting with $s$. Construct the isomorphism $\sigma : V(P') \to V(G)$ in the following way. First, map $s$ to $u$, the second vertex of the path to $a$, and the third vertex to $b$. Then greedily map the remaining $k - 2$ vertices inside $B$, since $\delta(G[B]) \ge k$ and exactly $k - 1$ vertices are to be mapped inside $B$, finding the next unoccupied neighbor is always possible. Finally, by Lemma~\ref{lemma:multi_saved_k_neighbours}, $\sigma$ can be extended to an isomorphism of the whole tree $T$ into $G$.

    \textit{Diameter of $G$ is at most $2$.} 
    First, assume there exist two non-adjacent vertices $u, v \in V(G)$ with $|N_G(u) \cap N_G(v)| < 6k$.
    Under this assumption, we always find an isomorphism similarly to the previous case of $\diam(G) \ge 3$. Let $B$ be the set of non-expanding vertices in $N_G(u)$. Again, $\delta(G[B]) \ge \delta(G) - 4k$, and since $|N_G(u) \cap N_G(v)| < 6k$, $\delta(G[B \setminus N_G(v)]) > \delta(G) - 10k \ge k$. We constuct an isomorphism of $P'$, a prefix of $P$ with $k + 1$ vertices, to $G$ by mapping $s$ to $v$, the next vertex to an arbitrary common neighbor of $u$ and $v$, the third vertex to $u$, and then procceding greedily inside $B \setminus N_G(v)$. Since the last $k - 1$ vertices of $P'$ are all outside of $N[v]$, by Lemma~\ref{lemma:multi_saved_k_neighbours}, we can extend this partial isomorphism to an isomorphism of $T$.

    Otherwise, for every two non-adjacent $u, v \in V(G)$, there exists a vertex $w\in N_G(u) \cap N_G(v)$ that is non-expanding for both $u$ and $v$.
    Let $u \in V(G)$ be the vertex achieving $\deg u = \delta(G)$, and let $v \in V(G)$ be another vertex that is non-adjacent to $u$. We can assume $v$ exists since otherwise $|V(G)| = \deg u + 1 = \delta(G) + 1 < |V(T)|$ which contradicts the conditions of the lemma. Let $w$ be a non-expanding vertex in $N_G(u) \cap N_G(v)$ for both $u$ and $v$. This means that $|N_G(w) \cap N_G(u)| \ge \delta(G) - k$ and $|N_G(w) \cap N_G(v)| \ge \delta(G) - k$. Also, $|N_G(u)|, |N_G(v)| \le \Delta(G) < \delta(G) + k$, thus both $|N_G(u) \setminus N_G(w)|$ and $|N_G(v) \setminus N_G(w)|$ are at most $2k$. Therefore, $|N_G(w)| \ge |N_G(u) \cup N_G(v)| - 4k$, and since $|N_G(w)| \le \Delta(G) < \delta(G) + k$,  $|N_G(u) \cup N_G(v)| < \delta(G) + 5k$. Since $|N_G(v)| \ge \delta(G)$, we get $|N_G(u) \setminus N_G(v)| < 5k$.
    
    Now, assume there exist $k - 1$ distinct non-neighbors of $u$, by the above each of them is not adjacent to less than $5k$ vertices in $N_G(u)$, thus in total less than $5k^2$ vertices of $N_G(u)$ are not adjacent to some of the selected $k - 1$ non-neighbors of $u$. This means that the number of expanding vertices in $N_G(u)$ is at least $|N_G(u)| - 5k^2 + 1 > 3k$, since $|N_G(u)| = \delta(G) \ge 8k^2$, which is a contradiction as we assume to not be in the expanding case.

Finally, we get that there are less than $k - 1$ non-neighbors of $u$. The whole graph $G$ therefore contains at most $\delta(G) + k - 1$ vertices, as $|N_G(u)| = \delta(G)$, which means that $|V(G)| < |V(T)| = \delta(G) + k$, contradicting the conditions of the lemma.

\end{proof}
\subsection{Proof of \Cref{thm:large_leaf_degree}}
We now complete the proof of the main result of this section by using Lemma~\ref{lemma:2k_path}.

\begin{proof}[Proof of \Cref{thm:large_leaf_degree}]
    First, if $\delta(G) < 11k^2$, then $|V(T)| \le 11k^2 + k$, and by using the algorithm from \Cref{col:color_coding_subtree} , we process the instance in time $2^{\Oh(k^2)} \cdot \polyn$.

    Then, assume there is no path of length $3k$ starting from $s$ in $T$, which implies that $\diam(T) < 6k$.
    In this case, we reduce the problem to \probHitIso. Specifically, we fix $S = \{s\}$, and try all possible variants of the isomorphism $\kappa: \{s\} \to V(G)$, which is equivalent to fixing the image $v \in V(G)$ of $s$, i.e., $\kappa(s) = v$. We also set $F_1 = V(G) \setminus N[v]$, and $k_1 =\ndef(v)$. We now invoke the algorithm given by \Cref{lemma:general_color_coding_algo} on the instance $(G, T, \kappa, \{F_1\}, \{k_1\})$ of \probHitIso. We report the yes-instance if for some choice of $v$ the isomorphism is found, and no-instance otherwise. The running time is bounded by $2^{\Oh(k^2)} \cdot \polyn$ by \Cref{lemma:general_color_coding_algo}. It remains to show that this procedure always returns the correct answer.

    In one direction, let $\sigma' : T' \to V(G)$ be a solution to $(G, T, \kappa, \{F_1\}, \{k_1\})$, i.e., $T'$ is a subtree of $T$, $\sigma'$ is an isomorphism that maps $s$ to $v$, and at least $\ndef(v)$ vertices of $T'$ are mapped outside of $N[v]$. We can assume that $T'$ contains no leaves adjacent to $s$---removing them from $T'$ does not change the property of the solution, i.e., $T'$ is still connected and an adjacent leaf of $s$ could not be mapped outside of $N[v]$. We now let $L$ be an arbitrary set of $k - 1$ leaves adjacent to $s$ in $T$, and invoke Lemma~\ref{lemma:multi_saved_k_neighbours} on $\sigma'$. Since by the above $|\Ima\sigma'\setminus N[v]| \ge \ndef(v)$, this gives an isomorphism of $T$ to $G$.

    In the other direction, let $\sigma$ be an isomorphism of $T$ into $G$. Let $v = \sigma(s)$, $\sigma$ maps at least $\ndef(v)$ vertices outside $N[v]$, since $\deg_G( v) \le |V(T)| - 1 - \ndef(v)$ by definition of $\ndef(v)$. Therefore $\sigma$ is a solution to the instance $(G, T, \kappa, \{F_1\}, \{k_1\})$ of \probHitIso constructed for $\kappa(s) = v$.

    Finally, if neither of the above cases occurs, then $\delta(G) \ge 11k^2$, and there is a path of length $3k$ starting with $s$ in $T$. By Lemma~\ref{lemma:2k_path}, the isomorphism of $T$ to $G$ always exists and can be constructed in polynomial time.
\end{proof}

\section{Small diameter trees and separable $G$}\label{section:smalldiameter}

In this section, we prove the following algorithmic result.
It allows to handle trees with small leaf-degree, but the scope of its application is restricted to the specific structure of $T$ and $G$. 

\begin{theorem}\label{thm:densest_part}
	Let $p\ge 1$ and $k\geq 3$  be integers.
	There is a randomized algorithm that for a given graph $G$,  a tree $T$ on $\delta(G)+k$ vertices such that
	\begin{enumerate}
		\item $\delta(G)\ge k^{3p+1}$,  
				\item there are no $k^{p}$-escape vertices in $G$, 
		\item $\ld(T)<k$, and
		\item $T$ is not \separable{$k^{p}$},
	\end{enumerate}
	determines whether  $G$ contains $T$ as a subgraph in time $2^{{k^{\Oh(p)}}}\cdot\polyn$ with probability at least $\frac{1}{2}$.
	The algorithm is one-sided error and reports no false-positives.
\end{theorem}

Before we proceed with the proof of \Cref{thm:densest_part}, we establish two key ingredients required for the proof.

\subsection{Extending leaf-adjacent mappings}\label{subsec:hard_mapping}

The main idea of the following lemma is close to \Cref{lemma:multi_saved_k_neighbours}.
However, \Cref{lemma:multi_saved_k_neighbours} gives only a sufficient condition for the existence of a subgraph isomorphism.
The proof of the following lemma requires more precise and sophisticated counting neighbors of sets.

\begin{lemma}\label{lemma:multi_leaf_algo}
	Let $G$ be a graph and let $T$ be a tree consisting of $\delta(G)+k$ vertices for $k\ge 1$.
	Let $S$ be a set of $k-1$ leaf-adjacent vertices of $T$.
	There is an algorithm that, given $G$ and $T$, a mapping $\kappa: S\to V(G)$, determines whether there exists a subgraph isomorphism from $T$ into $G$ respecting $\kappa$.
	The running time of the algorithm is $2^{k^{\Oh(1)}\cdot\diam(T)}\cdot\polyn$.
\end{lemma}

\begin{proof}
	Denote the vertices of $S$ by $s_1, s_2, \ldots, s_{k-1}$.
	Let $L=\{\ell_1, \ell_2, \ldots, \ell_{k-1}\}$ be a set of $k-1$ leaves of $T$, such that it contains exactly one neighbor $\ell_i$ of $s_i$ for each $i \in [k-1]$.
	By $T'$, we denote the subtree of $T$ without the leaves of $L$, i.e.\ $T':=T-L$.
	We also define $u_i:=\kappa(s_i)$ for each $i\in [k-1]$.
	
	To present the algorithm, we first study the parameters of the subgraph isomorphism $\sigma$ that extends $\kappa$.
	These parameters can be represented as a sequence of $\Oh(k)$ non-negative integers bounded by $k$.
	Our algorithm will consider every possible combination of the parameters.
	Based on a fixed combination, the algorithm tries to reconstruct a subgraph isomorphism satisfying these properties.
	In what follows, we show that if $\sigma$  exists, then our algorithm will successfully construct a subgraph isomorphism from $T$ into $G$ extending $\kappa$, for the choice of parameters corresponding to $\sigma$.

	Suppose that there exists a subgraph isomorphism $\sigma: V(T)\to V(G)$  such that the restriction of $\sigma$ onto $S$ equals $\kappa$.
	Let $\sigma'$ be the restriction of $\sigma$ onto $V(T')$.
	The first set of parameters of $\sigma$ (in fact, parameters of $\sigma'$) is the sequence $a_1, a_2, \ldots, a_{k-1}$.
	For each $i \in [k-1]$, we define $a_i$ to be the number of non-neighbors of $u_i$, the image of $s_i$,  that are occupied by $\sigma'$, but capped with the deficiency of $u_i$.
	That is,
	$$a_i(\sigma')=\min\{|\Ima\sigma'\setminus N_G[u_i]|, \ndef(u_i)\}.$$
	Thus $a_i$ represents the number of free neighbors of $u_i$ that can be used for mapping $\ell_i$  when extending $\sigma'$  to a   subgraph isomorphism of the whole $T$.
	Note that $a_i$ is formally defined as a function of $\sigma'$, but we will often omit the ``$(\sigma')$'' argument for simplicity. 
	
If for a subgraph isomorphism of $T'$ we have $a_i=\ndef(u_i)$ for each $i\in [k-1]$, then by \Cref{lemma:multi_saved_k_neighbours},  
the mapping of $T'$ could be extended to 
  a subgraph isomorphism of $T$. 
	When this condition is not satisfied, it could be that the subgraph isomorphism of $T'$ cannot be extended. Even worse, it also could happen that  $\sigma'$ might have $a_i=0$ for some (or even all) $i\in[k-1]$.
This forces us to identify more properties and parameters of $\sigma'$.

	Let $B$ be the set of ``problematic'' (for the extension of $\sigma'$) vertices among $u_1,u_2,\ldots, u_{k-1}$.
	That is,
	$$B(\sigma'):=\{u_i: i \in [k-1], a_i(\sigma')<\ndef(u_i)\}.$$	
Another property of a partial isomorphism that is important for analyzing whether it could be further extended is the number of neighbors of the problematic vertices $B(\sigma')$ in $G$. This brings us to the definition of the third parameter of $\sigma'$ 
	$$X(\sigma'):=\left(\Ima\sigma'\cap N_G(B(\sigma'))\right)\setminus \bigcap_{u_i\in B(\sigma')} N_G(u_i),$$
	that is, the neighbors of the vertices in $B$ occupied by $\sigma'$, except the vertices that are adjacent to all vertices of $B$.
	The following claim bounds the search space for $X$.
	
	\begin{claim}
		$|N_G(B)\setminus \cap_{u_i\in B}N_G(u_i)|< k^3$ and $|X|< k^2$.
	\end{claim}
	\begin{claimproof}
		First note that for each $u_i \in B$ we have $0\le a_i<\ndef(u_i)$, so $\deg_G(u_i)<\delta(G)+k$.
		Also $|\Ima \sigma'\cap N_G(u_i)|\ge |\Ima\sigma'|-a_i$.
		Hence, at most $a_i$ vertices in $\Ima\sigma'$ are non-neighbors to $u_i$ for each $u_i\in B$.
		
		We have that at most $\sum_{u_i\in B}a_i\le |B|\cdot (k-1)$ vertices in $\Ima\sigma'$ can be a non-neighbor to at least one $u_i \in B$.
		Hence vertices in $B$ have at least $|\Ima\sigma'|-|B|\cdot(k-1)$ common neighbors.
		Since none of the vertices of $X$ is adjacent to all vertices of $B$ inside $\Ima\sigma'$, we have that  $|X|\le |B|\cdot(k-1)<k^2$.
		
		For each $u_i\in B$, the number of its neighbors that are not the common neighbors of all vertices of $B$, is  is at most 
		$$\deg_G(u_i)-(|\Ima\sigma'|-|B|\cdot(k-1))<(\delta(G)+k)-(\delta(G)+1-|B|\cdot(k-1))=(|B|+1)\cdot(k-1).$$
				
		Summing up these bounds over all $u_i\in B$, we obtain the bound $|B|\cdot(|B|+1)\cdot (k-1)<k^3$, as required by the statement of the claim.
	\end{claimproof}
	
	The final parameter of $\sigma'$ that we consider is the number of vertices occupied by $\sigma'$ that are not related to $B$, that is,
	$$a_B(\sigma'):=|\Ima\sigma'\setminus N_G[B(\sigma')]|.$$
	This parameter is very close in meaning to $a_i$, since it represents the number of preserved neighbors of $B$ in $G$.
	Let us quickly bound its value.
	
	\begin{claim}
		$a_B\le \min\{a_i: u_i \in B\}$.
	\end{claim}
	\begin{claimproof}
		By definition, $a_i\le |\Ima\sigma'\setminus N_G[u_i]|$ for each $u_i\in B$.
		At the same time, $N_G[B]\supset N_G[u_i]$ and, consequently, $(\Ima\sigma'\setminus N_G[u_i])\supset(\Ima\sigma'\setminus N_G[B])$. 
		Hence, by definition of $a_B$, $a_i\ge a_B$ for each $u_i \in B$.
	\end{claimproof}
	
	We have defined all parameters of $\sigma'$, that is, $a_1,a_2,\ldots,a_{k-1}$, $B$, $X$ and $a_B$. While there could be $n^{\Oh(|V(T')|)}$ different ways for a mapping $\sigma'$ to map $T'$ into $G$, the number of possible combinations of the parameters is significantly smaller. Indeed, 
	the values of $a_i$ and $a_B$ are integers within the range $\{0,\ldots, k-1\}$.
	The value of $B$ is derived from $a_1,a_2,\ldots,a_{k-1}$.
	To form set $X$, we have at most $(k^3)^{k^2}$ possible options  depending on $B$ only.
	We conclude that there are $2^{\Oh(k^2)}$ different combinations of these parameters.  All these combinations are easily enumerated in time, up to a polynomial factor,  proportional to the total number of combinations.
	
	We move on to the core part of the algorithm. The algorithm does not go through all possible guesses for mapping $\sigma'$. Instead, we run the algorithm for every 
	 valid choice of the parameters $a_1,a_2,\ldots,a_{k-1}$, $B$, $X$, and $a_B$.
	For each choice of the parameters, the algorithm either outputs an isomorphic embedding of $T$ in $G$ or fails. We will show that 
	when the choice of the 
	parameters corresponds to $\sigma'$, then the algorithm always finds  the required subgraph  isomorphism of $T$ in $G$. 
	Thus, if there exists a subgraph isomorphism $\sigma: V(T)\to V(G)$  such that the restriction of $\sigma$ onto $S$ equals $\kappa$ and 
$\sigma'$ is the restriction of $\sigma$ onto $V(T')$, then our algorithm constructs  an isomorphic embedding of $T$ into $G$ for the choice of the parameters corresponding to $\sigma'$.

We assume that we guessed correctly the parameters of $\sigma'$. Let us remark that the algorithm does not recover mappings $\sigma'$ and $\sigma$. Instead, it uses the parameters of $\sigma'$ to compute another mapping $\xi'$, and then extends $\xi'$ to  subgraph  isomorphism $\xi$ mapping $T$ to $G$ that is also compatible with mapping $\kappa$. 
 %
	We want subgraph isomorphism $\xi'$ from $T'$ to $G$ to satisfy the following conditions:
	\begin{itemize}
		\item $a_i(\xi')\ge a_i(\sigma')$,
		\item $\Ima\xi'\supset X(\sigma')$, and
		\item $|\Ima\xi'\setminus N_G[B(\sigma')]| \ge a_B(\sigma')$.
	\end{itemize}
Let us remark that in particular, $\sigma'$ satisfies these three conditions. 
	\begin{claim}\label{claim:runningtimeofxi}
		The subgraph isomorphism $\xi'$ can be found in $2^{\Oh(k^2) \cdot \diam(T)}\cdot\polyn$ running time.
	\end{claim}
	\begin{claimproof}
		First we find a partial isomorphism from $T'$ into $G$ by reducing to \probHitIso.
		We choose $F_i:=N_G(u_i)$ and $k_i:=a_i(\sigma')$ for each $i\in[k-1]$.
		The next set is $F_k:=X(\sigma')$ with $k_k:=|X(\sigma')|$.
		The final set is $F_{k+1}:=V(G)\setminus N_G[B(\sigma')]$ with $k_{k+1}:=a_B(\sigma')$.
		
		The constructed instance of \probHitIso thus is $$\left(G,T',\kappa,(N_G(u_1),N_G(u_2),\ldots,N_G(u_{k-1}),X,V(G)\setminus N_G[B]),(a_1,a_2,\ldots,a_{k-1},|X|, a_B)\right).$$
		This instance's constraints exactly correspond to the three conditions in the definition of $\xi'$.
		This instance is a yes-instance by our assumption about the existence of $\sigma'$,  and, hence of $\xi'$.
		Apply algorithm of \Cref{lemma:general_color_coding_algo} to this instance and in time $2^{\Oh(k^2)\cdot \diam(T')}\cdot\polyn$ obtain the intermediate partial isomorphism.
		Since $|V(T')|=\delta(G)+1$, by \Cref{prop:chvatal-generalized}, this mapping could be extended to the subgraph isomorphism $\xi': V(T')\to V(G)$ in polynomial time. 
	\end{claimproof}

	The last subroutine of the algorithm extends $\xi': V(T')\to V(G)$ into an isomorphism $\xi:V(T)\to V(G)$.
	The subroutine here is quite straightforward, since we just need to find a matching between $\kappa(S)$ and $V(G-\Ima\xi')$ in $G$ saturating every $u_i\in\kappa(S)$.
	Then the edge of the matching incident with $u_i$ gives an image for $\ell_i$.
	The matching, if it exists,  could be found 
	in polynomial time. Thus what is left is the proof that such a matching always exists.
	
	\begin{claim}
		There is a matching between $\kappa(S)$ and $V(G-\Ima\xi')$ in $G$ saturating $\kappa(S)$.
	\end{claim}
	\begin{claimproof}
		We focus on saturating the vertices in $B(\sigma')$, since each $u_i\notin B(\sigma')$ has at least $k-1$ neighbours outside $\Ima\xi'$.
		They need not be taken care of, we just saturate them arbitrarily in the end.
		In the rest of the proof of the claim, we assume that $B(\sigma')$ is not empty.
	
		Define the common neighbors of $B$ as		
		$$C(\sigma')= \bigcap_{u_i\in B(\sigma')} N_G(u_i).$$
		Observe that by the definition of $B,C,X$ and $a_B$,
		$$|\Ima\sigma'|=|B(\sigma')|+|\Ima\sigma'\cap C(\sigma')|+|X(\sigma')|+a_B(\sigma'),$$
		On the other hand,
		$$|\Ima\xi'|=|B(\sigma')|+|\Ima\xi'\cap C(\sigma')|+|\Ima\xi'\cap N_G(B(\sigma'))\setminus C(\sigma')|+|\Ima\sigma'\setminus N_G[B(\sigma')]|.$$
		
		Since $|\Ima\xi'|=|\Ima\sigma'|$, we have
		\begin{equation}\label{eqn:main_iso_counting}
		|\Ima\sigma'\cap C|+|X|+a_B(\sigma')=|\Ima\xi'\cap C|+|\Ima\xi'\cap N_G(B)\setminus C|+|\Ima\xi'\setminus N_G[B]|.
		\end{equation}
		
		We split the vertices in $B$ in three groups depending on how $\sigma$ maps their neighbors in $L$.
		The first group consists of vertices $u_i\in B$ such that the image of its leaf $\ell_i$ in $\sigma$ is not occupied by $\xi'$: $$B_1:=\{u_i\in B: \sigma(\ell_i)\notin\Ima\xi'\}.$$
		The second group consists of $u_i\in B$ such that $\sigma(\ell_i)$ is occupied by $\xi'$, but is common to all vertices in $B$: $$B_2:=\{u_i\in B: \sigma(\ell_i)\in\Ima\xi'\cap C\}.$$
		The remaining group  is $$B_3:=\{u_i\in B:\sigma(\ell_i)\in \Ima\xi'\setminus C\}.$$
		
		We  construct the matching saturating $B=B_1\cup B_2\cup B_3$.
		Denote by $L_1:=\{\ell_i: u_i\in B_1\}$ the set of leaves for $B_1$.
		For each $u_i\in B_1$, we add edge $u_i\sigma(\ell_i)$ to the matching.
		This is legitimate since $\sigma(\ell_i)\notin\Ima\xi'$ for each $\ell_i\in L_1$.
		
		The rest of the vertices in $B$, that is, vertices of $B_2\cup B_3$ we match with vertices from  $C\setminus(\Ima\xi'\cup\sigma(L_1))$.
		Since $C$ are the common neighbors of $B$ ($B$ and $C$ induce a complete bipartite subgraph in $G$), we only have to prove that  $C$ has enough free vertices to match to.
		
		We have to perform some counting to proceed.
		Since the leaf neighbors of $B_2$ are matched to $C$ in $\sigma$, we have
		\begin{equation*}
			|C\setminus(\Ima\sigma'\cup\sigma(L_1))|\ge|B_2|.
		\end{equation*}
		Because  $X\subset\Ima\sigma'$ and $B_3\cap\Ima\sigma'=\emptyset$, and both $X$ and $B_3$ are subsets of $\Ima\xi'$, we have
		\begin{equation}\label{eqn:help_eqn_iso}
			|\Ima\xi'\cap N_G(B)\setminus C|\ge |X|+|B_3|.
		\end{equation}
		Combining \Cref{eqn:help_eqn_iso} and $|\Ima\xi'\setminus N_G[B]|\ge a_B$ with \Cref{eqn:main_iso_counting}, we conclude that
				\begin{equation*}
			|C\cap\Ima\xi'|\le |C\cap \Ima\sigma'|-|B_3|.
		\end{equation*}
		
	Therefore, 
		\begin{equation*}
		\begin{aligned}
		|C\setminus(\Ima\xi'\cup\sigma(L_1))|&=|C\setminus\sigma(L_1)|-|C\cap\Ima\xi'|\\ &\ge| C\setminus\sigma(L_1)|-|C\cap\Ima\sigma'|+|B_3|\\&=|C\setminus(\Ima\sigma'\cup\sigma(L_1))|+|B_3|\ge|B_2|+|B_3|.
		\end{aligned}
		\end{equation*}		
		
		It means that  $C$ has enough free vertices to be matched with $B_2\cup B_3$.
		This finishes the proof of the claim.
	\end{claimproof}

	This completes the proof of the correctness of the algorithm for the right choice of the parameters.
	Note that if our guess of the parameters is incorrect, the algorithm still could find a subgraph isomorphism, then we stop. Or, if it failed, we move on to the next choice of the parameters. By \Cref{claim:runningtimeofxi},
	the running time of the subroutine for the fixed choice of parameters is $2^{\Oh(k^2)\cdot\diam(T)}\cdot\polyn$ for both correct and incorrect choices of the parameters.
Thus the total running time of the algorithm is $2^{\Oh(k^3)}\cdot 2^{\Oh(k^2)\cdot\diam(T)}\cdot\polyn$. This completes the proof of the lemma. 
\end{proof}

As a side remark, we note that  \Cref{lemma:multi_leaf_algo} already provides an \textsf{XP}-algorithm (with parameter $k$) for tree containment when a tree is of constant diameter.

\subsection{Guessing a mapping randomly}

The second key ingredient of \Cref{thm:densest_part} allows to guess images (in $G$) of leaf-adjacent vertices (in $T$) efficiently.

\begin{lemma}\label{lemma:guess_a_mapping}
	Let $G$ be a graph,   $p\ge 1$, $k\ge 3$ be integers, and  $T$ be a tree on $\delta(G)+k$ vertices such that
	\begin{enumerate}
		\item  $\delta(G)\ge k^{3p+1}$,  
		\item there are no $k^{p}$-escape vertices in $G$,  
		\item $\ld(T)<k$, and
		\item $T$ is not \separable{$k^{p}$}.
	\end{enumerate}
	Then there is a set $\mathcal{W}$ of $2^{k^{\Oh(p)}}$ leaf-adjacent vertices in $T$, satisfying the following:
	
	If $G$ has a subgraph isomorphic to $T$, then there is a vertex $u\in V(G)$ 
and a subgraph isomorphism $\sigma$ from $T$ into $G$	 such that 
with  probability at least $\frac{1}{k^{p+2}}$  a randomly selected $(k-1)$-vertex set $U\subset N_G(u)$ has
 $\sigma^{-1}(U)\subset \mathcal{W}$.
 \end{lemma}
By a randomly selected set $U$, we mean here a set formed by selecting $k-1$ times uniformly at random (without repetitions) a vertex from $N_G(u)$.

\begin{proof}
	Let $r$ be a vertex in $T$ whose deletion separates $T$ into components of size at most $|V(T)|/2$.
	Suppose that one of the components is of size at least $k^p$.
	Then the edge between $r$ and this component separates $T$ into a part of size at least $k^{p}$ and a part of size at least $|V(T)|/2> \delta(G)/2\ge k^{p}$. 
	This is a contradiction because, by the lemma's assumption, $T$ is not \separable{$k^{p}$}. Therefore, every component of $T-r$ is of size at most $ k^{p}$.
	
	In the rest of the proof, we assume that   $T$ is rooted in $r$.
	Let the children of $r$ in $T$ be $c_1,c_2,\ldots,c_d$ for $d:=\deg_T(r)$.
	For each $i\in[r]$, the subtree $T_{c_i}$ rooted in  $c_i$ has less than $k^p$ vertices.


	We first note that there are not too many non-isomorphic rooted trees among $T_{c_i}$.
	In his work \cite{Otter1948}, Otter proved that there are at most $\Oh(2.95^n)$ non-isomorphic (unrooted) trees on $n$ vertices.
	If we move from unrooted to rooted trees, the number of the classes of equivalence grows by at most the maximum size of the tree, and the bound becomes $\Oh(2.95^n\cdot n)$.
	Since each $T_{c_i}$ consists of at most $k^p$ vertices, we have to sum up this bound over each tree size from $1$ to $k^p$.
	We obtain that the upper bound on the number of non-isomorphic trees from $T-r$  is $\Oh(2.95^{k^p}\cdot k^{p}\cdot k^p)=2^{k^{\Oh(p)}}$.

	We construct the set of leaf-adjacent vertices $\mathcal{W}$ via the following process.
	For each class of isomorphic trees from $T_{c_1}, T_{c_2}, \ldots, T_{c_d}$, we take $k-1$ (or all of them, if they are less then $k-1$) representatives of this class.
	Then for each of the selected trees $T_{c_i}$, we add all leaf-adjacent vertices in $V(T_{c_i})$ (with respect to the whole $T$) to $\mathcal{W}$.
	The size of $\mathcal{W}$ is in  $2^{k^{\Oh(p)}}$.
	
	The construction of  $\mathcal{W}$ is complete and we proceed with the proof of the probability statement of the lemma.
	It is based on the following claim.
	
	\begin{claim}\label{claim:probabilityargument}
		There are at least $\frac{\delta(G)}{k^p}$ leaf-adjacent vertices (excluding $r$) in $T$.
	\end{claim}
	\begin{claimproof}
		Since $\ld(T)<k$, at most $k-1$ subtrees among $T_{c_i}$ are single-vertex trees. Each of the other subtrees $T_{c_i}$ contains at least one leaf-adjacent vertex in $T$.
		The total number of vertices in these subtrees is at least $|V(T)|-1-(k-1)=\delta(G)$.
		Since the size of each of them is at most $k^p$, at least $\frac{\delta(G)}{k^p}$ vertices in total are leaf-adjacent vertices.
	\end{claimproof}
	
	Let $\sigma: V(T)\to V(G)$ be the subgraph isomorphism from $T$ into $G$.
	Put $u:=\sigma(r)$.
	Since there are no $k^p$-escape vertices in $G$, at most $k^p$ distinct $T_{c_i}$ have at least one vertex mapped outside $N_G(u)$ by $\sigma$, i.e.$$\sigma(V(T_{c_i}))\not\subset N_G(u)$$ holds for at most $k^p$ distinct values of $i$.
	Also $\deg_G(u)<\delta(G)+k^p$, again since $u$ is not a $k^p$-escape vertex.
	Since each $|V(T_{c_i})|$ is at most  $k^p$, we have that  $|\Ima\sigma\setminus N_G[u]|\le k^{2p}$. Pipelining with \Cref{claim:probabilityargument}, we deduce that at least 
	        \begin{equation}\label{equ:many_images}
		\begin{aligned}
		\frac{\delta(G)}{k^p}-k^{2p}\ge & \frac{\delta(G)}{k^p}-\frac{\delta(G)}{k^{p+1}}\ge\frac{\delta(G)\cdot (k-1)}{k^{p+1}}\\ \ge & \frac{\delta(G)+\delta(G)}{k^{p+1}}> \frac{\delta(G)+k^p+k^{p+2}}{k^{p+1}}>\frac{\deg_G(u)}{k^{p+1}}+k
		\end{aligned}
		\end{equation}
		vertices in $N_G(u)$ are images of leaf-adjacent vertices with respect to $\sigma$.
	
	The set $U$ is formed via iterative random selection of a (not previously selected) vertex in $N_G(u)$. There are $k-1$ iterations.
	By \Cref{equ:many_images},
	in each iteration with probability at least $\frac{1}{k^{p+1}}$ the selected vertex is an image of a leaf-adjacent vertex.
	Hence, with probability at least $(\frac{1}{k^{p+1}})^{k-1}>\frac{1}{k^{p+2}}$, a random subset $U\subset N_G(u)$ consists of $k-1$ distinct images of leaf-adjacent vertices of $T$ with respect to $\sigma$.
	We refer to such a choice of $U$ as to a \emph{good choice}.
	It remains to show that for each good choice of $U$, there is some subgraph isomorphism $\sigma'$ from $T$ to $G$ with property $\sigma'(\mathcal{W})\cap U=U$.
	
	Let $U$ be a fixed good choice, i.e.\ a subset of $N_G(u)$ of size $k-1$ such that $\sigma^{-1}(U)$ consists only of leaf-ajacent vertices in $T$.
	We now modify the isomorphism $\sigma$ and obtain another isomorphism $\sigma'$ with $(\sigma')^{-1}(U)\subset\mathcal{W}$.
	
	Recall that $\mathcal{W}$ consists of a union of leaf-adjacent vertices in some subtrees of $c_1, c_2,\ldots, c_d$, i.e.\ $T_{c_i}$ for $i\in I$ for some $I\subset[d]$.
	On the other hand, $\sigma^{-1}(U)$ is a subset of a union of leaf-adjacent vertices of at most $k-1$ subtrees of $c_1, c_2, \ldots, c_d$, i.e.\ $T_{c_j}$ for $j\in J$ for some $J\subset[d]$ of size at most $k-1$.
Therefore,  if $J\subset I$, then  $\sigma$ is the required isomorphism for $U$.
	
	Otherwise, we modify $\sigma$.
	Take a $j \in J\setminus I$.
	The leaf-adjacent vertices of $T_{c_j}$ are not in $\mathcal{W}$.
	However, by the construction of $\mathcal{W}$, there is $i\in I\setminus J$ such that $T_{c_i}$ is isomorphic to $T_{c_j}$, as $\mathcal{W}$ contains leaf-adjacent vertices of at least $k$ representatives of the isomorphism-equivalence class of $T_{c_j}$.
We modify $\sigma$ by interchanging the images of $T_{c_i}$ and $T_{c_j}$.
This modification reduces the size of $J\setminus I$ by exactly one.

	We repeat such modifications until we reach $J\subset I$, and hence an isomorphism $\sigma'$ for $U$ satisfying the required condition.
	This completes the proof of the lemma.
\end{proof}

\subsection{Proof of \Cref{thm:densest_part}}

\begin{proof}[Proof of \Cref{thm:densest_part}]
	The algorithm is a Monte-Carlo algorithm that uses \Cref{lemma:multi_leaf_algo} as its subroutine.
	Since every tree $T$ is \separable{$(\diam(T)/2)$}, it is guaranteed that $\diam(T)\le \frac{k^p}{2}$, so the running time of a single call of \Cref{lemma:multi_leaf_algo} is $2^{k^{\Oh(p)}}\cdot\polyn$.
	It only remains to find a suitable mapping $\kappa$ to call \Cref{lemma:multi_leaf_algo}.
	From now on, we assume that   $G$ contains  $T$ as a subgraph, since \Cref{lemma:multi_leaf_algo} never reports false-positives.
	
	We iterate several values of $\kappa$ using \Cref{lemma:guess_a_mapping}.
	Constraints on $G$ and $T$ imposed by \Cref{thm:densest_part} allow to apply \Cref{lemma:guess_a_mapping}.
	We obtain a set $\mathcal{W}$ of leaf-adjacent vertices of $T$ and pick a random set $U$ of $k-1$ leaf-adjacent vertices in $T$ (excluding $r$).
By \Cref{lemma:guess_a_mapping}, with probability at least $k^{-p-2}$, there is a subgraph isomorphism $\sigma$ from $T$ into $G$ with $\sigma^{-1}(U)\subseteq\mathcal{W}$.
	
	To produce the mapping $\kappa$ for a single run of the subroutine of \Cref{lemma:multi_leaf_algo}, we use $\mathcal{W}$ and $U$.
	That is, $\kappa$ is a mapping from a subset of $\mathcal{W}$ into $U$, and we iterate over all $|\mathcal{W}|^{|U|}=2^{k^{\Oh(p)}}$ possible such mappings.
	For each fixed mapping, we run the subroutine of \Cref{lemma:multi_leaf_algo}.
	The total running time is bounded by $2^{k^{\Oh(p)}}\cdot\polyn$.
	If the guess of $U$ was good, at least one of the runs correctly reports that $G$ contains $T$. In this case, the algorithm reports that  $G$ contains $T$ and stops.
	
	We have shown so far that a single guess of $U$ provides a correct answer with probability at least $k^{-p-2}$ in time $2^{k^{\Oh(p)}}\cdot\polyn$.
	To amplify the probability, we repeat this procedure (guess of $U$ and iterating $\kappa$) $2k^{p+2}$ times in total.
	By the standard arguments, 
	the probability that at least one of the $2k^{p+2}$ guesses of $U$ is correct is at least $\frac{1}{2}$.
	The total running time is  $2^{k^{\Oh(p)}}\cdot\polyn$.
	
	If at every run of the subroutine the algorithm reports that $G$ does not contain $T$, the algorithm reports that $G$ does not contain $T$.
	This happens with probability $1$ if $G$ does not contain $T$ as a subgraph, and with probability at most $\frac{1}{2}$ otherwise.
	The proof is complete.
\end{proof}


\section{Large diameter and preserving paths}\label{sec:largediameter}

In this section, we show that $G$ contains $T$ if the diameter $\diam(T)$ is sufficiently large and $n$ is at least slightly above $\delta(G)$.
The proof also yields a polynomial-time algorithm that constructs a subgraph isomorphism from $T$ into $G$.
More formally, the main result of the section is the following combinatorial result.

\begin{theorem}\label{thm:large_diameter_many_vertices}
	Let  $k \ge 3$ and let $G$ be a connected graph    with at least $n \ge (1+\frac{4}{k^4})\cdot \delta(G)$ vertices and of minimum vertex degree $\delta(G)>k^{16}$. Then $G$ contains as a subgraph every 
tree $T$ on at most $ \delta(G)+k$ vertices and of diameter  $\diam(T)\ge 8k^6 \cdot \log \delta(G)$.
\end{theorem}

We start the proof of \Cref{thm:large_diameter_many_vertices} by preparing auxiliary results in   \Cref{subsec:largediameteraux,subsec:largediameterauxdva}. The final step of the proof is given in \Cref{subsec:largediameterauxfinal}.
While we state  \Cref{thm:large_diameter_many_vertices}  and all lemmata in this section as combinatorial results---under certain conditions, the graph contains a certain object---all these proofs are constructive and imply polynomial time algorithms computing the corresponding objects in polynomial time. In particular,  
the subgraph isomorphism of tree $T$ in \Cref{thm:large_diameter_many_vertices} could be constructed in polynomial time. 
 

\subsection{Preserving paths and how to use them}\label{subsec:largediameteraux}

In this subsection we define preserving sets and preserving paths. Recall that by $\ndef(v)$ we denote the 
 {neighbor deficiency} of vertex $v$, that is, $\max\{\delta(G)+k-1-\deg_G(v),0\}$. Let us also remind that for a set $S\subseteq V(G)$ and vertex $v\not\in S$, the non-neighbors of $v$ in $S$ are the vertices of $S$ that are not adjacent to $v$. 

 \begin{definition}[Preserving path]\label{def:preserving}
	For a graph $G$ and integer $k\ge 0$, we say that $S\subseteq V(G)$ is \emph{$k$-neighbor-preserving}, or simply \emph{$k$-preserving}, if each vertex $v \in V(G)\setminus S$ has at least $\ndef(v)$ non-neighbors in $S$.
	If $P$ is a path in $G$ such that $V(P)$ is preserving, we say that $P$ is a \emph{preserving path} in $G$.
\end{definition}

The following lemma shows that a $k$-preserving path in $G$ guarantees a tree on $\delta(G)+k$ vertices of large enough diameter in $G$.

\begin{lemma}\label{lemma:preserving_path_gives_solution}
	Let $G$ be a graph, $k\ge 0$ be an integer, and let $P$ be a $k$-preserving path in $G$. If $\delta(G)\ge k$, then $G$ contains every tree $T$ on $\delta(G)+k$ vertices with $\diam(T)\ge 2|V(P)|-1$.
\end{lemma}
\begin{proof}
Let $T$ be a tree satisfying the conditions of the lemma. 
	We construct an embedding of  $T$ into  $G$ by making using the preserving path $P$.
	
\begin{figure}[ht]
\centering
\scalebox{0.7}{
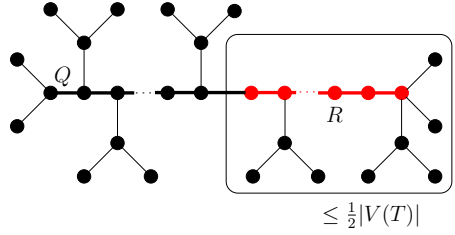}
\caption{Construction of $Q$ and $R$; $Q$ is shown by a thick line and $R$ is shown in red.}
\label{fig:R}
\end{figure}	
	
	Let $Q$ be  a path in $T$ with exactly $2|V(P)|-1$ edges. 
	Such a path exists because the diameter of $T$ is at least $2|V(P)|-1$.
	Removal of the middle edge of the path $Q$ splits $T$ (and $Q$) into two parts.
	One of these parts of $T$ consists of at most $|V(T)|/2$ vertices. We denote by $R$ the subpath of $Q$ belonging to this part of $T$. Let us remark that  $R$ is also a path in $T$ and that 
	$|N_T[V(R)]|\le |V(T)|/2+1=(\delta(G)+k)/2+1\le \delta(G)+1$.
	 (See \Cref{fig:R}.)
	
	We start constructing isomorphism $\sigma$ by mapping path $R$ into path $P$.
	Then we map all remaining neighbors of $R$ in $T$, i.e.\ $N_T(V(R))$, into the neighbors of their images in $G$.
	Since $|N_T[V(R)]|\leq \delta(G)+1$, such a mapping is always possible. 
	
	Thus, so far $\sigma$ maps   $N_T[V(R)]$ into $N_G[V(P)]$.  
For the remaining vertices of $V(T)\setminus N_T[V(R)]$ we extend $\sigma$ by repeating the following procedure. 
We pick up an edge $xy\in E(T)$ such that $x$ is already mapped and $y $ is not mapped. Let  $v:=\sigma(x)$. We claim that at least one neighbor of $v$ in $G$ is not in  $ \Ima \sigma$ yet; thus we can extend $\sigma$ by mapping $y$ into this free neighbor.   
Since all neighbors of $R$ in $T$ are already mapped, we have that $x$, which is adjacent to unmapped vertex $y$, is not in  $  V(R)$. Hence $v\notin V(P)$.  Because $P$ is a  {preserving path}, vertex $v$ has at least  $\ndef(v)$ non-neighbors in $V(P)\subseteq \Ima \sigma$. This means that $v$ is adjacent to at most  $|\Ima \sigma|-1 - \ndef(v)$ vertices occupied by $\sigma$. 
At least one vertex of $T$, namely $y$, is not mapped yet, and thus $|\Ima  \sigma| \leq |V(T)|-1=\delta(G)+k-1$.
Concluding, we have that the number of $v$ occupied by $\sigma$ is at most
 $$|\Ima \sigma|-1 - \ndef(v)\le \delta(G)+k-2 -\ndef(v)=((\delta(G)+k-1)-\ndef(v))-1\le \deg_G(v)-1,$$
and thus $v$ has  a  free neighbor. 
	
	The resulting isomorphism $\sigma$ is the subgraph isomorphism from $T$ into $G$.
\end{proof}

The rest of this subsection shows a way to construct a preserving path from a preserving set.
We start with a lemma showing how to convert a ``diameter modulator''  into a preserving set.

\begin{lemma}\label{lemma:diameter_modulator_to_preserving_path}
Let $G$ be a connected graph and 
	$k\ge 1$. If there is a vertex set  $S\subseteq V(G)$ such that $\diam(G-S)\ge 2k$ and $\delta(G)\ge |S|+k-1$, then $G$ contains a $k$-preserving path of length at most $4k-2+|S|$ in $G$.
\end{lemma}
\begin{proof}
	We start the proof with the case when $G-S$ is connected. In this case, we 
	select any two vertices in $G-S$ such that the distance between them in $G-S$ is exactly $2k$. Let $P$ be a shortest path connecting these two vertices. Because $P$ is a shortest path, each vertex  in $V(G)\setminus S$  has at most three neighbors in $P$.
	Since $|V(P)|=2k+1$, each vertex $v \in V(G)\setminus S$ has at least $(2k+1)-3=2(k-1)\ge\ndef(v)$ non-neighbors in $V(P)$.
	
	Let us remark that $P$ is not yet $k$-preserving   because some vertices of $S$ could have less than $\ndef$ non-neighbors in $V(P)$.
	%
We make $P$ $k$-preserving by inserting some vertices of $S$. More formally, if a vertex $u \in S$ has two consecutive neighbors $v,w$ in $P$, we add $u$ to $P$ by inserting it between $v$ and $w$. 
	We repeat this iteratively until either $S$ is exhausted or none of the vertices from $S$ can be inserted into $P$.
	Note that every vertex that remained in $S$ has at least $\lfloor|V(P)|/2\rfloor\ge k$ non-neighbors in $V(P)$.
	Hence, $P$ is a preserving path in $G$ of length at most $2k+|S|$.
	
\begin{figure}[ht]
\centering
\scalebox{0.7}{
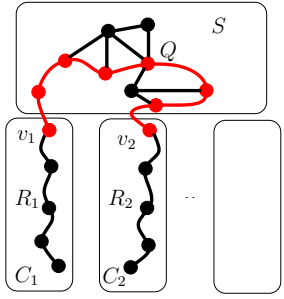}
\caption{The case when $G-S$ is not connected; $Q$ is shown in red.}
\label{fig:GS}
\end{figure}		
	
	Now consider the case when $G-S$ is not connected (see \Cref{fig:GS}). We assume that $\diam(G)<2k$, otherwise we can take $S:=\emptyset$ and proceed as in the connected case of the proof. Let $C_1, C_2, \ldots, C_p$ be  connected components of 
$G-S$, $p\ge 2$. Then  $\delta(C_i)\ge \delta(G)-|S|\ge k-1$ for each $i \in [p]$.
	
	In graph $G$, we take a shortest path $Q$ between $V(C_1)$ and $V(C_2)$. Since the diameter of $G$ is less than $2k$, we have that the length of $Q$ is at most $2k-1$.
	Let $v_1$ and $v_2$ be the endpoints of $Q$ such that  $v_i \in V(C_i)$, $i \in \{1,2\}$.
	For each $i \in \{1,2\}$, we construct a path $R_i$ of length exactly $k-1$ inside $C_i$ starting from $v_i$.
	Such path always exists in $C_i$ since $k-1\le \delta(C_i)$\footnote{A path of length $k-1$ is a tree on $k$ vertices. By \Cref{prop:chvatal-generalized}, $C_i$ contains it as a subgraph even if we fix an arbitrary starting vertex.}.
	Paths $R_1$ and $R_2$ are disjoint and for each $i \in \{1,2\}$, $R_i$ and $Q$ have only one common vertex, namely $v_i$.
	We obtain path $P$ by concatenating paths $R_1, Q, R_2$. 
	The length of $P$ is at least $2k-1$ (the lengths of $R_1$ and $R_2$ are exactly $k-1$ plus the length of $Q$ is at least one) and at most $4k-3$.
	Since $P$ contains $k$ vertices of $C_1$, it has at least $k>\ndef(v)$ non-neighbors for each $v \in V(G)\setminus ( S \cup V(C_1))=\bigcup_{i=2}^p V(C_i)$.
	Symmetrically, it contains $k$ vertices of $C_2$, so it also covers the deficiency of vertices of $C_1$.
	Hence, only vertices in $S$ can have less than $k$ non-neighbors in $P$. For such vertices, we use exactly the same trick as in the connected case:  
	Since $|V(P)|/2\ge k$, we can repeatedly insert such vertices of $S$ into $P$. 
	The resulting path $P$  is $k$-preserving and its length is at most $4k-3+|S|$.
\end{proof}

We finally show how to construct a preserving path from a preserving set in $G$ using the lemma above.

\begin{lemma}\label{lemma:preserving_set_to_preserving_path}
Let $G$ be a connected graph and 
	$k\ge 1$. If $G$ contains a  $k$-preserving set $S$ such that $\delta(G)\ge (2k-1) \cdot |S|$, then $G$ also contains a $k$-preserving path of length at most $(2k - 1) \cdot |S|$.
\end{lemma}
\begin{proof}
	The preserving path is found by joining all vertices of $S$ in some order via shortest paths.
	
	Let $s_1, s_2, \ldots, s_t$ be the vertices of $S$ in an arbitrary order, $t=|S|$.
	We construct a sequence of paths $P_1, P_2, \ldots, P_t$, such that for each $i \in [t]$, $P_i$ is a path between $s_1$ and $s_i$ with $V(P_i)\cap S=\{s_1, s_2, \ldots, s_i\}$.
	
	We start with the path $P_1$ that consists only of a single vertex $s_1$.
 	To obtain $P_{i+1}$ from $P_i$, we do the following. If $\diam(G-V(P_i-s_i)) < 2k$, we 
	concatenate $P_{i}$ with the shortest path between $s_i$ and $s_{i+1}$ in $G-V(P_i-s_i)$.
	We repeat this process $(t-1)$ times unless the condition $\diam(G-V(P_i-s_i)) < 2k$ fails.
 	If we succeed, then we have a path  $P=P_t$ of length at most $(t-1)\cdot(2k-1)$.
 	Since $S\subseteq V(P)$, $P$ is a preserving path in $G$.
 	
 	Suppose that for some $1<i<t$, we succeed to construct $P_i$ but we cannot proceed further because  $\diam(G-V(P_i-s_i))\ge 2k$.
	The length of $P_i$ is at most $(i-1)\cdot (2k-1)\le(t-2)\cdot(2k-1)$.
	Then $|V(P_i-s_i)|\le (t-2)\cdot (2k-1)$. By Lemma's assumption, we have $\delta(G)\ge (2k-1) \cdot |S|$, hence  $\delta(G)\ge |V(P_i-s_i)|+k$.
	By applying   \Cref{lemma:diameter_modulator_to_preserving_path} to $G$ and $V(P_i-s_i)$, we  obtain a preserving path of length at most $|V(P_i-s_i)|+4k-2\le (2k-1)\cdot (|S|-2)+4k-2=(2k-1)\cdot|S|$ in $G$.
\end{proof}

\subsection{Finding preserving sets of order $\log\delta(G)$}\label{subsec:largediameterauxdva} 

We first need the technical lemma about the properties of vertices with degrees below some threshold in sufficiently large graphs. Informally, the lemma says that for the set $A$ of vertices of   degrees at most $(1+\epsilon)\cdot\delta(G)$, it is possible to select sufficiently small vertex set $S$ such that no vertex of $A$ dominates all vertices of $S$. 

\begin{lemma}\label{lemma:anti_dom_set_existence}
	Let $G$ be an $n$-vertex graph with $\delta(G)\ge 2$.
	Let $\epsilon \in (0, 1)$ be a given number such that $n \ge (1+\epsilon)^2 \cdot \delta(G)$.
	Then there exists a set $S\subseteq V(G)$ such that $S$ has at least one non-neighbor for each vertex in $A$ and $|S|< 4\log \delta(G) / \log (1+\epsilon) + 1$, where $A$ is the set of vertices in $G$ of degree less than $(1+\epsilon)\cdot\delta(G)$.
\end{lemma}
\begin{proof}

	We construct the set $S$ starting from $S=\emptyset$. Fix arbitrary $v \in A$.
	By choosing a vertex $u\in V(G)$ uniformly at random, the probability that  $u$ is adjacent to  $v$ is $$\frac{\deg_G(v)}{n}<\frac{(1+\epsilon)\delta(G)}{n}.$$
	
	Hence, the expected  number of vertices in $A$ adjacent to a random $u \in V(G)$ is at most 
	\begin{equation*}\label{equ:exp_val_not_covered}
		(1+\epsilon)\cdot \frac{\delta(G)}{n} \cdot |A|.
	\end{equation*}
Therefore,  there exists $u_1\in V(G)$ such that all vertices of $A$ except at most $(1+\epsilon)\cdot \frac{\delta(G)}{n} \cdot |A|$ 
vertices are non-neighbors of $s_1$.
	We add $s_1$ to $S$.
	Let $A_1\subseteq A$ be the set of vertices that are adjacent to  $s_1$.
	Now apply the arguments above to $A_1$ instead of $A$.
	There is a vertex $s_2 \in V(G)$ that is non-adjacent to all but at most $$(1+\epsilon)\cdot \frac{\delta(G)}{n} \cdot |A_1|$$ vertices of $A_1$.
	Then $A_2$ is defined as a subset of vertices of $A_1$ that are adjacent to  $s_2$. We add $s_2$ to $S$. 
	
	By repeating this process until we arrive at $A_t=\emptyset$,  we obtain the family $\{A=A_0, A_1, \dots, A_t\}$ of subsets of $A$.
	For every $v\in A$, there is $i\in \{0,\dots, t-1\}$, such that $v\in A_{i}\setminus A_{i+1}$, and thus $v$ is non-adjacent to $s_i$.
	

	It remains to show the upper bound on the size of $S$.
	Observe that
	\begin{equation*}
		|A_i|\le (1+\epsilon)^i \cdot \left(\frac{\delta(G)}{n}\right)^i \cdot |A|.
	\end{equation*}

	To show that $A_t=\emptyset$ for $t:=\lceil4\log\delta(G)/\log(1+\epsilon)\rceil$, it is enough to show that
	\begin{equation*}
		 (1+\epsilon)^t \cdot \left(\frac{\delta(G)}{n}\right)^t < \frac{1}{n},
	\end{equation*}
	or, equivalently,
	\begin{equation}\label{eqn:preserving_set_main}
		(1+\epsilon)^{-t} \cdot \left(\frac{n}{\delta(G)}\right)^{t}> n.
	\end{equation}

	Since $t < 4\log \delta(G)/\log (1+\epsilon) + 1$, it follows that $(1+\epsilon)^t<\delta(G)^4\cdot (1+\epsilon)<\delta(G)^5$.
	It also holds that $t\ge 5$ since $\delta(G)>1+\epsilon$.
	
	We split the proof of \Cref{eqn:preserving_set_main} in two cases.
	The first case is when $n\ge \delta(G)^4$. Then 
	\begin{equation*}
		(1+\epsilon)^{-t} \cdot \left(\frac{n}{\delta(G)}\right)^{t}>
		\frac{1}{\delta(G)^5} \cdot \left(n^{3/4}\right)^{5}>\frac{1}{n^{5/4}}\cdot n^{15/4}>n,
	\end{equation*}
	and \Cref{eqn:preserving_set_main} follows.
	For  $n<\delta(G)^4$, we use the condition of the lemma that $n > (1+\epsilon)^2\cdot\delta(G)$. Then
	\begin{equation*}
		(1+\epsilon)^{-t} \cdot \left(\frac{n}{\delta(G)}\right)^{t}>(1+\epsilon)^{-t} \cdot \left(1+\epsilon\right)^{2t}=(1+\epsilon)^t>(1+\epsilon)^{4\log\delta(G)/\log(1+\epsilon)}=\delta(G)^4>n.
	\end{equation*}

	This completes the proof of \Cref{eqn:preserving_set_main}. Hence, $|A_t|<\frac{1}{n}\cdot|A|<1$ and therefore, $|S|\leq 
	\lceil4\log\delta(G)/\log(1+\epsilon)\rceil$. 
	
\end{proof}


We now use  \Cref{lemma:anti_dom_set_existence} for extracting a $k$-preserving set from  $G$.

\begin{lemma}\label{lemma:there_exists_preserving_set}
	Let $G$ be a $n$-vertex graph and $k,p>1$ be two integers.
	Let $q=4k^{p}\cdot\log\delta(G)$.
	If $n\ge (1+\frac{3}{k^p})\cdot \delta(G)+qk$ and $\delta(G)\ge qk\cdot (k^p+1)$,
	then $G$ contains a $k$-preserving set $S$ of size at most $qk$.
\end{lemma}
\begin{proof}
	Throughout the proof, we employ $(q+1)\cdot(k-1)=qk+((k-1)-q)< qk-k$.
	
	Put $\epsilon := \frac{1}{k^p}$.
	Note that $\log(1+\frac{1}{k^p})< \frac{1}{k^p}$ for $k>1$.
	Denote by $B$ the set of all vertices with non-zero deficiency in $G$, that is,  $B:=\{v \in V(G): \ndef(v) > 0\}$.
	Recall that each vertex from $B$ that is not the preserving set $S$ should have at least  $\ndef(v)=\delta(G)+k-1-\deg_G(v)\le k-1$ non-neighbor vertices in  $S$.
	
	\begin{claim}\label{claim:random_choice_preserving}
		Let $X\subseteq V(G)$ be such that $|X|\le qk-k$.
		Then for every $v \in B$, $\deg_{G-X}(v)<(1+\epsilon)\cdot \delta(G-X)$. Also $n-|X|\ge (1+\epsilon)^2\cdot \delta(G-X)$.
	\end{claim}
	\begin{claimproof}
		Note that $$\delta(G-X)\ge \delta(G)-|X|\ge qk\cdot(k^p+1)-|X|\ge (|X|+k)\cdot(k^p+1)-|X|> (|X|+k)\cdot k^p.$$
		
		Let $v\in B$. Since $\ndef(v)>0$, $\deg_G(v)<\delta(G)+k$. Observe
		\begin{equation*}
		\begin{aligned}
		\deg_{G-X}(v)\le &\deg_G(v)< \delta(G)+k\le \delta(G-X)+|X|+k\\ \le& \delta(G-X)+\delta(G-X)/k^p= (1+\epsilon)\cdot\delta(G-X).
		\end{aligned}
		\end{equation*}
		
		For the last inequality of the claim, note that $(1+\epsilon)^2=1+2\epsilon+\epsilon^2< (1+3\epsilon)$ for $\epsilon < 1$.
		Then
		\begin{equation*}
		\begin{aligned}
		n\ge&\left(1+\frac{3}{k^p}\right)\cdot\delta(G)+qk>\left(1+3\epsilon\right)\cdot\delta(G)+|X|\\ >& (1+\epsilon)^2\cdot\delta(G)+|X|\ge (1+\epsilon)^2\cdot\delta(G-X)+|X|.
		\end{aligned}
		\end{equation*}
	\end{claimproof}

 	\Cref{claim:random_choice_preserving} (with $X:=\emptyset$) allows to apply \Cref{lemma:anti_dom_set_existence} to $G$ and $\epsilon$.
	By \Cref{claim:random_choice_preserving}, the application of  \Cref{lemma:anti_dom_set_existence} gives a set $S_1$ that contains at least one  non-neighbor for each vertex in $B$.
	Also $|S_1|\le 4k^p\cdot\log\delta(G)+1=q+1$.
	
	Repeat the application of \Cref{lemma:anti_dom_set_existence}, but now apply it to $G-S_1$ and $\epsilon$.
	The obtained set $S_2$ is of size at most $q+1$ and has at least one non-neighbor for each vertex in $B\setminus S_2$.
	This application is legitimate by \Cref{claim:random_choice_preserving} (with $X:=S_1$).
	
	Repeat this $k-3$ more times to obtain sets $S_3, S_4, \ldots, S_{k-1}$.
	Each time, the set $S_{i+1}$ is obtained from \Cref{lemma:anti_dom_set_existence} applied to $G-(S_1\cup\ldots\cup S_i)$ and $\epsilon$.
	This is legitimate by \Cref{claim:random_choice_preserving} (with $X$ of size at most $i\cdot (q+1)<qk-k$).
	So $S_{i+1}$ has at least one non-neighbor for each vertex in $B\setminus(S_1\cup\ldots\cup S_i)$.
	
	Put $S:=\bigcup_{i=1}^{k-1} S_i$.
	Clearly, $S$ has at least $k-1$ non-neighbors for each vertex in $B\setminus S$.
	It follows that $S$ is the required preserving set in $G$.
\end{proof}

\subsection{Proof of \Cref{thm:large_diameter_many_vertices}}\label{subsec:largediameterauxfinal} 

Now everything is ready to proceed with the 
proof of \Cref{thm:large_diameter_many_vertices}.
Generally speaking, we show that there exists a preserving set in $G$ by \Cref{lemma:there_exists_preserving_set} and then transform it into a preserving path in $G$ using \Cref{lemma:preserving_set_to_preserving_path}.
Then we show that the diameter of $T$ is large enough to apply \Cref{lemma:preserving_path_gives_solution}.

\begin{proof}[Proof of \Cref{thm:large_diameter_many_vertices}]
Let  $k \ge 3$ and let $G$ be a connected graph    with at least $n \ge (1+\frac{4}{k^4})\cdot \delta(G)$ vertices and of minimum vertex degree $\delta(G)>k^{16}$.  We want to show that $G$ contains as a subgraph every 
tree $T$ on at most $ \delta(G)+k$ vertices and of diameter  $\diam(T)\ge 8k^6 \cdot \log \delta(G)$.

In order to apply  \Cref{lemma:there_exists_preserving_set}, we have to show that $G$ satisfies its conditions for $p:=4$.
	\begin{claim}\label{claim:proofoftheoremdiamter}
		\begin{enumerate}
			\item $\delta(G)\ge 4k^5\cdot\log \delta(G)\cdot (k^4+1)$;
			\item $n\ge (1+\frac{3}{k^4})\cdot \delta(G)+4k^5\log \delta(G)$.
		\end{enumerate}
	\end{claim}
	\begin{claimproof}
		Since $k\ge 3$, we have that $\delta(G)> 3^{16}> 2^{25}$.
		Then $\log\delta(G)<\delta(G)^{3/16}$.		
		Obtain 
		$$4k^5 \cdot \log\delta(G)\cdot (k^4+1)< \delta^{2/16}\cdot\delta(G)^{5/16}\cdot\delta(G)^{3/16}\cdot(\delta(G)^{4/16}+1)<\delta(G)^{10/16}\cdot \delta(G)^{5/16}<\delta(G).$$
		The first part of the claim is proved.
		
	 To see that the second part of the claim holds, from the first part obtain
		$4k^5 \cdot \log\delta(G)<\delta(G)/k^4.$
		Then, by the constraint imposed by the statement of \Cref{thm:large_diameter_many_vertices},
		$$n\ge \left(1+\frac{4}{k^4}\right)\cdot \delta(G)= \left(1+\frac{3}{k^4}\right)\cdot \delta(G) + \frac{\delta(G)}{k^4}>\left(1+\frac{3}{k^4}\right)\cdot \delta(G)+4k^5\log\delta(G).$$
		The claim is proved.
	\end{claimproof}

	By \Cref{claim:proofoftheoremdiamter}, we can apply \Cref{lemma:there_exists_preserving_set} to $G$ and $k$ with $p:=4$.
	The size of the obtained preserving set $S$ is   $|S|\le 4k^{5}\cdot\log \delta(G)$.
	By \Cref{claim:proofoftheoremdiamter}, $\delta(G)> 2k \cdot |S|$.  \Cref{lemma:preserving_set_to_preserving_path} implies that for  $G$ 
	has  a $k$-preserving path $P$ of length at most
	 $(2k-1)\cdot |S|< 8k^6 \cdot \log\delta(G)$.
	
	By \Cref{lemma:preserving_path_gives_solution}, $G$ contains every tree $T$ with $|V(T)|\le \delta(G)+k$ and $\diam(T)\ge 8k^6 \cdot \log\delta(G)$.
	This completes the proof of the main result of the section.
\end{proof}

\section{When $G$ has at most $(1+\varepsilon)\delta(G)$ vertices.
}\label{sec:diracs}

In this section, we consider graphs with a sufficiently small number of vertices. We show that such graphs contain all trees of size $\delta(G)+k$ whose 
  maximum leaf-degree $\ld(T)$  is less than $k$. Let us remind that by the maximum leaf-degree of $T$, we mean the maximum number of leaf-neighbors a vertex of $T$ could have.  More formally, the main result of this section is the following theorem.

\begin{theorem}\label{lem:dense}
	Let $G$ be a graph and let $k\geq 1$ be an integer such that $\delta(G)+k\leq |V(G)|\leq (1+\varepsilon)\delta(G)$ for $\varepsilon\leq \frac{1}{4k}$ and $\delta(G)\geq 12k^2$. Let also $T$ be a tree with at most $\delta(G)+k$ vertices such that $\ld(T)<k$. Then  $G$ contains $T$ as a subgraph.
\end{theorem}

The following combinatorial property of trees is useful for us. We remind that a \emph{hitting} set for a family of sets $\mathcal{S}$ is a set that contains at least one representative from each set of $\mathcal{S}$.

\begin{lemma}\label{lem:deg-two}
Let $T$ be a tree with  $\ell$ leaves. Then any hitting set for the family of the neighborhoods of its vertices, that is, for the  family $\mathcal{S}=\{N_T(v)\mid v\in V(T)\}$,  is of size at least 
$\frac{1}{2}(|V(T)|-3\ell+6)$. 
\end{lemma}

\begin{proof}
Let $V_1$, $V_2$, and $V_{\geq 3}$ be the sets of vertices of degree one, two, and at least three, respectively. 
Denote $n=|V(T)|=|\mathcal{S}|$, $n_1=|V_1|=\ell$, $n_2=|V_2|$, and $n_{\geq 3}=|V_{\geq 3}|$. 
We have that 
\begin{equation}\label{eq:degs}
\begin{aligned}
2n_1+2n_2+2n_{\geq 3}-2=&2(n-1)=\sum_{v\in V(T)}\deg_T(v)\\=&\sum_{v\in V_1}\deg_T(v)+\sum_{v\in V_2}\deg_T(v)+\sum_{v\in V_{\geq 3}}\deg_T(v)\\=&n_1+2n_2+\sum_{v\in V_{\geq 3}}\deg_T(v).
\end{aligned}
\end{equation}
In particular, \Cref{eq:degs} implies that $n_1+2n_{\geq 3}-2=\sum_{v\in V_{\geq 3}}\deg_T(v)\geq 3n_{\geq 3}$ and, therefore, $n_{\geq 3}\leq n_1-2$. 
Then $\sum_{v\in V_{\geq 3}}\deg_T(v)\leq  3n_1-6$. Note that a vertex $v\in V_{\geq 3}$ is included in $\deg_T(v)$ sets of $\mathcal{S}$. Hence, the vertices of $V_{\geq 3}$ hit at most  $\sum_{v\in V_{\geq 3}}\deg_T(v)$ sets of $\mathcal{S}$. Any vertex $v\in V_1\cup V_2$ is included in at most two sets of $\mathcal{S}$. Therefore, any hitting set of $\mathcal{S}$ contains at least 
$\frac{1}{2}(n-\sum_{v\in V_{\geq 3}}\deg_T(v))\geq \frac{1}{2}(n-3n_1+6)=\frac{1}{2}(|V(T)|-3\ell+6)$ vertices of $V_1\cup V_2$. This concludes the proof.
\end{proof}

\subsection{Proof of \Cref{lem:dense}}
We proceed with the proof of the main result of the section.

\begin{proof}[Proof of \Cref{lem:dense}]
The proof is by induction on the number of vertices of $T$. If $|V(T)|\leq\delta(G)+1$, then $T$ is a subgraph of $G$ by \Cref{chvatal-theorem}. Assume that $T$ has $\delta(G)+p$ vertices for $2\leq p\leq k$ and $G$ contains any tree with $\delta(G)+p-1$ vertices and the maximum leaf-degree at most $k-1$ as a subgraph. Observe that every vertex $x$ of $G$ has at most $|V(G)|-1-\delta(G)\leq \varepsilon \delta(G)-1$ non-neighbors in $G$. 

We argue that $T$ has a leaf $u$ such that $\ld(T-u)\leq \ld(T)$. 
Let us note that deleting a leaf $u$ could increase the value $\ld$ only if the deletion of $u$ turns its neighbor into a leaf. Hence, 
if there is a leaf $u$ such that its neighbor's degree is at least three, then $\ld(T-u)\leq \ld(T)$. 

Now we are in the situation where every leaf of $T$ is adjacent to a vertex of degree two. If there is a leaf $u$ with a neighbor $v$,   such that the neighbor $w$ of $v$, $w\neq u$,  is not adjacent to any leaf of $T$, then  $\ld(T-u)\leq \ld(T)$. If $w$ is adjacent to a leaf, then $\deg_T(w)=2$ by the assumption that each leaf is adjacent to a vertex of degree two. This means that $T$ is the path on four vertices. But this cannot happen because  $|V(T)|\geq \delta+2>4$. Hence there is a leaf $u$ such that $\ld(T-u)\leq \ld(T)$.

Consider an arbitrary leaf $u$ of $T$ such that $\ld(T-u)\leq \ld(T)\leq k-1$  
and let  $v$ be its unique neighbor.  Let $T'=T-u$. By the inductive assumption, $G$ contains $T'$ as a subgraph and we can assume that $V(T')\subseteq V(G)$ and $E(T')\subseteq E(G)$. If $v$ is adjacent to a vertex of $V(G)\setminus V(T')$ in $G$, then we obtain that $T$ is a subgraph of $G$. Assume that this is not the case. Then $N_G(v)\subseteq V(T')$ and $v$ has  at most $|V(T)|-1-\delta(G)=p-1\leq k-1$ non-neighbors among the vertices of $T'$. 

Suppose that $T$ has at least $(\delta(G)\varepsilon+k)(k-1)$ leaves. 
Let $U$ be the set of vertices of $T$ that are adjacent to at least one leaf and let $W=U\setminus\{v\}$. 
Because $\ld(T)\leq k-1$, $|U|\geq \delta(G)\varepsilon+k$. For every vertex of $x\in W$, we choose a leaf $\ell_x$ of $T$ adjacent to $x$ in $T'$ and define $L=\{\ell_x\mid x\in W\}$. 
Because $|L|=|U|-1\geq \delta(G)\varepsilon+k-1$ and $v$ has at most $k-1$ non-neighbors among the vertices of $V(T')$, $v$ is a adjacent to at least $|L|-k+1\geq \varepsilon\delta(G)$ vertices of $L$. Let $L'\subseteq L$ be the 
set of vertices adjacent to $v$ and consider $W'=\{x\in W\mid \ell_x\in L'\}$.   
Because $|V(G)|\geq \delta(G)+k>|V(T')|$, there is $w\in V(G)\setminus V(T')$. Since $|W'|=|L'|\geq \varepsilon \delta(G)$ and $w$ has at most $\varepsilon \delta(G)-1$ non-neighbors in $G$, we have that $w$ has a neighbor $x$ in $W'$. Thus $v$ and $x$ are adjacent to $\ell_x$ and $x$ is adjacent to $w$. Let $T''$ be the subgraph of $G$ with $V(T'')=V(T')\cup \{w\}$ and $E(T'')=(E(T')\setminus \{\ell_x\})\cup \{v\ell_x,xw\}$.  Note that trees  $T''$ and  $T$ are isomorphic: the leaf $\ell_x$ of $x$ is remapped to $w$ and the leaf $u$ of $T$ is mapped to $\ell_x$.  Hence  if  $T$ has at least $(\delta(G)\varepsilon+k)(k-1)$ leaves, then $G$ contains $T$. 

From now on, we assume that the number of leaves of $T$ is less than $(\delta(G)\varepsilon+k)(k-1)$. Consider the family $\mathcal{S}=\{N_{T'}(x)\mid x\in V(T')\}$ of the neighborhoods of the vertices of $T'$
and let $\mathcal{S}'=\{N_{T'}(x)\mid x\in V(T')\setminus\{v\}\text{ s.t. }vx\in E(G)\}$. Because $v$ has at most $p-1$ non-neighbors in $V(T')$, we have that
$|\mathcal{S}'|\geq |V(T')|-p=\delta(G)-1$.
By Lemma~\ref{lem:deg-two}, any hitting set for $\mathcal{S}$ is of size at least 
\begin{equation}\label{eq:hs}
\begin{aligned}
\frac{1}{2}(|V(T')|-3(\delta(G)\varepsilon+k)(k-1)+6)=&\frac{1}{2}(\delta(G)+p-1-3(\delta(G)\varepsilon+k)(k-1)+6)\\
\geq& \frac{1}{2}((1-3\varepsilon(k-1))\delta(G)-3k(k-1)).
\end{aligned}
\end{equation}
For a vertex  $w\in V(G)\setminus V(T')$, we consider the set 
$X=\{x\in V(T)\mid xw\notin E(G)\}$ of its non-neighbors in $V(T')$. 
Because $w$ has at most  $\varepsilon \delta(G)-1$ non-neighbors in $G$, by \Cref{eq:hs},  $X$ does not hit at least 
\begin{equation*}
\frac{1}{2}((1-3\varepsilon(k-1))\delta(G)-3k(k-1))-\varepsilon \delta(G)+1\geq \frac{1}{2}((1-3\varepsilon k)\delta(G)-3k(k-1))
\end{equation*}
sets of $\mathcal{S}$. Because $|\mathcal{S}|=|V(T')|\leq \delta(G)+k-1$ and $|\mathcal{S}'|\geq \delta(G)-1$, 
$X$ does not hit at least
\begin{equation*}\label{qe:hsprime}
\frac{1}{2}((1-3\varepsilon k)\delta(G)-3k(k-1))-k=\frac{1}{2}((1-3\varepsilon k)\delta(G)-k(3k-1))
\end{equation*}
sets of $\mathcal{S}'$. 
Since $\varepsilon\leq \frac{1}{4k}$ and $\delta(G)\geq 12k^2$, we have that 
$(1-3\varepsilon k)\delta(G)-k(3k-1)>0$. Therefore
$X$ does not hit at least one set of $\mathcal{S}'$. Thus, there is $x\in V(T')\setminus \{v\}$ such that $xv\in E(G)$ and for every $y\in N_{T'}(x)$, $yw\in E(G)$. 
We construct the subgraph $T''$ of $G$ from $T'$ by defining $V(T'')=V(T')\cup\{w\}$ and $E(T'')=(E(T')\setminus \{xy\mid y\in N_{T'}(x)\} )\cup\{wy\mid y\in N_{T'}(x)\} \cup\{vx\}$. Observe that $T''$ is a tree isomorphic to $T$ where $x\in V(T')$ is remapped to $w$ and $u$ is mapped to $x$. Thus, $G$ contains $T$ as a subgraph. This concludes the proof.
\end{proof}


\section{Medium diameter and escape vertices}\label{sec:medium}
In this section we prove a combinatorial result (\Cref{thm:medium_diameter_escape_vertices}) about containment in $G$ trees on  $\delta(G)+k$ vertices and of ``medium'' diameter $k^{\Oh(1)}\cdot\log\delta(G)$. 
Informally, \Cref{thm:medium_diameter_escape_vertices} ensures that  $G$  contains any such medium-diameter tree $T$ if its diameter is at least   $k^{\Omega(1)}$. 
  Moreover, \Cref{thm:medium_diameter_escape_vertices} collects two other cases when $G$ contains a tree of diameter $k^{\Oh(1)}\cdot\log\delta(G)$. One case is when $G$ 
   has a $k^{\Oh(1)}$-escape vertex. (Recall that a vertex $v\in V(G)$ is {$q$-escape}, if either its degree is at least $\delta(G)+q$ or there is a matching of size $q$  between $N[v]$ and $V(G)\setminus N[v]$.) The other case is when $T$ is \separable{$k^{\Oh(1)}$}, that is,  there is an edge in $T$ whose removal separates $T$ into two parts consisting of at least $k^{\Oh(1)}$ vertices.

 More formally, the main result of this section is the following theorem.

\begin{theorem}\label{thm:medium_diameter_escape_vertices}
	Let $G$ be a connected graph and $T$ be a tree on $\delta(G)+k$ vertices for $k\ge 3$.
	If $|V(G)| \ge \delta(G)+2k^{14}$, $\ld(T)<k$, $\delta(G)\ge k^{17}$, $\diam(T)\le 8k^6 \cdot \log \delta(G)$ and either
	\begin{itemize}
		\item $\diam(T)\ge 2k^{11}$, or
		\item there is a $4k^{13}$-escape vertex in $G$, or
		\item $T$ is \separable{$2k^{14}$},
	\end{itemize}
	then $G$ contains $T$ as a subgraph.
\end{theorem}

We organize the proof of   \Cref{thm:medium_diameter_escape_vertices}  in several stages. 
\subsection{Contracting trivial paths}\label{subsection:contracting-trivial-paths}

We start with trees of diameter   $k^{\Omega(1)}$.
The key property in this case is the following: Whenever we consider a minimal subtree connecting some set of (sufficiently distant) $k$ leaves, the subtree always contains long \emph{trivial paths}, defined below.
\begin{definition}[Trivial path]
	An $(s,t)$-path $P$ in a tree $T$ is a \emph{trivial path}, if  each inner vertex $v \in V(P)\setminus\{s,t\}$ is of degree two, i.e., $\deg_T(v)=2$.
	Additionally, if $\deg_T(s)\neq 2$ and $\deg_T(t)\neq 2$, we say that $P$ is a \emph{maximal trivial path} in $T$.
\end{definition}

We will use long trivial paths to embed non-neighbors of a fixed set of vertices and then extend such embeddings by making use of  \Cref{lemma:multi_saved_k_neighbours}.
We prove the following.

\begin{lemma}\label{lemma:large_diameter_trivial_paths}
	Let $G$ be a connected graph and $T$ be a tree on $\delta(G)+k$ vertices for $k\ge 3$.
	If $\ld(T)<k$, $\delta(G)\ge 2k\cdot \diam(T)$, and $\diam(T)\ge 2k^4$, then $G$ contains $T$ as a subgraph.
\end{lemma}
\begin{proof}
	Note that $T$ has at least $k-1$ leaves since $|V(T)|>\delta(G)> (k - 1) \cdot \diam(T)$, by \Cref{lemma:leaves_diameter}.
	Construct a set $L$ of $k-1$ leaves of $T$ as following. We first put in $L$ a diametral pair of leaves of $T$, that is, two leaves such that the distance between these leaves is exactly $\diam(T)$. Then we extend 
 $L$ by adding arbitrary $k-3\ge 0$ leaves of $T$.
	
	Let $W:=N_T(L)$ be the set of neighbors of these leaves in $T$, $|W|\le k-1$.
	Let $T_W$ be the minimal subtree of $T$ containing all vertices of $W$.
	We have that $\diam(T_W)=\diam(T)-2$ and $T_W$ has at most $k-1$ leaves.
	
\begin{figure}[ht]
\centering
\scalebox{0.7}{
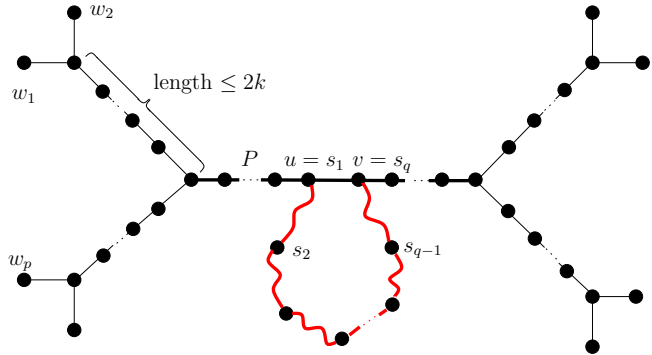}
\caption{The tree $T_W'$ for $W=\{w_1,\ldots,w_p\}$ and the extension of $P$ via $s_1,\ldots,s_q$; $P$ is shown by a thick line and the paths that are used for extending $P$ are shown in red. For the sake of illustration, we identify the respective vertices of $T$ and $G$ under the isomporphism.}
\label{fig:W}
\end{figure}	
	
Now obtain a tree $T'_W$ from $T_W$ by the following procedure (see also~\Cref{fig:W}).
For each maximal trivial path $P$ in $T_W$ of length more than $2k$, we contract some of its inner edges such that the resulting path is of length exactly $2k$.
	It is important that the endpoints of $P$ are not changed by the contractions.
	Note that $T'_W$ has the same number of leaves as $T_W$ (their set is exactly $W$).
	The endpoints of the maximal trivial paths are preserved in $T'_W$ as well.
	If $\ell$ is the number of leaves in $T'_W$, then $T'_W$ contains less than $2(\ell-1)$ maximal trivial paths.
	Since each edge of $T'_W$ belongs to exactly one maximal trivial path, $|E(T'_W)|< 2k \cdot 2(k-2)=4k^2-8k$.
	Hence, $T'_W$ has at most $4k^2-8k$ vertices.
	
	Since $\delta(G)\ge 4k^2> |V(T'_W)|$, by Chv{\'{a}}tal's Lemma (\Cref{chvatal-theorem}), there is an isomorphism from $T'_W$ into a subgraph of $G$.
	Denote this isomorphism by $\sigma'$.
	Our goal is to extend $\sigma'$ to an isomorphism $\sigma$ of $T_W$ into a subgraph of $G$, ensuring that $\Ima \sigma$ has enough non-neighbors of the images of vertices in $W$ to fit  \Cref{lemma:multi_saved_k_neighbours}.
	
	Construct the set $S$ of desired non-neighbors by picking arbitrary $\ndef(\sigma'(w))\le k-1$ non-neighbors in $G$ for each $w\in W$.
	After that, remove all vertices of $\Ima \sigma'$ from $S$, since they are already used by the isomorphism.
	The set $S$ consists of at most $(k-1)\cdot |W|\le (k-1)^2$ vertices of $G$.
	To embed $S$ in  the image of $\sigma'$, we need to show that there is a sufficiently long trivial path in $T_W$.
	All vertices of $S$ will be images of the vertices of this path.
	
	\begin{claim}
	    \label{claim:long_path}
		There is a maximal trivial path in $T_W$ of length at least $2k^3+4k^2$.
	\end{claim}
	\begin{claimproof}
		Let $D$ be the path between a diametral pair of leaves in $T_W$.
		Since $T_W$ has at most $k-1$ leaves, $D$ has at most $(k-1)-2=k-3$ inner vertices of degree at least $3$. Then $D$ 
		consists of at most $k-2$ maximal trivial paths and at least one of these paths is of length at least
		$$\frac{\diam(T_W)}{k-2}\ge \frac{\diam(T)-2}{k-2}\ge \frac{2k^4-2}{k-2}=\frac{2k^4-4k^3+4k^3-2}{k-2}= 2k^3+\frac{4k^3-2}{k-2}>2k^3+4k^2.$$
	\end{claimproof}

	This maximal trivial path is contracted in $T'_W$ and has length $2k$, denote it by $P$.
	Therefore at least $(2k^3+4k^2)-2k>2k^3$ edges are contracted.
	Let $uv$ be an edge of $P$.
	We will replace this edge with a longer $(u,v)$-path (see~\Cref{fig:W}), which will fall under the ``budget'' of $2k^3$ edges.
	Its image in $G$ will be a $(\sigma'(u), \sigma'(v))$-path containing all vertices in $S$.
	Thus, for convenience we consider $\sigma'(u)$ and $\sigma'(v)$ to be a part of $S$ too.

	Let vertices in $S$ be $s_1, s_2, \ldots, s_q$, such that $s_1=\sigma'(u)$ and $s_q=\sigma'(v)$, and $q \le (k-1)^2+2 \le k^2$.
	We first find a path in $G$ that contains all vertices of $S$.
	Additionally, this path should avoid all vertices of $I=\Ima\sigma' \setminus \{\sigma'(u), \sigma'(v)\}$. Let us note that   $|I|=|V(T'_W)| - 2 \le 4k^2-8k$.
	The construction is  similar to the one in the proof of \Cref{lemma:preserving_set_to_preserving_path}: we connect vertices $s_1, \ldots, s_q$ in order via shortest paths avoiding vertices that are already used (including vertices in $I$ and $S$), see \Cref{fig:W}.

We proceed with constructing such a path inductively. Let $U_1=I\cup \{s_3 \ldots, s_q\}$. If the distance between $s_1$ and $ s_2$ in  $G-U_1$ is at most $2k+1$, we find
 a shortest $(s_1, s_2)$-path in $G-U_1$ and denote it by $Q_1$.   
 For each $i \in [2, q - 1]$, let $U_i=U_{i-1}\cup V(Q_{i-1})\setminus \{s_{i+1}\}$. If the distance between $s_i$ and $ s_{i+1}$ in  $G-U_i$ is at most $2k+1$, we define  $Q_i$ as the $(s_1, s_{i+1})$-path obtained by appending to $Q_{i-1}$  a shortest  $(s_i, s_{i+1})$-path in  $G-U_i$.
%
%
%
	This construction is not always possible, since at some moment an $(s_i,s_{i+1})$-path might either not exist or it might be too long.
	
	Let $i$ be the smallest index such that there is no $(s_i, s_{i+1})$-path of length at most $2k+1$ in $G-U_i$ for $i\le q-1$.
	We have that 
	
	\begin{equation*}
	\begin{aligned}
		|U_i|\le & |I|+|V(Q_i)\cup S|\le (4k^2-8k)+(2k \cdot i+q)\\\le & (4k^2-8k)+(2k\cdot (q-1)+q)=(4k^2-8k)+q\cdot(2k+1)\\\le &(4k^2-8k)+k^2\cdot(2k+1)=2k^3+5k^2-8k\\< &2k^3+2k^3-8k=4k^3-8k.
	\end{aligned}
	\end{equation*}
	Then we can apply  \Cref{lemma:diameter_modulator_to_preserving_path}, since $\delta(G)>4k^5>|U_i|+(k-1)$ and $\diam(G-U_i)>2k$. Thus 
    there is a $k$-preserving path $P'$ of length at most $4k-2+|U_i|$ in $G$. Recall that a path $P'$ in $G$ is $k$-preserving if any vertex $v$ not on $P'$ has at least $\ndef(v)$ non-neighbors on $P'$ (\Cref{def:preserving}).
	Observe that the number of vertices in $P'$ is at most $$4k-2+|U_i|+1\le 4k-1+4k^3-8k<4k^3-4k< \frac{\diam(T)}{2},$$
	and hence  by \Cref{lemma:preserving_path_gives_solution}, $G$ contains $T$.
	
	Therefore, in the rest of the proof we can assume that the $(s_1, s_q)$-path $Q_{q-1}$ is constructed successfully and that its length is at most 
	\begin{equation*}
	\begin{aligned}
	(2k+1)\cdot (q-1)< & (2k+1)\cdot ((k-1)^2+1)=(2k+1)\cdot(k^2-2k+2)\\=&(2k^3-4k^2+4k)+(k^2-2k+2)\\=&2k^3-(3k^2-2k-2)<2k^3-(3k^2-3k)<2k^3.
	\end{aligned}
	\end{equation*}
	Recall that at least $2k^3$ edges were contracted to obtain the path $P$ from the respective maximal trivial path, thus the bound above shows that enough ``space'' was left during the contraction for the extended isomorphism.
	
	We now transform the isomorphism $\sigma'$ using the obtained path $Q_{q-1}$.
	First, we transform $T'_W$ into a new tree $T''_W$ by replacing $uv \in E(T'_W)$ with a path isomorphic to $Q_{q-1}$.
	This corresponds to reversing some contractions made to obtain $T'_W$ from $T_W$; that is, $T''_W$ still remains a minor of $T_W$.
	To transform the isomorphism $\sigma'$ into the corresponding isomorphism $\sigma''$ from $T''_W$, we simply extend $\sigma'$ with the mapping of the inserted $uv$-path of $T''_W$ into the $\sigma'(u)\sigma'(v)$-path $Q_{q-1}$ in $G$.
	At this point, by the choice of the path $Q_{q - 1}$, the obtained isomorphism $\sigma''$ would satisfy the condition of \Cref{lemma:multi_saved_k_neighbours}. However, $T''_W$ is a minor of $T$, while we need a subgraph of $T$ in order to apply the lemma.

	We further extend the isomorphism $\sigma''$ into the isomorphism $\sigma$ from $T_W$ into $G$.
	To achieve that, we reverse all edge contractions that transform $T_W$ into $T''_W$ while simultaneously extending trivial paths in the image of the isomorphism.
	The only procedure for extension we use here is an insertion of a single vertex in $G$ between two adjacent vertices in the image of the tree.
	
	Formally, we proceed as follows.
	Put $T_0=T''_W$ and $\sigma_0=\sigma''$.
	Then, for each $i$ between $0$ and $t - 1$, where $t = |V(T_W)|-|V(T''_W)|$ is the number of vertex insertions we need to make, obtain a larger tree $T_{i + 1}$ together with an isomorphism $\sigma_{i + 1}$ from $T_{i + 1}$ to $G$ by extending the tree $T_i$ and the isomorphism $\sigma_i$. We always have that $|V(T_i)| = |V(T''_W)| + i$, and $T_i$ is a minor of $T_W$.
	
	In order to construct $T_{i+1}$ from $T_i$, note that since $i < t$, $T_i$ has at least one maximal trivial path that is shorter than its original counterpart in $T_W$.
	Denote this maximal trivial path by $P_i$.
	We want to find a vertex in $G-\Ima \sigma_i$ that has two consecutive neighbors on the path $\sigma_i(P_i)$.
	If such a vertex exists in $G$, denote it by $v_i$ and ``insert'' $v_i$ between its consecutive neighbors in $\sigma_i(P_i)$, with its preimage in $T_i$ being a new vertex inserted in the corresponding place in $P_i$. In formal terms, let $T_{i + 1}$ be obtained from $T_i$ by inserting a new vertex $x$ between $\sigma_i^{-1}(u)$ and $\sigma_i^{-1}(w)$, where $u$ and $w$ are the consecutive neighbors of $v_i$ on $\sigma_i(P_i)$, and 
	let $\sigma_{i+1}$ be obtained by extending $\sigma_i$ with mapping the new vertex $x$ to $v_i$.

    Assume now that on every iteration $i \in [0, t - 1]$ the suitable vertex $v_i$ exists, therefore we obtain the tree $T_t$ and the isomorphism $\sigma_t$ from $T_t$ to $G$. Clearly, $T_t = T_W$ as $T_t$ is a minor of $T_W$, and $|V(T_t)| = |V(T_W)|$; denote $\sigma = \sigma_t$.
	We apply \Cref{lemma:multi_saved_k_neighbours} to the tree $T_W$ with the isomorphism $\sigma$; by construction, $T_W$ contains the set of the leaf neighbors $W$, and sufficiently many of their non-neighbors are in $\Ima \sigma$. Thus by \Cref{lemma:multi_saved_k_neighbours}, there is also an isomorphism from $T$ to $G$.
	
	It is left to consider the case where no suitable vertex exists on the step $i \in [0, t - 1]$.
	Then $V(G)\setminus\Ima\sigma_i$ has no vertices with two consecutive neighbors on the path $\sigma_i(P_i)$.
	The length of $P_i$ is at least $2k$  by the construction of $T'_W$. We take any subpath of $\sigma_i(P_i)$ consisting of exactly $2k$ vertices and denote it by $R$.
	Each vertex in $V(G)\setminus\Ima\sigma_i$ has at least $k-1$ non-neighbors on $R$.
	Hence, $R$ is a $k$-preserving path for the graph $G-(\Ima \sigma_i\setminus V(R))$.
	
	We have 
	\begin{equation*}
	\begin{aligned}
		\delta(G-(\Ima \sigma_i\setminus V(R)))> &\delta(G)-|V(T_W)|= \delta(G)-|E(T_W)|+1\\> &\delta(G)-\diam(T_W)\cdot|W|+1\\\ge & 2k\cdot\diam(T)-\diam(T_W)\cdot|W|+1\\> &2|W|\cdot\diam(T_W)-\diam(T_W)\cdot|W|+1\\=&\diam(T_W)\cdot|W|+1\ge |V(T_W)|.
	\end{aligned}	
	\end{equation*}
	We then argue that $T_W$ admits an isomorphism into a subgraph of $G-(\Ima\sigma_i\setminus V(R))$.
	Consider any trivial path $D$ of $T_W$ on at least $2k$ vertices that does not intersect $W$, which exists, e.g., inside the long trivial path given by \Cref{claim:long_path}.
	Clearly, there exists an isomorphism from $D$ into $R$.
	We extend this isomorphism to an isomorphism $\sigma$ from $T_W$ into a subgraph of $G-(\Ima\sigma_i\setminus V(R))$ greedily by \Cref{prop:chvatal-generalized}; this is possible since by the inequality above, $\delta(G-(\Ima \sigma_i\setminus V(R))) > |V(T_W)|$.
	Since every vertex in $V(G) \setminus \Ima \sigma_i$ has at least $k - 1$ non-neighbors on $R$, so do the images of $W$ in $G$ under the newly constructed isomorphism. Therefore by \Cref{lemma:multi_saved_k_neighbours} applied to $T_W$ and $\sigma$, $G$ contains $T$ as a subgraph.
	This completes the proof of the lemma.
\end{proof}

\subsection{Escaping neighborhoods and separating $G$}\label{subsec:escape}

The second tool we use for preserving neighbors is the notion of an escape vertex in $G$.
In the following lemma, we require $T$ to have a vertex of a large enough degree.
We will map this vertex of $T$ to an escape vertex in $G$ and then use it for embedding non-neighbors into partial isomorphism in order to further pipeline with \Cref{lemma:multi_saved_k_neighbours}.
Recall that a vertex $u$ in a graph $G$ is a $q$-escape vertex if either $\deg_G(u) \ge \delta(G) + q$, or the maximum matching size between $N[u]$ and $V(G) \setminus N[u]$ is at least $q$ (\Cref{def:escape}).

\begin{lemma}\label{lemma:escape_vertex_and_high_degree_t}
	Let $G$ be a graph and $T$ be a tree on $\delta(G)+k$ vertices for $k\ge 2$.
	Let $q=2k^2\cdot\diam(T)$. 
	If $\delta(G)\ge q$, $\Delta(T)\ge k^2$,  $\ld(T)<k$, and $G$ has a $q$-escape vertex, then $G$ contains $T$ as a subgraph.
\end{lemma}
\begin{proof}
	Let $u$ be a $q$-escape vertex in $G$ and let $t$ be a maximum-degree vertex in $T$.
	To prove the lemma, we take $W$, the 
	set of neighbors of arbitrary $k-1$ leaves of $T$, and construct a partial isomorphism of $T$ that fits into \Cref{lemma:multi_saved_k_neighbours} with  $W$.
	
	We start with $T'$, the minimal  subtree of $T$ containing all vertices of $W\cup \{t\}$ in $T$. Thus  $|V(T')|\le \diam(T)\cdot(|W|+1)\le \diam(T)\cdot k< \delta(G)$.
	We initialize $\sigma$ to be an arbitrary isomorphism from $T'$ into a subgraph of $G$ that maps $t$ to $u$, the $q$-escape vertex.  That is,  $\sigma(t)=u$; such an isomorphism can be found greedily by \Cref{prop:chvatal-generalized}.
	
	Now we expand $\sigma$ (and $T'$ correspondingly) so it occupies at least $\ndef(\sigma(w))$ non-neighbors of $\sigma(w)$ for each $w\in W$.
	We expand $T'$ and $\sigma$ iteratively using the following two claims.
	
	\begin{claim}\label{claim:escape_w_neighborhood}
		If $|V(T')| \le \frac{q}{2}-k$ and $\ndef(\sigma(w))>0$, then there is a vertex $v \in V(G)$ such that $v \notin N_G(\sigma(w))\cup \Ima \sigma$ and the distance between $u$ and $v$ in $G-(\Ima\sigma\setminus\{u\})$ is at most $2$.
	\end{claim}
	\begin{claimproof}
		If there is $v \in N_G(u)\setminus(N_G(\sigma(w))\cup \Ima \sigma)$, then we are done.
	    Clearly this is the case if $\deg_G(u) \ge \delta(G) + q$, as $|N_G(\sigma(w))\cup \Ima \sigma| \le \delta(G) + k - 1 + \frac{q}{2} - k < \delta(G) + q$.
		
		Thus we may assume $N_G(u)\subseteq N_G(\sigma(w))\cup \Ima\sigma$, and there is a matching $M$ of size $q$ between $N_G[u]$ and its complement in $G$ by the definition of a $q$-escape vertex.
		Since $|N_G(u)|\ge \delta(G)$ and $|N_G(\sigma(w))|<\delta(G)+k$, we have that $u$ and $\sigma(w)$ should have at least $$\delta(G)-|\Ima\sigma|\ge \delta(G)-\left(\frac{q}{2}-k\right)>\deg_G(\sigma(w))-\frac{q}{2}$$ common neighbors.
		Hence, at most $\frac{q}{2}-1$ neighbors of $\sigma(w)$ lie outside of $N_G(u)$.
		
		Now, at least $q-|\Ima\sigma|=q-|V(T')|>\frac{q}{2}$ of the edges of $M$ do not have endpoints in $\Ima\sigma$.
		At least one of these edges has its endpoint outside of $N_G(\sigma(w))$, since at most $\frac{q}{2} - 1$ neighbors of $\sigma(w)$ lie outside of $N_G(u)$, and every edge in $M$ has one endpoint outside of $N_G(u)$.
		The distance between this endpoint and $u$ is two, completing the proof of the claim.
	\end{claimproof}
	
	Initially, $T'$ has at most $k$ leaves, and each expansion will add one leaf to $T'$.
	The following claim provides a vertex in $T$ to expand $T'$ with.
	
	\begin{claim}\label{claim:escape_free_vertex_in_t}
		If $T'$ has at most $k^2-k$ leaves, then there is a non-leaf neighbor $x\in N_T(t)$ of $t$ such that $x \notin V(T')$. 
	\end{claim}
	\begin{claimproof}
		Since $\ld(T)<k$, the vertex $t$ has more than $\deg_T(t)-k\ge k^2-k$ non-leaf neighbors in $T$. 
		Also $\deg_{T'}(t)\le k^2-k$, since $T'$ has at most $k^2-k$ leaves.
		Then at least one non-leaf neighbor of $t$ in $T$ is not in $V(T')$.
	\end{claimproof}
	
	As we noted above each expansion of $T'$ adds exactly one leaf to it, and at most $\sum{w \in W}\ndef(\sigma(w))\le (k-1)^2$ expansions are made in total.
	Hence, $T'$ has  less than $(k-1)^2+k=k^2-k+1$ leaves until the last possible expansion.
	It follows that $|V(T')|\le (k^2-k)\cdot \diam(T)\le \frac{q}{2}-k$.
	We get that two claims above can be applied to $T'$ at each iteration of the expansion.

\begin{figure}[ht]
\centering
\scalebox{0.7}{
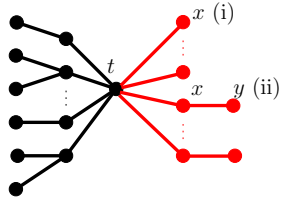}
\caption{The expansion of $T'$: the original vertices and edges of $T'$ are shown in black and the added vertices and edges are red.}
\label{fig:expT}
\end{figure}

	One iteration of the expansion process is as following (see also \Cref{fig:expT}). 
	Until $\sigma$ and $T'$ satisfy the conditions of \Cref{lemma:multi_saved_k_neighbours}, take a vertex $w \in W$ such that $\sigma(w)$ has less than $\ndef(\sigma(w))$ non-neighbors in $\Ima\sigma$.
	By \Cref{claim:escape_w_neighborhood}, there is a vertex $v \notin N_G(\sigma(w))\cup\Ima\sigma$  within distance at most $2$ from $u$, and the shortest path to this vertex does not go through $\Ima\sigma$.
	Take $x$ given by \Cref{claim:escape_free_vertex_in_t}, and take its neighbor $y\in N_T(x)$ such that $t\neq y$; it exists since $x$ is not a leaf in $T$.
	If $v$ is a neighbor of $u$, expand $T'$ with $x$, making it adjacent with $t$, and put $\sigma(x):=v$ (see case (i) in \Cref{fig:expT}).
	Otherwise, expand $T'$ with the path $t-x-y$ and map it to the shortest $(u,v)$-path in $G-\Ima\sigma$ (see case (ii) in \Cref{fig:expT}).
	
	Since expansion is always possible, we reach the situation when \Cref{lemma:multi_saved_k_neighbours} can be applied to $T'$ and $\sigma$. The proof of the lemma is thus complete. 
\end{proof}

We also have to deal with trees without vertices of large degree.
The basic idea of the following lemma is quite similar to the last one, but the mechanism of mapping extension is different.
Sometimes the extension is not possible; in this case, we obtain a small vertex separator of $G$. By small we mean that its size is significantly smaller than $\delta(G)$.

\begin{lemma}\label{lemma:escape_or_separator}
	Let $G$ be a graph and $T$ be a tree on $\delta(G)+k$ vertices for $k\ge 2$.
	If $\ld(T)<k$, $\Delta(T)< k^2$, and $\delta(G)\ge k^5\cdot\diam(T)$, then either
	\begin{itemize}
		\item $G$ contains $T$ as a subgraph, or
		\item there is a vertex separator of $G$ of size at most $2k\cdot(k-1)\cdot(\diam(T)+2)$.
	\end{itemize}
\end{lemma}
\begin{proof}
	The first part and general idea of the proof is similar to the previous lemma: we pick a set $W$ of neighbors of some $k - 1$ leaves, a spanning tree $T'$ of $W$ and a subgraph isomorphism $\sigma: V(T')\to V(G)$.
	Then we try to expand $\sigma$ and $T'$ to satisfy the non-neighbor condition on each $w \in W$.	

	However, the statement of the current lemma does not guarantee us a vertex of a large enough degree in $T$. Therefore, 
	instead of starting from a mapping of one single high-degree vertex and then extending this mapping to $T'$, we use a different strategy. We select distinct vertices for each single extension iteration.
	We call this set $U$, and we require its size to be $(k-1)^2\ge\sum\ndef(w)$.
	We first claim that the choice of   sets $U$ and $W$ consisting of distinct vertices always exists.
	
	To simplify further arguments, we consider $T$ a rooted tree.
	We pick as its root an arbitrary vertex $r\in V(T)$.
	For each $v \in V(T)$, by $T_v$ we denote the subtree of $T$ rooted in $v$.
	
	\begin{claim}\label{claim:setsSandW}
		There is a set $U\subset V(T)$ of size $(k-1)^2$ and a set $W\subset V(T)$ of size $k-1$ such that
		\begin{itemize}
			\item $U\cap W=\emptyset$;
			\item For each $s\in U$, the depth of tree $T_{s}$ is two;
			\item  For each $w \in W$, the depth of tree $T_{w}$ is one;
			\item Rooted subtrees of $T$ corresponding to vertices in $U\cup W$ are pairwise disjoint.
		\end{itemize}
	\end{claim}
	\begin{claimproof}
		In this proof, by a \emph{subtree} we mean a rooted subtree  $T_x$ for some $x\in V(T)$.
		Note that the depth of $T$ is at least $2$, since $|V(T)|>\delta(G)>\Delta(T)+1.$
		
		Assume first that  $T$ has at least $k^2-k$ subtrees of depth exactly $2$.
		All these trees are pairwise disjoint; we  pick $U$ as the roots of any $(k-1)^2$ of these subtrees.
		There are $k^2-k-(k-1)^2=k-1$ depth-$2$ subtrees left, pick an arbitrary depth-$1$ subtree from each one of them.
		The roots of these depth-$1$ trees form the set $W$, and, up to our assumption,  the claim follows.
		
		It is left to show that $T$ cannot have less than $k^2-k$ subtrees of depth exactly $2$. Targeting a contradiction, we 
		suppose that $T$ has less than $k^2-k$ such subtrees.
		
		There are three types of vertices in $T$: (a) roots of subtrees of depth two or greater; (b) roots of subtrees of depth one; (c) leaves.
		Consider the subtree $T_a$ of $T$ where all vertices of types (b) and (c) are removed. The tree $T_a$ thus consists of all type (a) vertices of $T$, and it has less than $k^2-k$ leaves, since the leaves are exactly the roots of depth-$2$ subtrees. 
		 By \Cref{lemma:leaves_diameter}, the number of vertices in $T_a$, which is equal to the number of  vertices of type (a) in $T$, is less than  $(k^2-k)\cdot \diam(T_a) \le (k^2-k)\cdot \diam(T)$.
		
		For vertices of type (b) in $T$, note that each vertex of type (b) is a neighbor to some vertex of type (a).
		So the number of these vertices is at most $(k^2-k)\cdot \diam(T)\cdot \Delta(T)$.
		Each leaf in $T$ (vertex of type (c)) has a vertex of type (a) or type (b) as its only neighbor, and each vertex can have at most $\ld(k)$ adjacent leaves.
		
		We obtain that   $$\delta(G)+k = |V(T)|\leq (k^2-k)\cdot\diam(T)\cdot (1+\Delta(T))\cdot (1+\ld(k))\le (k^2-k)\cdot k^2\cdot k \cdot \diam(T)<\delta(G),$$
		which is a contradiction.
	\end{claimproof}
	
\begin{figure}[ht]
\centering
\scalebox{0.7}{
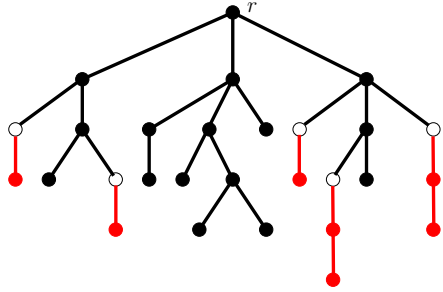}
\caption{The expansion of $T'$; the vertices of $U$ are shown is white and  the added vertices and edges are red.}
\label{fig:expTS}
\end{figure}	
	
Let $U$ and $W$ be the sets given by \Cref{claim:setsSandW}. We start constructing the subgraph isomorphism by mapping a minimal subtree $T'$  of $T$ containing
$U\cup W$.  
	Since $\delta(G)>|V(T')|$, by \Cref{chvatal-theorem}, we can construct an isomorphism $\sigma$ from $T'$ into a subgraph of $G$.
	Similarly to the proof of \Cref{lemma:escape_vertex_and_high_degree_t}, we expand the embedding of $T'$ in $G$ as follows. Until $T'$ and $\sigma$ satisfy the statement of \Cref{lemma:multi_saved_k_neighbours}, we consider an arbitrary $w\in W$ that has not enough non-neighbors in $\Ima\sigma$.
	If there is $u \in U$ such that there is a vertex $v\notin N_G(\sigma(w))$ within a  distance at most $2$ from $\sigma(u)$ in the graph $G-(\Ima\sigma\setminus\{\sigma(u)\})$, then extend $T'$ with either one or two vertices from $T_{u}$ and map the corresponding path into the $(\sigma(u),v)$-path in $G$ (see~\Cref{fig:expTS}).
	After that, remove $u$ from $U$.
	Note that $U$ is chosen in such a way that $T'$ can always be extended with either one or two vertices. 
	
	However, there is a problem that could prevent a successful iteration. It could happen that none of the $u \in U$ is suitable for saving neighbors of any $w \in W$ that requires more non-neighbors in $\Ima\sigma$.
	To dive into the details, suppose that we arrive at such a situation.
		Let $W'$ be the set of $w\in W$ that does not satisfy the condition of \Cref{lemma:multi_saved_k_neighbours}.
	The set $U$ consists of vertices that were not yet used in the extension.
	Since the extension is not possible, none of the neighbors of $u\in U$ suits any of $w \in W'$. That is, 
	$$A:=\left(\bigcup_{u\in U}N_G(\sigma(u))\setminus\Ima\sigma\right) \subseteq B:=\bigcap_{w\in W'}N_G(\sigma(w)).$$

	Since the extension is not possible, the neighbors of $A$ in $G$ are also not suitable for any $w \in W'$. Thus $$N_G(A)\setminus\Ima\sigma\subseteq B$$ holds as well.
	Equivalently, vertex set $S:=\Ima\sigma\cup (B\setminus A)$ separates $A$ from the rest of $G$.
	But we know that $|A|\ge \delta(G)-|\Ima\sigma|$ while $|B|<\delta(G)+k$.
	Hence, the size of the vertex separator $S$ is less then $2|\Ima\sigma|+k$.
	
	To estimate $|\Ima\sigma|=|V(T')|$, note that we start $T'$ from a spanning tree with $k^2-k$ leaves, and each iteration adds at most two vertices to $T'$.
	So $|\Ima\sigma|<(k^2-k)\cdot\diam(T)+1+2(k-1)^2$.
	Thus the size of the separator $S$ is at most 
	\begin{equation*}
	\begin{aligned}
		2|\Ima\sigma|+k\le &2(k^2-k)\cdot \diam(T)+4(k-1)^2+k\\
		=& 2(k-1)\cdot(k\cdot\diam(T)+2(k-1))+k\\
		\le & 2(k-1)\cdot(k\cdot\diam(T)+2(k-1)+1)\\
		< & 2(k-1)\cdot k\cdot(\diam(T)+2),
	\end{aligned}
	\end{equation*}
	which completes the proof of the lemma.
\end{proof}

The last result of this subsection shows a way to employ a small vertex separator of $G$. Recall that a tree $T$ is $q$-separable if there exists an edge whose removal separates $T$ into two subtrees each of size at least $q$ (\Cref{def:separable}).

\begin{lemma}\label{lemma:separator_and_separable_tree}
	Let $G$ be a connected graph and let $T$ be a tree on $\delta(G)+k$ vertices. 
	Let also $S$ be a vertex separator of $G$ such that $\delta(G)\ge 3|S|$.
	If $\delta(G)\ge 15k$, $\ld(T)<k$ and $T$ is \separable{$(|S|+k)$}, then $G$ contains $T$ as a subgraph.
\end{lemma}
\begin{proof}
	Without loss of generality, we assume that $S$ is an inclusion-wise minimal separator of $G$.

	The graph $G-S$ consists of at least two connected components.
	Let $A$ be the vertex set of one of them and let $B:=V(G-S)\setminus A$ be the vertex set of all other connected components in $G-S$.
	
	Take arbitrary $s \in S$.
	Since $S$ is minimal, both $G[A\cup \{s\}]$ and $G[B\cup \{s\}]$ are connected.
	In total, $s$ has at least $\delta(G)-|S|$ neighbors in $A$ and $B$.
	Without loss of generality, $s$ has at least $\frac{1}{2}(\delta(G)-|S|)$ neighbors in $A$.
	
	Now consider the edge of $T$ that separates it into two connected parts of size at least $|S|+k$.
	Denote the endpoints of this edge by $x$ and $y$.
	Denote the two parts of $T$ by $T_x$ and $T_y$, such that they contain $x$ or $y$ respectively.
	Without loss of generality, we assume that $\deg_T(x)\le \deg_T(y)$.
	
	\begin{claim}
		$\deg_T(x)\le \frac{\delta(G)}{3}$.
	\end{claim}
	\begin{claimproof}
		The proof is by contradiction.
		If $\deg_T(x)>\frac{\delta(G)}{3}$, then $\deg_T(y)>\frac{\delta(G)}{3}$ as well.
		Since $\ld(T)<k$, both $x$ and $y$ have more than $\delta(G)/3-k$ neighbors (other than $x$ and $y$) in $T$ that are not leaves.
		
		Hence, in $T-\{x,y\}$ there at least $2\delta(G)/3-2k$ connected components of size at least two.
		But then $|V(T)|-2\ge 4\delta(G)/3-4k\ge\delta(G)+k$, which is a contradiction. 
	\end{claimproof}

	Now we have that both $V(T_x)$ and $V(T_y)$ are of size at most $$|V(T)|-(|S|+k)=(\delta(G)+k)-(|S|+k)=\delta(G)-|S|.$$
	Since $\delta(G[B])\ge \delta(G)-|S|$, there is a subgraph isomorphism from $T_y$ into $G[B]$.
	Construct an arbitrary such isomorphism using \Cref{prop:chvatal-generalized} that maps $y$ to a neighbor of $s$ in $G[B]$.
	
Now we show how  to map $T_x$ into $G[A\cup \{s\}]$.
	While the minimum degree of $G[A\cup \{s\}]$ is not guaranteed to be at least $\delta(G)-|S|$,
	$s$ is the only vertex that could break this condition.
	We ensured that $$\deg_{G[A\cup\{s\}]}(s)\ge \frac{\delta(G)-|S|}{2}\ge \frac{\delta(G)}{3}\ge \deg_T(x),$$
	so it is possible to map $x$ and its neighbors in $T_x$ into $s$ and its neighbors in $G[A\cup\{s\}]$.
	Since each vertex in $G[A]$ is of degree at least $\delta(G)-|S|$, this partial mapping extends into full isomorphism of $T_x$ into $G[A\cup \{s\}]$ by \Cref{prop:chvatal-generalized}.
	
	To conclude,  we constructed an isomorphism from  $T_x$ to $G[A\cup \{s\}]$ and from $T_y$ to $G[B]$.
	Since $x$ is mapped into $s$ and $y$ is mapped into a neighbor of $s$, the union of these mappings is the required isomorphism from $T$ to $G$.
\end{proof}

\subsection{Proof of \Cref{thm:medium_diameter_escape_vertices}}

We combine the results of this section into the proof of its main result.

\begin{proof}[Proof of \Cref{thm:medium_diameter_escape_vertices}]
	The proof consists of considering several cases.
	
	\medskip\noindent
	\textsl{$T$ has large diameter}. Suppose that $\diam(T)\ge 2 k^{11}$.
	Since $\delta(G)\ge 3^{17}>2^{16}$, $\delta(G)>(\log\delta(G))^4$.
	
	Then, as $\diam(T)\le 8k^6\cdot\log\delta(G)$, 
	\begin{equation*}
	\begin{aligned}
	\delta(G)>& (\log\delta(G))^4>\frac{\diam(T)^4}{8^4\cdot k^{24}}>\diam(T)\cdot \frac{\diam(T)^3}{3^{8}\cdot k^{24}}\\\ge & \diam(T)\cdot\frac{8k^{33}}{k^{32}}\ge \diam(T)\cdot 8k>2k\cdot \diam(T).
	\end{aligned}
	\end{equation*}
	Then, by \Cref{lemma:large_diameter_trivial_paths}, $G$ contains $T$.
	For the remaining part of the proof, we assume that $\diam(T)<2k^{11}$.
	
	\medskip\noindent
	\textsl{$T$ is of small max-degree.}
	Suppose now that $\Delta(T)<k^2$ and apply \Cref{lemma:escape_or_separator} to $G$ and $T$.
	If $G$ contains $T$, then we are done.
	Otherwise, there is a vertex separator $S$ of $G$ of size at most $$(\diam(T)+2)\cdot 2k(k-1)< (2k^{11}+2)\cdot 2k^2<6k^{13}\le 2k^{14}.$$
	Note that $\delta(G)\ge 3|S|$.
	
	We now argue that there is an edge in $T$ separating it into sufficiently large parts.
	\begin{claim}
	    In a tree $T$ with $|V(T)| \ge 2$ there exists an edge separating it into two parts of size at least $\frac{|V(T)|-1}{\Delta(T)}$.
	    \label{claim:separable}
	\end{claim}
	\begin{proof}
	    If $\Delta(T) = 1$ the statement is trivial since then $|V(T)| = 2$, hence we assume $\Delta(T) \ge 2$.
    There exists a vertex $v \in V(T)$ such that each connected component of $T - v$ is of size at most $\frac{|V(T)| - 1}{2}$.
    Consider the largest component, its size is also at least $\frac{|V(T)| - 1}{\Delta(T)}$  since there are at most $\Delta(T)$ components.
    The edge from the largest component to $v$ is the desired edge, since the size of the remaining part is at least $\frac{|V(T)| + 1}{2} \ge \frac{|V(T)| - 1}{\Delta(T)}$, as $\Delta(T) \ge 2$.
	\end{proof}

	Note that $$\frac{|V(T)|-1}{\Delta(T)}>\frac{\delta(G)}{k^2}\ge k^{15}>|S|+k.$$
	That is, by \Cref{claim:separable} $T$ is \separable{$(|S|+k)$}.
	Therefore $G$ contains $T$ by \Cref{lemma:separator_and_separable_tree}.
	In what follows we assume that $\Delta(T)\ge k^2$.
	
	\medskip\noindent
	\textsl{$G$ has an escape vertex.}
	If $G$ has a $4k^{13}$-escape vertex, then \Cref{lemma:escape_vertex_and_high_degree_t} can be applied to $G$ and $T$, since $4k^{13}> 2k^2\cdot \diam(T)$.
	Then $G$ contains $T$.
	We further assume that $G$ has no $4k^{13}$-escape vertices.
	
	\medskip\noindent
	\textsl{$T$ is \separable{$2k^{14}$}.}
	It is left to consider the case when $T$ is \separable{$2k^{14}$}.
	Since $G$ has no $4k^{13}$-escape vertices and $n\ge \delta(G)+2k^{14}\ge \delta(G)+6k^{13}>(\delta(G)+k)+4k^{13}$,
	then $G$ has vertex separator $S$ of size at most $4k^{13}$.
	Indeed, consider a minimum-degree vertex $v \in V(G)$, $\deg_G(v) = \delta(G)$.
	Since $v$ is not a $4k^{13}$-escape vertex, the maximum size of a matching between $N_G[v]$ and $V(G) \setminus N_G[v]$ is less than $4k^{13}$.
	This means that the corresponding bipartite graph has a vertex cover $S$ of size less than $4k^{13}$.
    Clearly $S$ is a separator since both $N_G[v] \setminus S$ and $V(G) \setminus N_G[v] \setminus S$ are non-empty: the former since we may assume $v \notin S$ ($v$ is not adjacent to any vertex in $V(G) \setminus N_G[v]$), and the latter by $n > \delta(G) + k + 4k^{13}$.

	As $2k^{14}\ge 6k^{13}>4k^{13}+k$,
	apply \Cref{lemma:separator_and_separable_tree} to $G, T$ and $S$ and obtain that $G$ contains $T$.
	The proof is now complete.
\end{proof}

\section{Final proof: Putting it all together}\label{section:mainresult}

This section finalizes the proof of our main result by combining the previous sections' main theorems.
We restate the theorem here.

\maintheorem*

\begin{proof}
	Let $G$ be a non-empty graph and $T$ be a tree on exactly $\delta(G)+k$ vertices. We assume that $k\ge 2$  (if $k\leq 1$, by  
	Chv{\'{a}}tal's Lemma, $T$ is a subgraph of $G$). 
%
%
%
	We also assume that $|V(G)|\ge |V(T)|$, otherwise,  trivially,  $G$ does not contain $T$. 
	
	If $G$ and $T$  satisfy the conditions of \Cref{thm:densest_part} with $p:=15$, then by \Cref{thm:densest_part}, we can identify in time $2^{k^{\Oh(1)}}\cdot\polyn$ whether $T$ is a subgraph of $T$.  In the rest of the proof, we go through all the cases when $G$ and $T$ do not satisfy the conditions of  \Cref{thm:densest_part}.
	
	\medskip\noindent
	\textsl{Case 1: $\delta(G)<k^{3p+1}$}.
	In this case, we use the color coding of Alon et al., see \Cref{col:color_coding_subtree},  to decide whether  $G$ contains $T$.  The algorithm works in time $2^{\Oh(|V(T)|)}\cdot\polyn=2^{\Oh(\delta(G)+k)}\cdot\polyn=2^{k^{\Oh(p)}}\cdot\polyn$.
	
	\medskip\noindent
	\textsl{Case 2: $k < 3$}.
	Since $T$ consists of at least two vertices, its leaf-degree is at least one.
	Hence, the case $k<3$ is equivalent to $\ld(T)\ge k-1$, which we consider in the next case.
	
	\medskip\noindent
	\textsl{Case 3: $\ld(T)\ge k-1$}.	
	In this case, 
	we apply  \Cref{thm:large_leaf_degree} to decide whether $G$ contains  $T$.
	The running time of this algorithm is $2^{\Oh(k^2)}\cdot\polyn$.
	
	\medskip\noindent
	\textsl{Case 4: $|V(G)|\le (1+\frac{1}{4k})\cdot \delta(G)$}.
	We also assume that previous cases are not applicable. Because $\delta(G)\ge k^{46}>2k^{12}$ and $\ld(T)<k$, by 
  \Cref{lem:dense},   $G$  contains $T$.

	\medskip\noindent
	\textsl{Case 5: $\diam(T)\ge 8k^6\cdot\log\delta(G)$}.
	As before, we assume that the previous cases are not applicable.
	Then $k\ge 3$, $\delta(G)>k^{16}$, and $|V(G)|/\delta(G)\ge 1+\frac{1}{4k}=1+\frac{4}{16k}>1+\frac{4}{k^4}$.
	Then by \Cref{thm:large_diameter_many_vertices},  $G$ contains $T$.
	
	\medskip\noindent
	\textsl{Case 6: There is a $k^p$-escape vertex in $G$}.
	We want to apply \Cref{thm:medium_diameter_escape_vertices}. We assume 
	that the conditions of the prior cases do not apply.
	Then $|V(G)|>(1+\frac{1}{4k})\cdot \delta(G)\ge\delta(G)+\frac{k^{46}}{4k}>\delta(G)+2k^{14}.$
	As $k^p=k^{15}>4k^{13}$, we know that $G$ has a $4k^{13}$-escape vertex.
	Conditions $\diam(T)\le 8k^6\cdot\log\delta(G)$, $\ld(T)<k$ and $\delta(G)\ge k^{17}$  also hold because the previous cases are not applicable.
	Then by \Cref{thm:medium_diameter_escape_vertices}, in this case, $G$ contains $T$ as a subgraph.
	
	\medskip\noindent
	\textsl{Case 7: $T$ is \separable{$k^p$}}. In this case, we again use \Cref{thm:medium_diameter_escape_vertices} but with the different condition.
	Since $k^p=k^{15}>2k^{14}$, we have that $T$ is \separable{$2k^{14}$}.
	Then by  \Cref{thm:medium_diameter_escape_vertices},  $G$ contains $T$.
	
	We have shown that if $G$ and $T$ do not meet the conditions of  \Cref{thm:densest_part}, then we either able to resolve the subtree isomorphism in time \classFPT in $k$ or use the established combinatorial theorems to prove that $G$ should contain $T$. 
	This completes the proof.
\end{proof}

\section{Why the guarantee cannot be improved}
 In this section, we prove \Cref{thm:lower-bound}. The theorem shows that parameterization 
 of  \probSTI above $(1+\varepsilon)\delta(G)$ for any $\varepsilon>0$ makes the problem 
 
 the problem 
 
 For any $\varepsilon>0$, \probSTI is \classNP-complete when restricted to instances $(G,T)$  
with $|V(T)|\leq (1+\varepsilon)\delta(G)$.
 We restate it here.
 
  \lowerboundtheorem*

\begin{proof}
We reduce from the \textsc{$3$-Partition} problem. In this problem, we are given a set $A$ of size $m=3n$, a ``size'' function $s\colon A\rightarrow \mathbb{Z}_{\geq 0}$, and an integer $B>0$ such that $\frac{1}{4}B< s(a)<\frac{1}{2}B$ for every $a\in A$ and $\sum_{a\in A}s(a)=nB$, and the task is to decide whether there is a partition of $A$ into $n$ disjoint sets $S_1,\ldots,S_n$ such that for every $i\in[n]$, $\sum_{a\in S_i}s(a)=B$. This problem is well-known to be \classNP-complete in the strong sense~\cite{GareyJ79}. We remind that because of the constraint  $\frac{1}{4}B< s(a)<\frac{1}{2}B$, each set $S_i$ should contain three elements of $A$ whenever a partition of $A$ with the required property exists.  

\begin{figure}[ht]
\centering
\scalebox{0.7}{
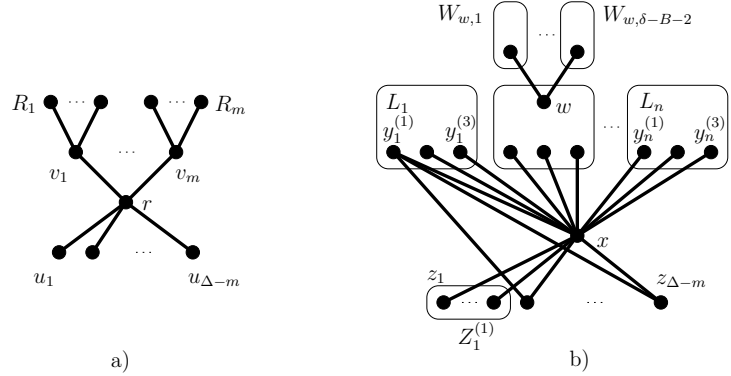}
\caption{Construction of $T$ (a) and $G$ (b). In (b), the edges between the vertices $y_i^{(h)}$ for $i\in [n]$ and $h\in[3]$ are not shown. Similarly, the edges between $z_1,\ldots,z_{\Delta-m}$ are not shown. Also, we do not show the edges between $y_i^{(h)}$ and $z_j$ for  $i\in [n]$, $h\in[3]$, and $j\in[\Delta-m]$, except the edges incident to $y_1^{(1)}$.}
\label{fig:lb}
\end{figure}

Consider an instance of \textsc{$3$-Partition} with $A=\{a_1,\ldots,a_m\}$ and set $\ell=\sum_{i=1}^ms(a_i)$. We define $\delta=\max\{\lceil\frac{\ell+3}{\varepsilon}\rceil,\ell+4n-2\}$ and $\Delta=\delta+2$. We construct the following tree $T$ (see \Cref{fig:lb} (a)).
\begin{itemize}
\item For each $i\in[m]$, construct a vertex $v_i$, a set $R_i$ of $s(a_i)$ vertices, and make $v_i$ adjacent to every vertex of $R_i$.
\item Construct a vertex $r$ and make it adjacent to the vertices $v_1,\ldots,v_m$.
\item Construct $\Delta-m$ vertices $u_1,\ldots,u_{\Delta-m}$ and make them adjacent to $r$.
\end{itemize}
Next, we construct the graph $G$ as follows  (see \Cref{fig:lb} (b)). 
\begin{itemize}
\item For every $i\in[n]$, construct a set $L_i$ of $B+3$ vertices and make it a clique, and then select three vertices $y_i^{(1)},y_i^{(2)},y_i^{(3)}\in L_i$.
\item Construct $\ell(\delta-B-2)$ copies $W_{w,i}$ of the complete graph $K_{\delta+1}$ for all $w\in \bigcup_{i=1}^n(L_i\setminus\{y_i^{(1)},y_i^{(2)},y_i^{(3)}\})$ and $i\in[\delta-B-2]$. 
\item For each $w\in \bigcup_{i=1}^n(L_i\setminus\{y_i^{(1)},y_i^{(2)},y_i^{(3)}\})$ and $i\in[\delta-B-2]$,
make $w$ adjacent to one vertex  of $W_{w,i}$.
\item For all $i,j\in[n]$ such that $i<j$, make $y_i^{(h)}$ adjacent to $y_j^{(1)},y_j^{(2)},y_j^{(3)}$ for each $h\in[3]$.
\item Construct  a vertex $x$ and make it adjacent to $y_i^{(1)},y_i^{(2)},y_i^{(3)}$ for all $i\in[n]$.
\item Construct $\Delta-m$ vertices $z_1,\ldots,z_{\Delta-m}$, make them pairwise adjacent and adjacent to $x$.
\item Find $m$ pairwise disjoint sets of vertices $Z_i^{(h)}\subseteq \{z_1,\ldots,z_{\Delta-m}\}$ of size $B+1$ for $i\in[n]$ and $h\in[3]$. For every $i\in[n]$ and $h\in[3]$, make $y_i^{(h)}$ adjacent to the vertices of $\{z_1,\ldots,z_{\Delta-m}\}\setminus Z_i^{(h)}$.
\end{itemize}
Notice that because $\Delta-m\geq \ell+n=(B+1)n$, disjoint sets $Z_i^{(h)}\subseteq \{z_1,\ldots,z_{\Delta-m}\}$ exist. Observe that for every $w\in \bigcup_{i=1}^n(L_i\setminus\{y_i^{(1)},y_i^{(2)},y_i^{(3)}\})$, $\deg_G(w)=\delta$,
for every vertex $v$ of each $W_{w,i}$, $\delta\leq \deg_G(v)\leq \delta+1$, for each $y_i^{(h)}$, $\deg_G(y_i^{(h)})=B+\Delta-(B+1)=\Delta-1=\delta+1$, for each $z_i$, $\Delta-1\leq \deg_G(z_i)\leq \Delta$, and $\deg_G(x)=\Delta$.
In particular, $\delta(G)=\delta$ and the maximum degree $\Delta(G)=\Delta$.

We show that $G$ contains $T$ as a subgraph if and only if  there is a partition of $A$ into $n$ disjoint sets $S_1,\ldots,S_n$ such that for every $i\in[n]$, $\sum_{a\in S_i}s(a)=B$.

Assume that there is a partition of $A$ into $n$ disjoint sets $S_1,\ldots,S_n$ such that for every $i\in[n]$, $\sum_{a\in S_i}s(a)=B$. We construct the isomorphism $\sigma$ mapping $T$ into a subgraph of $G$ as follows.
\begin{itemize}
\item Set $\sigma(r)=x$.
\item Set $\sigma(u_i)=z_i$ for all $i\in[\Delta-m]$.
\end{itemize}
Now we consider each $h\in[n]$. Assume that $S_h=\{a_i,a_j,a_k\}$. 
\begin{itemize}
\item We set $\sigma(v_i)=y_h^{(1)}$,  $\sigma(v_j)=y_h^{(2)}$, and  $\sigma(v_k)=y_h^{(3)}$.
\item Finally, we map the vertices of $R_i\cup R_j\cup R_k$ into distinct vertices of $L_i\setminus\{y_i^{(1)},y_i^{(2)},y_i^{(3)}\}$ using the fact that 
$|R_i|+|R_j|+|R_k|=\sum_{a\in S_h}s(a)=B=|L_i\setminus\{y_i^{(1)},y_i^{(2)},y_i^{(3)}\}|$.
\end{itemize}
The construction of $T$ and $G$ implies that $\sigma$ is a subgraph isomorphism of $T$ to $G$.

For the opposite direction, assume $T$ is a subgraph of $G$, that is, there is a subgraph isomorphism  $\sigma$ mapping 
$T$ into a subgraph of $G$. We have that $\deg_T(r)=\Delta$. From the other side, $\deg_G(x)=\Delta$ and only other vertices of degree $\Delta$ are the vertices of  $\{z_1,\ldots,z_{\Delta-m}\}\setminus \bigcup_{i=1}^n\bigcup_{h=1}^3Z_i^{(h)}$. However, the latter vertices are true twins with $x$. Therefore, we can assume without loss of generality that $\sigma(r)=x$. Because $\deg_T(r)=\deg_G(x)=\Delta$, $\sigma$ bijectively maps $N_{T}(r)$ to $N_G(x)$. If the vertex $v_i$ for some $i\in[m]$ is mapped to some $z_j$ for $j\in[\Delta-m]$, then the vertices of $R_i$ should be mapped to vertices of $N_G(x)$ because $N_G(z_j)\subseteq N_G[x]$. However, this is impossible because 
$\sigma^{-1}(N_G(x))=N_T(r)$. This implies that $\sigma(\{u_1,\ldots,u_{\Delta-m}\})=\{z_1,\ldots,z_{\Delta-m}\}$. Thus,  $\sigma(\{v_1,\ldots,v_m\})=\bigcup_{h=1}^n\{y_h^{(1)},y_h^{(2)},y_h^{(3)}\}$ and the sets 
$S_h'=\sigma^{-1}(\{y_h^{(1)},y_h^{(2)},y_h^{(3)}\})$ for $h\in[n]$ form a partition of $\{v_1,\ldots,v_m\}$.  For each $h\in[n]$, we define $S_h=\{a_i\in A\mid v_i\in S_h'\}$. The sets $S_1,\ldots,S_n$ form a partition of $A$. We claim that $\sum_{a\in S_h}s(a)\leq B$ for each $h\in[n]$. To see  this, consider some $S_h$ and assume that $\sigma^{-1}(\{y_h^{(1)},y_h^{(2)},y_h^{(3)}\})=\{v_i,v_j,v_k\}$ for distinct $i,j,k\in[m]$. We have that the vertices of $R_i$, $R_j$, and $R_k$ are mapped by $\sigma$ to distinct ertices of $L_h\setminus \{y_h^{(1)},y_h^{(2)},y_h^{(3)}\}$. Because $|L_h\setminus \{y_h^{(1)},y_h^{(2)},y_h^{(3)}\}|=B$, we have that
$s(a_i)+s(a_j)+s(a_k)=|R_i|+|R_j|+|R_k|\leq B$. This implies that $\sum_{a\in S_h}s(a)\leq B$. Since the inequality holds for every $h\in[n]$ and $nB=\sum_{a\in A}s(a)$, we obtain that $\sum_{a\in S_h}s(a)= B$ for every $h\in[n]$. 

To complete the proof, notice that $|V(T)|=1+m+\sum_{i=1}^ms(a_i)+\Delta-m=\delta+3+\sum_{i=1}^ms(a_i)=\delta+(3+\ell)$. Because $\delta\geq \frac{\ell+3}{\varepsilon}$, we have that $|V(T)|\leq (1+\varepsilon)\delta(G)$.
\end{proof}

We remark that \Cref{thm:lower-bound} is proved for constant $\varepsilon$ but the proof works even if $\varepsilon=\frac{1}{n^c}$ for any $c>0$ where $n$ is the number of vertices of the input graph. 



\section{Conclusion}\label{sec:concl}

In our exploration of algorithmic extensions of classical combinatorial theorems, we have demonstrated that it is possible to determine, in time $2^{k^{\Oh(1)}} \cdot n^{\Oh(1)}$, whether a graph $G$ contains a tree $T$ with at most $\delta(G)+k$ vertices as a subgraph. Our algorithm is a one-sided error Monte Carlo algorithm. This naturally raises two questions. First, can we develop a deterministic algorithm for this problem? Second, is there room for improvement in the running time? Can the problem be solved in time $2^{\Oh(k\log k)} \cdot n^{\Oh(1)}$ 
 or even in time $2^{\Oh(k)} \cdot n^{\Oh(1)}$?

Another question related to our work. Brandt, in his work~\cite{Brandt94}, extended Chvátal's Lemma  for forests.
\begin{proposition}[\cite{Brandt94}]\label{prop:forest}
Let $G$ be a graph, and $F$ be a forest such that $|V(F)| \leq |V(G)|$ and $|E(F)| \leq \delta(G)$. Then, $G$ contains $F$ as a subgraph.
\end{proposition}
Is the \textsc{Forest Containment} problem  (for a given graph $G$ and forest $F$, to decide whether $G$ contains $F$) FPT when parameterized by $k=|E(F)|-\delta(G)$?

%
%
%






\bibliographystyle{alpha}
\bibliography{ref}

\end{document}